\documentclass[11pt]{article} % required 11pt
\usepackage[margin=1in]{geometry} % required 1in margin at four sides
\usepackage{setspace}
\usepackage[autostyle, english = american]{csquotes}

\usepackage[normalem]{ulem} % for strike out text
\usepackage{hyperref}
\usepackage[english]{babel}

\usepackage{amsthm}
\usepackage{amsfonts}
\usepackage{newtxmath}
\usepackage{amsmath}

\newtheorem{definition}{Definition}
\newtheorem{proposition}{Proposition}
\newtheorem{theorem}{Theorem}
\newtheorem{lemma}{Lemma}
\DeclareMathOperator*{\argmin}{arg\,min}

\usepackage{hyperref} 
\hypersetup{
	colorlinks,
	linkcolor={red!50!black}, 
	citecolor={blue!50!black},
	urlcolor={blue!80!black}
}
\usepackage{mathrsfs}  
\usepackage{tikz}
\usetikzlibrary{shadows} 
\usetikzlibrary{chains,fit,shapes}
\usetikzlibrary{positioning}

\usepackage{xcolor}
\usepackage{xspace}
\usepackage{booktabs}
\usepackage{subcaption}
\usepackage{caption}
\usepackage{comment}
\usepackage{multirow}
\usepackage{bbm}
\usepackage{rotating}
\usepackage{pdflscape}

\usepackage[noend]{algpseudocode}
\usepackage{algorithm}

\usepackage{endnotes}
\let\footnote=\endnote

\usepackage{natbib}
 \bibpunct[, ]{(}{)}{,}{a}{}{,}%
\newcommand{\fileset}{\mathscr{F}}
\newcommand{\numfiles}{\ensuremath{n}}	
\newcommand{\tapelength}{\ensuremath{L}}
\newcommand{\file}{\ensuremath{i}}
\newcommand{\filep}{\ensuremath{j}} 
\newcommand{\leftfile}[1]{\ell_{#1}}
\newcommand{\rightfile}[1]{r_{#1}}
\newcommand{\filesize}[1]{s_{#1}}
\newcommand{\fileweight}[1]{w_{#1}}

\newcommand{\numrequests}{\ensuremath{m}}
\newcommand{\ordertuple}{\boldsymbol{\pi}}
\newcommand{\order}{\pi}
\newcommand{\pordertuple}{\boldsymbol{\rho}}
\newcommand{\porder}{\rho}

\newcommand{\dist}[2]{d_{#1,#2}}

\newcommand{\pval}{R}
\newcommand{\policyval}[2]{\pval_{#1}(#2)}
\newcommand{\optpolicyval}{\pval^*}

\newcommand{\rewindphase}{\mathcal{R}}
\newcommand{\forwardphase}{\mathcal{F}}
\newcommand{\response}{R}
\newcommand{\optordertuple}{\ordertuple^*}
\newcommand{\fifo}{\texttt{FIFO}}
\newcommand{\fififi}{\texttt{FIFF}}
\newcommand{\fifila}{\texttt{FILA}} 
\newcommand{\lafila}{\texttt{LFL}}
\newcommand{\zigzag}{\texttt{ARI}} 
 
\newcommand{\exact}{\texttt{DP}} 
\newcommand{\ssf}{\texttt{SSF}}

\newcommand{\neigh}[1]{\eta(#1)}
\newcommand{\neighs}[2]{\eta_{#1}(#2)}

\newcommand{\bigo}{\mathcal{O}}

\newcommand{\rewindrec}{\mathfrak{R}}
\newcommand{\forwardrec}{\mathfrak{F}}

\newcommand{\block}[2]{\mathscr{B}_{#1,#2}}

\newcommand{\blocksize}{b}

\newcommand{\optblock}[2]{\mathcal{V} \left( #1, #2 \right) }

\newcommand{\pendreqs}[2]{k'_{#1,#2}}

\newcommand{\fileweightRV}[1]{W_{#1}}
\newcommand{\prob}[1]{p_{#1}}

\newcommand{\expect}{\mathbb{E}}
\newcommand{\probleft}{\rho} 
\newcommand{\maxprob}{p^{\max}} 
\newcommand{\minprob}{p^{\min}} 

\newcommand{\stocrewindrec}{\mathfrak{R}^s}
\newcommand{\stocforwardrec}{\mathfrak{F}^s}
\newcommand{\stocpendreqs}[2]{\probleft'_{#1,#2}}
\newcommand{\probapprox}[2]{\prob{#2}^{#1}}

\begin{document}
%%%%%%%%%%%%%%%%

%\title{On Exact and Approximate Scheduling Policies for Linear Storage Devices in Data Centers}
\title{On Approximate Sequencing Policies for Linear Storage Devices}

\author{Carlos Cardonha\thanks{Corresponding author} \\ \small Department of Operations and Information Management, School of Business, University of Connecticut
	\and
	Andre A. Cire\\ \small Dept. of Management, University of Toronto Scarborough \& Rotman School of Management
	\and
	Lucas C. Villa Real \\ \small IBM Research, Brazil
}
\date{}
\maketitle

\abstract{%
This paper investigates sequencing policies for file reading requests in linear storage devices, such as magnetic tapes. Tapes are the technology of choice for long-term storage in data centers due to their low cost and reliability. However, their physical structure imposes challenges to data retrieval operations reflected in classic optimization and operations research problems. In this work, we provide a theoretical and numerical performance analysis of low-complexity algorithms under deterministic, stochastic, and online settings, which are key in practice due to their interpretability and the large scale of existing data services. In the deterministic setting, we show that traditional policies, such as first-in first-out (FIFO), have arbitrarily poor performance, and we develop and investigate new constant-factor approximations. For the stochastic setting, we present a fully polynomial-time approximation scheme that weighs files based on their access frequencies. Finally, we investigate an online extension and propose a new algorithm with constant competitive-factor guarantees. Our numerical analysis on synthetic and real-world data suggest that the proposed algorithms may significantly outperform policies currently adopted in practice with respect to average reading times.
}%

%\KEYWORDS{deterministic inventory theory; infinite linear programming duality;
%  existence of optimal policies; semi-Markov decision process; cyclic schedule}

%\KEYWORDS{Production-Scheduling: sequencing; Production-Scheduling: Approximations; Dynamic programming, Deterministic; Online Algorithms; Inventory-production : Stochastic}
%\HISTORY{First submitted in \textbf{X}.}

\maketitle
%%%%%%%%%%%%%%%%%%%%%%%%%%%%%%%%%%%%%%%%%%%%%%%%%%%%%%%%%%%%%%%%%%%%%%

% Samples of sectioning (and labeling) in OPRE
% NOTE: (1) \section and \subsection do NOT end with a period
%       (2) \subsubsection and lower need end punctuation
%       (3) capitalization is as shown (title style).
%
%\section{Introduction.}\label{intro} %%1.
%\subsection{Duality and the Classical EOQ Problem.}\label{class-EOQ} %% 1.1.
%\subsection{Outline.}\label{outline1} %% 1.2.
%\subsubsection{Cyclic Schedules for the General Deterministic SMDP.}
%  \label{cyclic-schedules} %% 1.2.1
%\section{Problem Description.}\label{problemdescription} %% 2.

% Text of your paper here

%%%%%%%%%%%%%%%%%%%%%%%%%%%%%%%%%%%%%%%%%%%%%%%%%%%%%%%%%%%%
% INTRODUCTION
%%%%%%%%%%%%%%%%%%%%%%%%%%%%%%%%%%%%%%%%%%%%%%%%%%%%%%%%%%%%

\section{Introduction} 
\label{sec:intro}

Scheduling methodologies play a pivotal role in the operations of data storage solutions, an industry that is forecasted to reach a valuation of US \$35 billion in 2022 \citep{idcTrends2020}. Data storage solutions offer a portfolio of random-access technologies for quick and recurrent file access, such as non-volatile memories and solid-state drivers, as well as cold-storage technologies where data are preserved for less frequent but longer-term access.  Due to the critical relevance of both technologies to data processing, the effective storage and retrieval of files have led to an extensive number of classical problems in the operations research and optimization literature (see, e.g., \S A.4 in \citealt{GareyJ79}, \citealt{parker1994creed}, and \citealt{Schaeffer11}).

In this work, we investigate file retrieval policies for linear storage devices, or more precisely, magnetic tapes. Tapes are the standard choice in cold storage due to their compelling cost, space, and security advantages \citep{lantz2018future}. For example, due to increase in cybersecurity threats and demand for high-volume storage, shipments for tape-based storage media increased by 40\% in 2021 \citep{tapes2022}. Tapes are extensively used by the entertainment industry~\citep{Coughlin2019}, oil companies \citep{Wittenborn2007}, space mission analysis \citep{nasaHECC},
and for multi-year research data for projects such as the Large Hadron Collider \citep{Cavalli10}. Standards for magnetic tapes are open and established by the Linear Tape-Open (LTO) consortium~\citep{ultriumlto}, composed of Hewlett Packard Enterprise, Quantum, and IBM, the latter a partner in this research study.

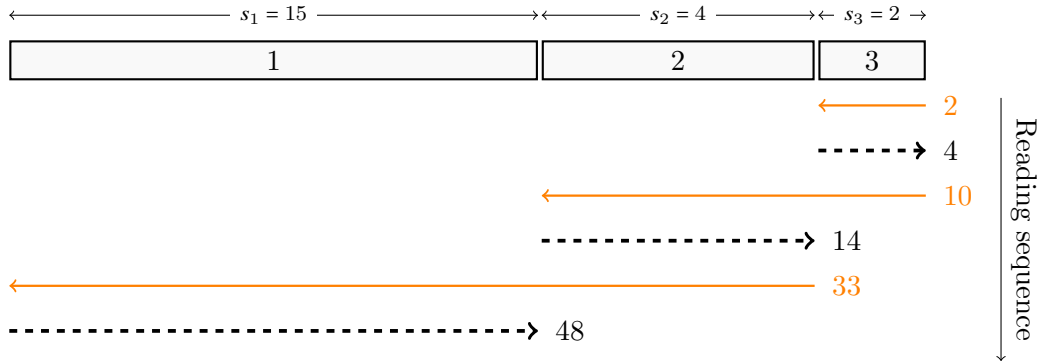
\begin{figure}[h!]
	\centering
	\usetikzlibrary{arrows,backgrounds,snakes}
	\begin{tikzpicture}
		[
			rewindfile/.style={rectangle, draw=black, fill=black!2, thick, minimum size=5mm},
			forwardfile/.style={rectangle, draw=purple, fill=purple!2, thick, minimum size=5mm, text=purple},
			file/.style={rectangle, draw=black, fill=gray!5, thick, minimum size=5mm, text=black},
		]
		
		% Files
		\node[file, minimum width=70mm] (f1) {1};
		\node[file, minimum width=36mm, right=0.5mm of f1, node distance=50pt] (f2) {2};
		\node[file, minimum width=14mm, right=0.5mm of f2, node distance=50pt] (f3) {3};
		
		% Empty file - use to align the reading time label of each file
		\node[right=5mm of f3] (nofile) {};
		
		% File sizes
		\newcommand\drawsize[2] {
			\coordinate[above = 6mm of #1.west] (c1a#1);
			\coordinate[above = 6mm of #1.east] (c1b#1);
			\draw[<->] (c1a#1) -- (c1b#1) node[midway, fill=white, font=\scriptsize] {#2};
		}
		
 		\drawsize{f1}{$\filesize{1} = 15$};
 		\drawsize{f2}{$\filesize{2} = 4$};
 		\drawsize{f3}{$\filesize{3} = 2$}; 

		% Sequence
		\newcommand\drawsequence[6] {
			\coordinate[below = #4mm of #1.east] (d1a#1#2);
			\coordinate[below = #4mm of #2.west] (d1b#1#2);
			\draw[->, line width=0.3mm,color=orange] (d1a#1#2) -- (d1b#1#2);
			%\node[font=\normalsize, color=blue, right=#3mm of d1a#1#2.west] {#6};
			%\node[font=\normalsize, color=blue, below=#4mm of nofile, yshift=4mm] {#6};
			\node[font=\normalsize, color=orange, right=1mm of d1a#1#2.west] {#6};
			\coordinate[below = #5mm of #2.west] (d2a#1#2);
			\coordinate[below = #5mm of #2.east] (d2b#1#2);
			\draw[->, dashed, line width=0.5mm] (d2a#1#2) -- (d2b#1#2);
			%\node[font=\normalsize, right=#3mm of d2b#1#2.west] {#3};
			%\node[font=\normalsize, below=#5mm of nofile, yshift=4mm] {#3};
		    \node[font=\normalsize, right=1mm of d2b#1#2.west] {#3};
		}
 		\drawsequence{f3}{f3}{4}{6}{12}{2};
  		\drawsequence{f3}{f2}{14}{18}{24}{10};
  		\drawsequence{f2}{f1}{48}{30}{36}{33};

		% Reading sequence arrow
		\coordinate[below = 5mm of f3.east, xshift=+10mm] (p1);
		\coordinate[below = 40mm of f3.east, xshift=+10mm] (p2);
		\draw[->] (p1) -- (p2) node[midway,above,rotate=270] {Reading sequence};

	\end{tikzpicture}
	\caption{Example of tape movement when reading files $3$, $2$, and $1$, in this order. Reading operations start at the rightmost file. Solid (colored) and dashed (black) arcs depict repositioning and reading operations, respectively. Labels represent the distance traversed when the movement is finished.
	}
	\label{fig:exampleSchedule1}
\end{figure}

The main characteristic of tapes is that files are distributed sequentially and contiguously throughout their tracks, leading to difficult operational aspects that share a close relationship with fundamental single-resource scheduling problems \citep{pinedo2016scheduling}. Figure \ref{fig:exampleSchedule1} illustrates a tape structure with three files with sizes $15$, $4$, and $2$ bits. Files have a left-to-right orientation, in that a retrieval operation must read the files linearly from their leftmost bit to their rightmost bit. Restrictions on the tape head also specify that only one file can be processed at a time. Thus, when reading a file, the tape head must first reposition itself to the left of the file, read it by traversing its bits, and again reposition itself to the left of the next file in the reading sequence. The figure depicts the retrieval sequence (3, 2, 1) and the total distance traversed for each repositioning and reading operation, assuming that the tape starts at the end of the rightmost file.

Our work investigates the retrieval problem faced by data centers considering offline and online settings. Specifically, in each planning period the data center receives a service request from customers or automated systems describing a set of files to be recovered from a tape. The objective is to define a sequence of file reads to minimize the total (or average) \textit{response time}, i.e., the time elapsed until each file starts to be read for the first time. This is a standard measure, as it captures quality of service and technical considerations \citep{hillyer1996random}. For instance, in Figure \ref{fig:exampleSchedule1} the response time is $2+10+33=45$ for a tape speed of one time unit per bit. 

The challenge in practical storage systems is scalability and interpretability. The deterministic problem variant (i.e., the files to retrieve are known in advance) is tractable with a time complexity that is $\bigo(n^4)$ in the number of files $n$. However, tapes often include thousands of files, and in standard guidelines no more than a few seconds is alloted to file-sequencing algorithms; for example, to enable batch analytics directly on archived data \citep{analyticsOnTape2014}. Tape hardware is also limited to less expensive, slower processors that are more energy efficient given the scale of operations. Further, data centers have a strong preference for low-complexity and understandable algorithms that can be justified, as small policy changes have a propagating effect in file retrieval times due to the large number of tapes and files in a data center. Not surprisingly, the standard tape file system employs first-in first-out (\fifo{}) heuristic strategies \citep{ISO20919}, processing files in the order that they are received. 

In view of these properties, this work provides a formal study of efficient but interpretable sequencing policies considering deterministic, stochastic, and online service request variants. In the deterministic setting, we consider a cold-storage environment where batches of read requests are uncorrelated but possibly few and far between; thus, the reading sequence can be processed offline. We investigate the theoretical performance of \fifo{} and approximation algorithms based on intuitive reverse- and forward-reading directions. We also show that this perspective reveals an alternative polynomial-time exact approach based on dynamic programming (DP).

The stochastic setting assumes that service requests are frequent and that a probability distribution over file requests is known. Our objective is to leverage distributional information to predetermine a reading sequence to be applied to all requests, which is in line with the previous interpretable policies. We show that, in this setting, the proposed DP for the exact deterministic approach can be adapted to a fully polynomial-time approximation scheme for the stochastic problem. Of particular importance is that the underlying DP only needs to be solved once, and the resulting sequence can be reapplied if the distributions remain unchanged.

Finally, we investigate an online setting incorporating real-time service requests when no distribution information is available. We first show that~\fifo{} has an arbitrarily poor performance in worst-case but common scenarios. We then develop a constant-factor competitive policy by iteratively alternating linear forward and backward reading operations. We also show that no competitive factor exists for an alternative objective criterion that considers file release times. 

To evaluate the policies, we provide a numerical study using artificial and real tape systems, the latter used for satellite imaging managed by an industry partner. The results suggest that the ordering of policies with respect to their theoretical performance is also observed empirically. In particular, our best proposed algorithm with constant-factor guarantee is often optimal and of sufficiently low complexity to be adopted in practice.
% simple greedy algorithms adopted in practice have better performance than expected according to their worst-case theoretical guarantees. In particular, one  strategy proposed in this work is fast enough to be adopted in practice and often delivers  optimal solutions.
%, thus far exceeding worst-case theoretical guarantees. %However, numerical insights indicate that the most appropriate policy is driven by the structure of the service requests and allowed time complexity, more precisely the variance of file sizes and the number of files. 
Moreover, if distributional information is available, the sequence derived from the approximation scheme outperforms approximate policies, especially if the request probabilities are relatively high.
%significantly outperformed policies approximately by 40\%.

\smallskip
\textit{Summary of contributions.} The primary contribution of this work is the formal study of approximate but low-complexity and interpretable policies for sequencing tasks in a linear storage device, considering both deterministic and uncertain cases. This adds to the existing literature that has focused on offline heuristics \citep{Zhang06,Schaeffer11}. In the deterministic setting, we identify worst-case performance ratios as well as cases where the approximations are optimal. In the stochastic setting, we adapt a polynomial-time dynamic program to develop a fully polynomial approximation scheme for use when access frequencies are available. For the online case, we develop constant-factor approximations and inapproximability of certain objective criteria. Finally, we analyze the numerical performance of the policies and draw insights on their empirical approximation ratio based on the instance structure, such as tape size and file size variance.

\smallskip 
\textit{Paper structure.} The paper is organized as follows. \S\ref{sec:relatedwork} reviews existing work and discusses the relationship between linear tape storage and existing scheduling problems. \S\ref{sec:problemdescription} formalizes the problem, and \S\ref{sec:deterministic} investigates the deterministic setting. \S\ref{sec:stochastic} studies the stochastic setting and introduces the fully polynomial time approximation scheme. \S\ref{sec:online} discusses the online approach and associated competitive ratios. \S\ref{sec:numerical} presents our numerical study. Finally, \S\ref{sec:conclusions} concludes and discusses future work. Proofs omitted in the main text are included in the electronic appendix.

\section{Related Work}
\label{sec:relatedwork}

Scheduling for sequential storage systems is a classical research stream in both the operations and computer science literature. To the best of our knowledge, \cite{Day1965letter} is the first to formalize a request problem over a multiple-file data storage system, proposing an integer programming (IP) formulation and heuristics considering potentially overlapping files. We note that model-based approaches such as IP typically do not scale due to the strict solution time adopted in industry and the limited processing power of tape hardware.
%, which has received limited attention in the current literature.

Early work has primarily focused on storage design, i.e., how to arrange files in a storage medium to minimize the long-run average response time. The study by \cite{cody1976record} was pioneering in this field, demonstrating that assigning files of the same size to different storage sectors is NP-hard if the file access distribution is known; other related problems are presented in Section A.4 of~\cite{GareyJ79}.

The simplest variant with no uncertainty or side constraint in tape design is the so-called \textit{optimal storage in tapes}, which can be efficiently solved by a greedy algorithm that sorts file in ascending size order \citep{parker1994creed}. It has since become a seminal example of greedy algorithms and fostered related fields, such as greedoid theory \citep{korte2012greedoids}. In this paper, we focus on the operational aspect of establishing file read orderings, assuming that (i) files have already been placed on the tape and (ii) access distributions may change. Further, we show that greedy approaches are suboptimal and investigate their theoretical performance ratios. 

Recent literature has focused on heuristic policies for tapes with varying organizational structures. Most notably, \cite{Hillyer96} and \cite{Sandsta99} investigated the problem of estimating tape velocity in serpentine tapes, proposing simple scheduling heuristics that orders files based on their physical position in the tape. \cite{hillyer1996random} developed a greedy heuristic based on an assymetric traveling salesperson problem (TSP) reduction for tertiary storage systems, simulating its performance against a pre-defined weave ordering for a serpentine tape. Similar sorting or greedy heuristics are found in other file systems, e.g., as discussed in \cite{Zhang06} and \cite{Schaeffer11}. 
% In recent practice, enterprise drives support Recommended Access Order (RAO) tape recalls; given a list of tape block addresses to visit, RAO outputs an ordered list of addresses that reduces the physical tape movement and wear \citep{Moraru2017CERN}. To date, the algorithms that enable RAO have not been disclosed and are supported by only a few specific combinations of tape drives and cartridges.

Our problem is closely related to the traveling repairman problem (TRP) \citep{fischetti1993delivery,coene2011charlemagne,pinedo2016scheduling}. Particularly relevant to our work is the line-TRP, where vertices are distributed on a straight line. While the TRP is NP-complete in general \citep{AfratiCPPP86,simchi1991minimizing}, the line-TRP can be solved in polynomial time if the processing times are zero \citep{Bock15,PsaraftisSMK90}. The complexity of the line-TRP with general processing times is still unknown. A related problem is the dial-a-ride on the line (line-DR), which transports products between pairs of vertices using one or more capacitated vehicles. \cite{PaepeLSSS04} present a classification of dial-a-ride problems, noting that minimizing the total completion times for the line-DR is NP-hard even if vehicles have capacity one. We discuss the formal relationship between our problem and the TRP in \S\ref{sec:connection}.

Online variations of the problems above are also related to sequential storage problems. One example is the online TSP on the line, in which new vertices to visit appear during a tour \citep{jaillet2006online}. \cite{bjelde2017tight} present a 1.64-competitive algorithm for the online TSP on the line, which is the best-possible competitive factor for the problem~\citep{AusielloFLST01}.  

For linear storage systems, \cite{honore2022exact} showed that the deterministic problem is solvable with time complexity that is quartic in the number of files on the tape, a result which we also obtain independently from an alternative perspective. Moreover, in a preliminary version of this work,  \citet{cardonha2016online}  proposed heuristic strategies to minimize flow time for interleaved read and write operations. In contrast to this earlier work, we consider the more realistic cold-storage setting where the tape is locked (no writes are possible), incorporate the response-time criteria, and provide a study of the theoretical worst-case performance of approximate methodologies for deterministic and uncertain variants of the problem. 

% Our stochastic variant can be cast within the framework of a priori optimization. The framework considers a fixed solution that adapts to the instance as uncertainty is revealed. The classical example is the a priori traveling salesperson problem \citep{bertsimas1990priori,shmoys2008constant}, which fixes the sequence of cities to visit and ``skips'' cities that are not realized according to a Bernoulli process. Similarly, our stochastic approach is a two-stage problem that pre-determines the reading sequence and ignores accesses that are not present in a service request. In particular, \cite{bertsimas1990priori} shows that a priori extensions of polynomially solvable problems, such as the identification of a minimum spanning tree, become difficult in the a priori setting. 

\section{The Linear Tape Sequencing Problem}
\label{sec:problemdescription}

In this section, we formalize in \S \ref{sec:description} the base linear tape sequencing model to be investigated in this work. Next, in~\S\ref{sec:connection} we describe its connection to two classical problems in the scheduling and routing literature. Finally, in~\S\ref{sec:practicalConsiderations} we discuss  the practical assumptions underlying the model.

\subsection{Problem Description}
\label{sec:description}

A tape is a set of $\numfiles$ files $\fileset \equiv \{1,\dots, \numfiles\}$ distributed sequentially and contiguously on a line discretized by bit units. The files have a left-to-right storage orientation, i.e.,  each file $\file \in \fileset$ begins in its left-bit position $\leftfile{\file}$ and has a size of $\filesize{\file}$ bits. The right-bit position of file $\file$ is the first bit where the succeeding file starts, i.e., $\rightfile{\file} = \leftfile{\file} + \filesize{\file}$.  The first file begins at the initial position of the tape, $\leftfile{1} = 0$, and the last file ends at position $\tapelength = \rightfile{\numfiles}$, which coincides with the logical end of the tape. The tape length is defined by $\tapelength = \sum\limits_{\file \in \fileset} \filesize{\file}$. 

The tape is traversed by its head driver, which always start at the end of the tape at bit $\tapelength$. At each decision epoch, the data center receives a subset of files in $\fileset$ to be retrieved from the tape. We introduce a parameter $\fileweight{\file} = 1$ to indicate if file $\file \in \fileset$ is requested, and $\fileweight{\file} = 0$ otherwise.
%; such notation simplifies the derivation of the structural results. 
We assume that the first file is always requested, $\fileweight{1} = 1$; otherwise, the left-position of the tape can be adjusted accordingly. The objective is to find a permutation $\ordertuple = (\order_1, \dots, \order_{\numfiles})$ of $\fileset$ that minimizes the total response time of requested files in the reading sequence implied by $\ordertuple$. More precisely, the response time of the $t$-th file, $\order_t \in \fileset$, is the time elapsed until the first bit of $\order_t$ is reached, i.e., 
\begin{align}
	\label{eq:responseTime}
	\policyval{t}{\ordertuple}
	\equiv	
		\sum_{t' = 1}^{t} \left(\filesize{\order_{t'-1}} + \dist{\order_{t'-1}}{\order_{t'}}\right) \nu,
\end{align}
where $\filesize{\order_{0}} = 0$, $\rightfile{\order_{0}}$ is the position at which the tape head begins, $\dist{\order_{t'-1}}{\order_{t'}} \equiv |\leftfile{\filep} - \rightfile{\file}|$ accounts for the time spent reading $\order_{t'-1}$ and repositioning the right-bit position of~$\order_{t'-1}$ to the left-bit position of~$\order_{t'}$, and $\nu$ is the velocity of the tape in time units per bit. We assume without loss of generality that $\nu = 1$. Thus, the deterministic linear tape sequencing problem (\ref{model:LTS}) solves
\begin{align}
	\tag{LTS} \label{model:LTS}
	\optpolicyval 
	\equiv 
	\min_{ \ordertuple } 
		\sum_{t = 1}^{\numfiles} \fileweight{\order_{t}} \policyval{t}{\ordertuple}.
\end{align}

Note that $\numfiles=3$, $\filesize{1} = \rightfile{1} = \leftfile{2} = 15$,  and $\rightfile{3} = \tapelength = 15+4+2 = 21$ for the example in Figure \ref{fig:exampleSchedule1}. For an additional illustration, we also include a larger example in \ref{sec:example} of the electronic companion comparing the response time of two different sequences.

\subsection{Connection to Scheduling and Routing}
\label{sec:connection}

The~\ref{model:LTS} shares connections with other fundamental optimization models from the scheduling and routing literature. In particular, note that, for $t > 1$, we have $\policyval{t}{\ordertuple} = \policyval{t-1}{\ordertuple} + \left(\filesize{\order_{t-1}} + \dist{\order_{t-1}}{\order_{t}}\right)$ in \eqref{eq:responseTime}. Expanding each term in the summation of the objective provides the reformulation
\begin{align}
	\optpolicyval 
	=
	\min_{ \ordertuple } 
	\;
	\sum_{t=1}^{\numfiles} \left( \numrequests - \sum_{t'=1}^{t-1} \fileweight{\order_{t'}} \right) \left( \filesize{\order_{t-1}} + \dist{\order_{t-1}}{\order_{t}} \right),
	\label{eq:objRewriting}		
\end{align}
where $m = \sum\limits_{\file \in \fileset} w_{\file}$ is the number of requested files. Thus, $\optpolicyval$ is equivalent to a latency (or time-dependent) cost function in scheduling. That is, prior to reading a file at the $t$-th position of the sequence, each bit the tape head traverses increases the total response time by the quantity $\left( \numrequests - \sum\limits_{t'=1}^{t-1} \fileweight{\order_{t'}} \right)$,  depicting how many requests are left to be serviced. Based on this perspective, we show below that the \ref{model:LTS} is a special case of the traveling repairperson problem (TRP) and the dial-a-ride problem. Thus, our results are applicable if the  structure of the \ref{model:LTS} is present.

\smallskip
\noindent \textit{Traveling Repairperson Problem (TRP).}
% \noindent \textit{TRP.}
Given a set of points $\mathcal{V} \cup \{0\}$ and symmetric distances $D_{\file, \filep}$ for any pair $\file, \filep \in \mathcal{V} \cup \{0\}$, the TRP asks for a Hamiltonian tour starting at $0$ that minimizes the sum of distances traversed from point $0$ to each other point in the tour. We can cast the~\ref{model:LTS} as an asymmetric variant of the TRP as follows. The vertex $0$ is represented by an artificial zero-sized file $\numfiles+1$ located at the end of the tape, i.e., $\leftfile{\numfiles+1}=\rightfile{\numfiles+1}=\tapelength$, while the remaining points $\mathcal{V} = \{\file \in \fileset \colon \fileweight{\file} = 1 \}$ are mapped to requested files. The distances between any $\file, \filep \in \mathcal{V}$ are $D_{\file, \filep} = \filesize{\file} + \dist{\file}{\filep}$ and $D_{\filep, \file} = \filesize{\filep} + \dist{\filep}{\file}$. That is, during a tour, we arrive on the left $\leftfile{\file}$ of a file $\file$, and moving to the next file $\filep$ requires us to first traverse the length $\filesize{\file}$ of $\file$. 
%The line-TRP, where vertices are distributed on a line, is polynomially solvable in the absence of side constraints, such as processing times, time windows, or deadlines \citep{AfratiCPPP86,Bock15}. 
Thus, %it follows that 
the \ref{model:LTS} is also a special case of the time-dependent TSP \citep{abeledo2013time}.
%similar to the TRP. 

\smallskip
\noindent \textit{Dial-a-ride problem (DARP).} Let $(o_1,q_1), (o_2, q_2), \dots, (o_n, q_n)$ be a set of pickup-and-delivery pairs, where each $o_t$ and $q_t$ is an origin point and a destination point, respectively. The DARP asks for vehicle routes to serve each pair while observing vehicle capacities and a quality metric associated with the distance traversed. The \ref{model:LTS} corresponds to a DARP with a single vehicle of unitary capacity. That is, the requests are distributed on a line and each mapped to a requested file $\file$, $\fileweight{\file} = 1$, with $o_t = \leftfile{\file}$ and $q_t = \rightfile{\file}$. Moreover, $o_t \leq q_t$, i.e., the vehicle always moves from the left to the right when delivering a request, and all requests are positioned on the left of the start point of the vehicle. The objective is to minimize the distances to reach each origin point. In general, the line variant of DARP with a single vehicle and capacity one is NP-complete \citep{PaepeLSSS04}.

\subsection{Practical Considerations}
\label{sec:practicalConsiderations}

Next, we discuss the modeling assumptions and reasoning underlying the \ref{model:LTS}. In our setting, all reading operations concern files located in the same data track, i.e., the tape only moves horizontally. While generally tapes may have other organizational structures, such as a serpentine arrangement, our single-track setting results in a linear positioning of the files, as is common and desirable by data center managers to preserve file locality \citep{Oracle11}. For reference, a single track of a modern tape (e.g., IBM TS1160) may store up to 400 GB of data.

Tape hardware consists of a single reader, so only one file is read at a time. To position the tape  head at a particular bit location, the tape medium is either rewound or fast-forwarded accordingly, which we refer to as tape head positioning in this work. We also note that files are read from left to right. Such a reading (or traversal) direction is due to an operational restriction of tape hardware, as data cannot be retrieved when the track is traversed backwards \citep{ISO20919}. %File data is always retrieved when the tape head is moving from left to right.

The tape heads start at the last position $\tapelength$ because of the append-base nature of tapes. That is, tapes always store their table of contents (TOC) with the list of files, their sizes, and corresponding positions at the last bit $\tapelength$ of the customer-provided data. When loading a tape, the driver must always read the TOC first to locate the remaining files, which positions the head at $\tapelength$. %While that is an ideal location for new data writes, this operation penalizes an eventual demand for data reads. 

We assume that the tape medium is rewound and fast-forwarded at a constant speed $\nu$. This is a standard modeling simplification in tape design and scheduling, as the physical components of a tape drive are mechanical and therefore require acceleration and deceleration. Finally, we note that  the tape speed is not affected by its traversal direction or the execution of a reading operation.

%%%%%%%%%%%%%%%%%%%%%%%%%%%%%%%%%%%%%%%%%%%%%%%%%%%%%%%%%%%%
% SIMPLE POLICIES
%%%%%%%%%%%%%%%%%%%%%%%%%%%%%%%%%%%%%%%%%%%%%%%%%%%%%%%%%%%%

\section{The \ref{model:LTS} with Deterministic Service Requests}
\label{sec:deterministic}

This section investigates the offline \ref{model:LTS}, where we establish $\pi$ knowing all files that have been requested. We begin in \S\ref{sec:stagePartition} by describing a partition of a solution into stages, which will serve as the basis of our results. We then  analyze current policies in \S\ref{sec:basicpolicies}. Next, in~\S\ref{sec:constantRatio} we propose  two new constant-ratio policies that are also easy to justify and implement. Finally, in \S\ref{sec:exactapproach} we show a more technical result of a polynomial-time dynamic program based on the stage partioning constructs.

\subsection{Stage Partitioning}
\label{sec:stagePartition} 

The linear structure of tapes implies a partition of any sequence~$\ordertuple$ in stages of directional movement. Specifically, note that the tape head will always move to bit 0
%the left-most position of the track 
to read the first file. We state this in Definition \ref{def:phases}, which is key for drawing policy intuition. 
\begin{definition}[Rewind and Forward Stages]
	\label{def:phases}
	The \textit{rewind stage} of $\ordertuple$ refers to the set of reading operations and translational movements prior to reaching bit 0 for the first time. The \textit{forward stage} refers to all movements and reading operations after the rewind stage. 
\end{definition}

Given a sequence $\ordertuple$, we denote by $\rewindphase_{\ordertuple}$ and $\forwardphase_{\ordertuple} \equiv \fileset \setminus \rewindphase_{\ordertuple}$ the set of files read during the rewind and forward stages, respectively.  In the example shown in Figure \ref{fig:exampleSchedule1}, $\rewindphase_{\ordertuple} = \{2,3\}$ and $\forwardphase_{\ordertuple} = \{1\}$; note that file 1 always belongs to the the forward stage. Proposition \ref{prop:allforward} shows that once the forward stage begins, the tape head only needs to perform a single left-to-right movement. Consequently, we can restrict our attention to cases where bit 0 is reached only once.
\begin{proposition}[\sc{Forward-Stage Ordering}]
	\label{prop:allforward}
	There exists an optimal sequence $\optordertuple$ such that the forward-stage files are read in ascending order, i.e., $\order^*_{t} < \order^*_{t+1}$ for all $\order^*_{t}, \order^*_{t+1} \in \forwardphase_{\ordertuple^*}$, $t < \numfiles$.
\end{proposition}

\subsection{Analysis of Low-Complexity Policies}
\label{sec:basicpolicies} 

We now analyze three %base 
low-complexity policies that are adopted in practice. We first provide a general description and formalize their theoretical performance in Proposition \ref{prop:basepolicies}.

\smallskip
\noindent \textit{First-In, First-Out} (\fifo{}). The \fifo{} policy, featured in tape management software, sequences files according to the (typically arbitrary) order determined by the system that submitted the requests. \fifo{} is equivalent to a fully randomized policy and, thus, can deliver arbitrarily poor results.

\smallskip
\noindent \textit{First-File-First} (\fififi{}). The \fififi{} policy reads files in ascending order of their indices, i.e., it generates a sequence $\ordertuple$ such that $\order_i < \order_{i+1}$ for all $i = 1,\dots,\numfiles$.
Thus, the policy delays all files to the forward stage.
%and has constant time complexity. 
The practical intuition is that the response time would be low because the tape head makes only two  movements, one to rewind the tape to its bit 0 and another to read the files in a single pass. Indeed, \fififi{} delivers strong constant-factor approximations in line routing problems (e.g., \citealt{bhattacharya2008single}). 

\smallskip
\noindent \textit{Shortest Size First (\ssf)}. The \ssf{} reads files in ascending order of file size, i.e., it is a generalization of the classical shortest processing time first (SPT) policy and has time complexity of $\bigo(m \log m)$. The motivation follows from the latency-type objetive reformulation \eqref{eq:objRewriting}, where the SPT is optimal for related scheduling problems with similar criteria \citep{pinedo2016scheduling}. The \ssf{} also solves the optimal storage in tapes problem (\S\ref{sec:relatedwork}). 

\begin{proposition}[\sc{Performance of \fifo{}, \fififi{}, and \ssf{}}]
	\label{prop:basepolicies}
	The following statements hold:
	\begin{itemize}
		\item[(a)] \fifo{}, \fififi{}, and \ssf{} are  $\Omega(\numfiles)$-approximations for the \ref{model:LTS}.
		\item[(b)] \fififi{} is optimal if all the $\numfiles$ files in the tape are requested and they are of the same size.
	\end{itemize}
\end{proposition}

Contrary to intuition, Proposition \ref{prop:basepolicies} shows that \fififi{} delivers arbitrarily poor solutions to the~\ref{model:LTS}, similar to \fifo{}. The worst-case scenario occurs in the presence of large requested files close to the end of the tape, as reading them  penalizes requested files positioned in the beginning of the tape; \ssf{} also has  arbitrarily poor performance in these cases. Nonetheless, \fififi{} is optimal for a special scenario where all file sizes are equal. This is observed in practice for certain application domains, such as images or videos with the same resolutions or frame rates.

\subsection{Constant-Ratio Approximation Policies}
\label{sec:constantRatio}

The %base 
policies presented in \S\ref{sec:basicpolicies} are common in practice because they are scalable and simple to justify. However, their approximation ratios can be significantly high, which is also reflected in our empirical results (\S\ref{sec:numerical}). Below, we introduce two alternative interpretable policies with  stronger constant-factor theoretical guarantees and better empirical performance. 

\smallskip
\noindent \textit{First-File-Last} (\fifila{}). The \fifila{} policy is the reverse of \fififi{}, in that it services all requested files, except the first, in the rewind stage. That is, it generates a sequence $\ordertuple$ in constant time such that the contiguous subsequence $(\order_1, \order_2, \dots, \order_{\numrequests-1},\order_{\numrequests})$ that prefixes $\ordertuple$ spans the $\numrequests$ requested files and $\order_1 > \order_2 > \dots > \order_{\numrequests} = 1$. 
In other words, requested files are read in descending order.

\fifila{} is akin to a myopic policy because files closest to the tape head are read first. We show in Proposition \ref{prop:fifilaPerformance} that it is a 3-approximation. Further, \fifila{} is also optimal for the case where  all files in the tape have equal size, but it does not require all tape files to be requested as in \fififi{}.
\begin{proposition}[\sc{\fifila{} Performance}]
	\label{prop:fifilaPerformance}
	The \fifila{} policy is a 3-approximation algorithm for the~\ref{model:LTS}. Moreover, it is optimal if all tape files are of the same size.
\end{proposition}

\noindent \textit{Large-Files-Last Policy (\lafila{})}.
The worst-case scenario of \fifila{} occurs if some files with requests are considerably larger than others, especially if they are located at the end of the tape. Servicing such large files in the rewind stage delays all subsequent read operations, thus increasing overall response times. The intuition of the~\lafila{} policy is to improve the solution of~\fifila{} by 
postponing ``large'' files to the forward stage. To this end, we show in Proposition \ref{prop:necessary_optimality_condition} a condition satisfied by any sequence at optimality, which we use to modify \fifila{}.

\begin{proposition}[\sc{Necessary Optimality Condition}]
	\label{prop:necessary_optimality_condition}
	A sequence $\order$ with $\numfiles > 3$ is optimal to the \ref{model:LTS} only if, for any $t \in \{2,\dots,\numfiles-2\}$ such that $\order_{t} \in \rewindphase_{\ordertuple}$, we have
	\begin{align}
		\label{eq:optcond}
		%\dist{\order_{t-1}}{\order_{t+1}} + 
		\sum\limits_{t'=t}^{ \neigh{t} } 
		\left( \filesize{\order_{t'}} + \dist{\order_{t'}}{\order_{t'+1}} \right)
% 		+
% 		\filesize{\order_{\neigh{t}}}
% 		+
% 		\dist{\order_{\neigh{t}}}{\order_{t}}
		\geq
		2\filesize{\order_{t}}\left( \numrequests - \sum_{t' = 1}^{t} \fileweight{\order_{t'}} \right),
	\end{align}
	where $\neighs{}{t} \equiv \argmin_{t' > t} \left\{ \leftfile{\order_{t'}} < \leftfile{\order_{t}} < \leftfile{\order_{t'+1}} \right\}$		that is, $\order_{t}$ is traversed for the second time when the tape moves from 
	$\order_{\neighs{\ordertuple}{t}}$ to $\order_{\neighs{\ordertuple}{t}+1}$.
% 	head is reading the subsequence
% 		$(\order_{\neighs{\ordertuple}{t}},\order_{\neighs{\ordertuple}{t}+1})$.
\end{proposition}

\medskip

The idea of the \lafila{} policy algorithm is to reposition files in a sequence provided by \fifila{} whenever they violate inequality \eqref{eq:optcond}. The resulting algorithm is as follows:

\smallskip
\noindent \underline{\lafila{} Policy}:
\begin{enumerate}
\item Generate an initial solution $\ordertuple$ using \fifila{}.
	\item While there exists some $t \in \{2,\dots,\numfiles\}$ such that inequality \eqref{eq:optcond} is violated:
	\begin{enumerate}
		\item Update~$\ordertuple$ by moving  $\order_{t}$ from~$\rewindphase_{\ordertuple} $ to~$\forwardphase_{\ordertuple}$.
		%Switch files $\order_{t}$ and $\order_{\neigh{t}}$ for the minimum $t$ associated with a violation.
	\end{enumerate}
\end{enumerate}
\medskip

The \lafila{} policy preserves the  3-approximation ratio of attained by~\fifila{} because of Proposition \ref{prop:necessary_optimality_condition}. Moreover, even though \lafila{} has a time complexity of $\bigo(\numrequests^2)$, which could be prohibitive for instances with many requests, we observed empirically in \S\ref{sec:numerical} that its computational performance is adequate for real-world applications. Moreover,  \lafila{} was the best-performing policy overall.

%We also observed that empirically it was the best-performing policy. Nonetheless, it has a time complexity of $\bigo(\numrequests^2)$ as opposed to the constant factor of \lafila{}, which can be prohibitively large if the number of requests is too large. We remark, however, that the condition \eqref{eq:optcond} can also be checked only for a few pairs of (large) files if needed, based on the worst-case analysis of \fifila{}.

%%%%%%%%%%%%%%%%%%%%%%%%%%%%%%%%%%%%%%%%%%%%%%%%%%%%%%%%%%%%
% POLYNOMIAL-TIME ALGORITHM
%%%%%%%%%%%%%%%%%%%%%%%%%%%%%%%%%%%%%%%%%%%%%%%%%%%%%%%%%%%%

\subsection{An Exact Polynomial-time Algorithm}
\label{sec:exactapproach}

We show next that we can derive an exact polynomial-time dynamic programming (DP) approach using the rewind and forward stage constructs. The result is based on a  decomposable structure exhibited in the rewind stage in an optimal solution $\optordertuple$, formalized by Proposition \ref{prop:rewindblock}.

\begin{proposition}[\sc{Order Consistency}]
	\label{prop:rewindblock}
	Every instance of the~\ref{model:LTS} admits an optimal solution~$\optordertuple$ such that, for any file $\file \in \forwardphase_{\optordertuple}$, if $\order_{t'} < \file < \order_{t}$ for any two files $\order_{t}, \order_{t'} \in \rewindphase_{\optordertuple}$ (i.e., file~$\file$ is positioned after~$\order_{t'}$ and before~$\order_{t}$  in the tape), then $\order_{t}$ is read prior to $\order_{t'}$, that is, $t < t'$.
% 	Let $\optordertuple$ be an optimal solution to the~\ref{model:LTS} and consider any file $\file \in \forwardphase_{\optordertuple}$. If $\order_{t'} < \file < \order_{t}$ for any two files $\order_{t}, \order_{t'} \in \rewindphase_{\optordertuple}$ (i.e., file~$\file$ is positioned after~$\order_{t'}$ and before~$\order_{t}$  in the tape), then $\order_{t}$ is read prior to $\order_{t'}$, that is, $t < t'$.
\end{proposition}

Proposition~\ref{prop:rewindblock} states that rewind-stage files at optimality may be organized in contiguous ``blocks,'' in that each block is separated by one or more forward-stage files, and its files compose a contiguous subsequence of~$\ordertuple$; Figure \ref{fig:exampleOptSolution} of the electronic companion illustrates such a solution. 

We propose a DP formulation that enumerates and optimizes possible rewind-stage blocks. More precisely, the formulation combines a \textit{forward-stage recursion} $\forwardrec(\file,\filep, k)$  and a \textit{rewind-stage recursion} $\rewindrec(\file,\filep,k)$. Both recursions are defined on a state space associated with a subset~$\{\file,\file+1, \ldots,\filep \} \subseteq \fileset$ and by some~$k \in \{1,2,\ldots,m\}$, which represents the number of pending requests for the first time the tape head reaches position~$\rightfile{\filep}$. 

Let $\pendreqs{\file'}{\filep'} \equiv k - \sum\limits_{\file'' = \file'}^{\filep'}\fileweight{\file''}$ represent the number of requests still pending  after reading files $\file',\file'+1,\ldots,\filep'$. The recursion~$\rewindrec(\file,\filep,k)$ considers only solutions where the tape head finishes at~$\leftfile{\file}$, i.e.,
\begin{eqnarray}\label{eq:rewindrec}
		\rewindrec(\file,\filep,k) 
		=		\begin{cases}
			\min 
			\left \{  
				\min\limits_{\file < \filep' \le \filep} 
				\left\{
				    \rewindrec(\filep',\filep,k) 
					+ 
                	\rewindrec(\file,\filep'-1,\pendreqs{\filep'}{\filep}) 
				\right\},
				\forwardrec(\file,\filep, k) + \pendreqs{\file}{\filep} \dist{\filep}{\file}
			\right\},
				& \textnormal{$\file < \filep$} \\
			k \filesize{\file} + 2(k - \fileweight{\file}) \filesize{\file},
				& i = j.%\textnormal{o.w.}
		\end{cases}
\end{eqnarray}
 If~$\file = \filep$, the block consists of file $\file$ only. In this case, the tape head moves first to~$\leftfile{\file}$, reads~$\file$ by moving to~$\rightfile{\file}$, and then returns to~$\leftfile{\file}$.  Otherwise, if~$\file < \filep$, two different  strategies are considered. The first case, represented by the inner-most ``$\min$'' expression, decomposes the ``block'' $\{\file,\file+1,\ldots,\filep\}$ into two sub-problems based on a file~$\filep'$ positioned between~$\file+1$ and~$\filep$. Given~$\filep'$, we recursively invoke~$\rewindrec(\filep',\filep,k)$ to read the requests in~$\{\filep',\filep'+1,\ldots,\filep\}$ first, and then we invoke~$\rewindrec(\file,\filep'-1,\pendreqs{\filep'}{\filep})$ to read the other files; observe that~$\filep'$ is read before all files in~$\{\file,\file+1,\ldots,\filep'-1\}$. The second case considers solutions of~$\forwardrec(\file,\filep,k)$ (whereby the tape head finishes at~$\rightfile{\filep}$) followed by a movement to~$\leftfile{\file}$. The forward-stage recursion~$\forwardrec(\file,\filep,k)$ is defined as follows:
\begin{eqnarray}\label{eq:valuerec}
	\forwardrec(\file,\filep,k) =\;\;& 
				k \filesize{\filep}
				%&[\textnormal{term (1)}]			
			+
			\min\limits_{\file \le \filep' < \filep} 
			\left \{ 
			    \rewindrec(\filep'+1,\filep-1,k)
				+
				\forwardrec(\file,\filep',\pendreqs{\filep'+1}{\filep-1})
				+ 				
				\pendreqs{\file}{\filep-1}
				\dist{\filep'}{\filep}
			\right \} 
			+ 
			\pendreqs{\file}{\filep}
			\filesize{\filep}.
			\label{eq:valuerec}
\end{eqnarray}
Solutions of recursion~$\forwardrec(\file,\filep, k)$ incorporate a forward-stage  movement, in which the tape head moves from~$\leftfile{\file}$ to~$\rightfile{\filep}$; in particular, $\filep$ is the last file to be read in~$\forwardrec(\file,\filep, k)$. Moreover, file~$\filep'$ identified  in~\eqref{eq:valuerec} is the forward-stage file preceding~$\filep$ in the sequence. Thus, we can solve the new block $\{\filep',\dots, \filep\}$ recursively through~$\rewindphase(\filep'+1,\filep-1,k)$. Afterwards, the tape head must move leftwards and read~$\file,\file+1,\ldots,\filep'-1$ before reading~$\filep'$; this sequence is identified by~$\forwardrec(\file,\filep',\pendreqs{\filep'+1}{\filep-1})$, which is then succeeded by the movement from~$\rightfile{\filep'}$ to~$\leftfile{\filep}$.

\begin{theorem}	\label{thm:lts_exact}
	The \ref{model:LTS} is given by~$\rewindrec(1, \numfiles, \numrequests)$, which can be computed in time~$\bigo(\numfiles^4)$. 
\end{theorem}

The DP is excessively technical and of high complexity to be implemented in practice.  However, we show that it could still be useful to develop tractable approximation schemes (\S\ref{sec:stochastic}).

%%%%%%%%%%%%%%%%%%%%%%%%%%%%%%%%%%%%%%%%%%%%%%%%%%%%%%%%%%%%
% Stochastic Extension
%%%%%%%%%%%%%%%%%%%%%%%%%%%%%%%%%%%%%%%%%%%%%%%%%%%%%%%%%%%%
\section{The Stochastic~\ref{model:LTS}}
\label{sec:stochastic}

In this section, we consider scenarios with frequent service requests and where file access distribution is known. Similar in spirit to the low-complexity policies in \S\ref{sec:deterministic}, the objective is to predetermine a file sequence that would be applied to all service requests and minimizes the expected response time. We begin in \S\ref{sec:stsp} with the formalization of the stochastic~\ref{model:LTS} and present a fully polynomial-time approximation scheme (FPTAS) based on probability scaling in \S\ref{sec:fptas}.

\subsection{Stochastic Setting Description}
\label{sec:stsp}

We consider that the request status of each file $\file \in \fileset$ is a random variable $\fileweightRV{\file}$ described by a Bernoulli distribution with parameter $\prob{\file} \in [0,1]$, i.e., file $\file$ is requested with probability $\prob{\file}$ and not requested with probability $1 - \prob{\file}$. We wish to find a sequence $\ordertuple$ that solves
\begin{align}
	\tag{SLTS} \label{model:SLTS}
	\expect \left[ \optpolicyval \right] 
	\equiv 
	\min_{ \ordertuple } 
		\sum_{t = 1}^{\numfiles} \expect \left[ \fileweight{\order_{t}} \policyval{t}{\ordertuple} \right]
	=
	\min_{ \ordertuple } 
		\sum_{t = 1}^{\numfiles} \expect \left[ \fileweight{\order_{t}} \right] \policyval{t}{\ordertuple}
	=
	\min_{ \ordertuple } 
	\sum_{t = 1}^{\numfiles} \prob{\order_{t}} \policyval{t}{\ordertuple}.
\end{align}
That is, the \ref{model:SLTS} is a variant of the \ref{model:LTS} where response times are now weighted by probabilities $\prob{\file}$. Analogously to \eqref{eq:objRewriting}, we can rewrite the \ref{model:SLTS}' objective as
\begin{align}
	\expect \left[ \optpolicyval \right] 
	=
	\min_{ \ordertuple } 
	\;
	\sum_{t=1}^{\numfiles} \left( \probleft_0 - \sum_{t'=1}^{t-1} \prob{\order_{t'}} \right) \left( \filesize{\order_{t-1}} + \dist{\order_{t-1}}{\order_{t}} \right),
	\label{eq:objRewritingS}		
\end{align}
where $\probleft_0 \equiv \sum\limits_{\file \in \fileset} \prob{\file}$ is the sum of file probabilities.

In practice, the probabilities $\prob{\file}$ can be derived as empirical access frequencies based on historical data for tape contents that are retrieved often (e.g., for storages containing popular videos and images). The sequence $\optordertuple$ solving the \ref{model:SLTS} could be pre-computed and adopted  for the period where the distribution remains unchanged, which is acceptable and preferred in light of \S \ref{sec:practicalConsiderations} and \S \ref{sec:deterministic}. 
Notice that the sequence $\optordertuple$ can also be dynamically adjusted to skip files that are not present in the request (e.g., similar to \citealt{bertsimas1990priori}).

\subsection{Fully Polynomial Approximation Scheme (FPTAS)}
\label{sec:fptas}

The recursions associated with the rewind-stage $\rewindrec(\cdot)$ in \eqref{eq:rewindrec} and with the forward-stage $\forwardrec(\cdot)$ in \eqref{eq:valuerec} can be equivalently rewritten in terms of the sum of probabilities left, as opposed to the number of requests still to process, to solve~\eqref{eq:objRewritingS}.
%Based on the expectation objective \eqref{eq:objRewritingS}, the recursions associated with the rewind-stage $\rewindrec(\cdot)$ in \eqref{eq:rewindrec} and with the forward-stage $\forwardrec(\cdot)$ in \eqref{eq:valuerec} can be equivalently rewritten in terms of the sum of probabilities left, as opposed to the number of requests still to process. 
That is,
\begin{align}
	\label{eq:stocrewindrec}
	\stocrewindrec(\file, \filep, \probleft) 
	=		
	\begin{cases}
		\min 
		\left \{  
			\min\limits_{\file < \filep' \le \filep} 
			\left\{
				\stocrewindrec(\filep',\filep, \probleft) 
				+ 
				\stocrewindrec(\file,\filep'-1, \stocpendreqs{\filep'}{\filep}) 
			\right\},
			\stocforwardrec(\file,\filep, \probleft) + \stocpendreqs{\file}{\filep} \dist{\filep}{\file}
		\right\},
			& \textnormal{$\file < \filep$} \\
		\probleft \filesize{\file} + 2(\probleft - \prob{\file}) \filesize{\file},
			& i = j,%\textnormal{o.w.}
	\end{cases}
\end{align}
and 
\begin{align}
\label{eq:stocvaluerec}
\stocforwardrec(\file,\filep, \probleft) =\;\;& 
			\probleft \filesize{\filep}
			%&[\textnormal{term (1)}]			
		+
		\min\limits_{\file \le \filep' < \filep} 
		\left \{ 
			\stocrewindrec(\filep'+1,\filep-1, \probleft)
			+
			\stocforwardrec(\file,\filep',\stocpendreqs{\filep'+1}{\filep-1})
			+ 				
			\stocpendreqs{\file}{\filep-1}
			\dist{\filep'}{\filep}
		\right \} 
		+ 
		\stocpendreqs{\file}{\filep}
		\filesize{\filep},
\end{align}
where $\stocpendreqs{\file'}{\filep'} \equiv \probleft - \sum\limits_{\file'' = \file'}^{\filep'}\prob{\file''}$ is the remaining sum of probabilities of files in~$\probleft$ after reading~$\{\file', \dots, \filep'\}$. That is, the adapted recursions change with respect to the third state variable $\probleft$, which  now tracks the total sum of probabilities of the files yet to be read. Note that the state space of $\probleft$ is $\left\{ \sum\limits_{\file \in \fileset'} \prob{\file} \colon \forall \fileset' \subseteq \fileset \right \}$. Thus, the recursions \eqref{eq:stocrewindrec}-\eqref{eq:stocvaluerec} are not solvable in polynomial time using traditional value or policy iterations. Theorem \ref{thm:fptas}, our main result in this section, shows that any instance can be scaled appropriately to obtain a tractable approximation to the problem. 
\begin{theorem}[\sc{FPTAS}]
	\label{thm:fptas}
    Let $I$ be an arbitrary instance of the \ref{model:SLTS} and any $\epsilon > 0$. 
	Consider the new instance $I'$ obtained by applying steps (a), (b), and (c) consecutively:
	\begin{itemize}
		\item[(a)] multiply all probabilities by $1/\minprob$, where $\minprob \equiv \min\limits_{\file \in \fileset \colon \prob{\file} > 0} \prob{\file}$.
		\item[(b)] for the scaled probabilities, let $\maxprob = \max\limits_{\file \in \fileset} \prob{\file}$ and add two new files to the end of the tape with size-probability pairs $(\filesize{\numfiles+1}, \prob{\numfiles+1}) = (\tapelength,0)$ and 
		$(\filesize{\numfiles+2}, \prob{\numfiles+2}) = (\tapelength, 2 \numfiles^2 \maxprob)$, in order; and
		\item[(c)] change file probabilities to 
		$\probapprox{(\epsilon)}{\file} \equiv \lfloor \prob{\file} / K \rfloor$, where 
		$$
		K \equiv \frac{\epsilon \, \numfiles^2 \, \maxprob}{(\numfiles+3)(\numfiles+1)}.
		$$
	\end{itemize}
	 Then, $I'$ is polynomially solvable in $n$ and $1/\epsilon$ and provides an $(1+\epsilon)$-approximation for $I$.
\end{theorem}

\begin{proof}{Proof of Theorem \ref{thm:fptas}.}
	Step (a) ensures that the objective coefficients are all larger than one, which simplifies calculations and does not change the optimal sequence, as we show below.% It is supported by the following lemma, showed in the electronic companion:
	\begin{lemma}[\sc{Scale Invariance}]
		\label{lemma:scale}
		An optimal sequence $\optordertuple$ to an instance~$I$ of the~\ref{model:SLTS} remains optimal if we multiply all file sizes and probabilities by~$\alpha > 0$ and~$\beta > 0$, respectively.
	\end{lemma}

	The sum of probabilities of the resulting new instance $I'$ satisfies 
	$$
	\probleft'_0 \equiv \sum_{\file \in \fileset} \probapprox{(\epsilon)}{\file}
	\le
		\numfiles \left\lfloor \frac{\maxprob}{K} \right\rfloor + 
	\left\lfloor \frac{2n^2\maxprob}{K} \right\rfloor
	= 
	\numfiles \left\lfloor \frac{(\numfiles+3)(\numfiles+1)}{\epsilon \numfiles^2} \right\rfloor +  \left\lfloor \frac{2(\numfiles+3)(\numfiles+1)}{\epsilon} \right\rfloor,
	$$
	which is polynomial in~$\numfiles$ and~$\frac{1}{\epsilon}$ for any given $\epsilon > 0$. Thus, the state space associated with the variable $\probleft$ in \eqref{eq:stocrewindrec}-\eqref{eq:stocvaluerec} for $I'$ is polynomially bounded, and we can solve $I'$ in polynomial time in~$\numfiles$ and~$1/\epsilon$ using the DP from \S\ref{sec:exactapproach}.
	Let $\optordertuple$ be an optimal solution to $I'$. The solution value of $I'$ with respect to $\optordertuple$ is
	\[
	v(\optordertuple,I') = 
		\sum_{t=1}^{\numfiles} 
		\left( 
			\probleft'_0 - \sum_{t'=1}^{t-1} \probapprox{(\epsilon)}{\order_{t'}} 
		\right) 
		\left( 
			\filesize{\order_{t-1}} + \dist{\order_{t-1}}{\order_{t}} 
		\right).
	\]

	For each iterate $t$ in the left-most sum above, let $\probleft'_t \equiv \probleft'_0 - \sum\limits_{k'=1}^{t-1} \probapprox{(\epsilon)}{\order_{t'}}$
	and $\Delta'_t \equiv \filesize{\order_{t-1}} + \dist{\order_{t-1}}{\order_{t}}$ be the scaled leftover probability and the time elapsed reading 
	file $\order_{t-1}$ and moving to $\leftfile{\order_{t-1}}$, respectively. 
	Further, let $c'_t \equiv \probleft'_t \Delta'_t$, so that $v(\optordertuple,I') = \sum\limits_{t=1}^n c'_t$.  We analogously use $\probleft_t$, $\Delta_t$, and $c_t$ to denote the same quantities above for $v(\optordertuple,I)$, i.e., the solution value of $\optordertuple$ for the original instance $I$. 
	
	We have $\prob{\file} - K \probapprox{(\epsilon)}{\file} < K$ by construction. Thus, for every~$t$ in~$\{0,1,\ldots,\numfiles\}$, we have $\probleft_t \leq K \probleft'_t + (\numfiles + 1) K$, i.e., an (additive) factor bounded by $K$ may be lost for each unserviced file with non-zero probability. Since $\Delta'_t = \Delta_t$,
	%and as file sizes are the same in~$I$ and~$I'$, 
	we have~$c_t \leq K c'_t + (\numfiles + 1) K \Delta_t$, and therefore
	\begin{eqnarray}\label{eq:fptas1}
	 v(\ordertuple,I) 
		&=&
		\sum_{t=1}^n c_t 
		\leq 
			\sum_{t=1}^n [Kc'_t + (\numfiles+1) K \Delta_t] 
		\leq
			v(\optordertuple,I) + (\numfiles+1) K \sum_{t=1}^n \Delta_t \nonumber \\
		&\leq&
			v(\optordertuple,I) + \frac{\epsilon \numfiles^2 \maxprob  }{\numfiles+3}\sum_{t=1}^n \Delta_t.
	 \end{eqnarray}
	Observe that~$\sum\limits_{t=1}^\numfiles \Delta_t \leq 2\numfiles \tapelength + 6\tapelength$ in any optimal sequence, which is achieved if the tape head reads the last two files first and then traverses the first~$\tapelength$ bits of the tape forward and back for each of the~$\numfiles$ files. Additionally, we must have~$v(\optordertuple,I) \geq 2 \tapelength \numfiles^2 \maxprob$, as~$\tapelength$ is the minimum response time of the last (artificial) file, which has scaled probability~$2 \numfiles^2 \maxprob$. It follows that
	\begin{eqnarray}\label{eq:fptas2}
	(2 \tapelength n^2 \maxprob) \left( \frac{\sum_{t=1}^n \Delta_t}{2\numfiles \tapelength +6 \tapelength}\right) 
	\leq v(\optordertuple,I) \implies
	\maxprob \sum_{t=1}^n \Delta_t \leq 
	\frac{\numfiles+3}{n^2}v(\optordertuple,I)
	\end{eqnarray}
	Finally, we obtain
	$
	v(\ordertuple,I) \leq v(\optordertuple,I) +
	\epsilon v(\optordertuple,I) \leq (1 + \epsilon) v(\optordertuple,I)
	$ by replacing~\eqref{eq:fptas2} into~\eqref{eq:fptas1}. \hfill $\blacksquare$
\end{proof}

Theorem \ref{thm:fptas} can be applied to any arbitrary ``weighted'' version of the \ref{model:LTS}, in that some files are given priority based on their $\prob{\file}$ values. If the weights are discrete and polynomially bounded in $n$, it follows that the weighted \ref{model:LTS} can be solved efficiently. However, the question of whether the problem is polynomially solvable for general weights remains open. 

\smallskip

%%%%%%%%%%%%%%%%%%%%%%%%%%%%%%%%%%%%%%%%%%%%%%%%%%%%%%%%%%%%
% Online Extension
%%%%%%%%%%%%%%%%%%%%%%%%%%%%%%%%%%%%%%%%%%%%%%%%%%%%%%%%%%%%
\section{The Online~\ref{model:LTS}}
\label{sec:online}

In this section, we consider an online extension of the \ref{model:LTS} whereby new file requests may become available during read operations of current files. Our objective concerns the \textit{online file response time}, i.e., the total time elapsed from the beginning of operations until the file was read. We show that \fifo{} is inefficient in this context and propose a constant-ratio policy in \S\ref{sec:ari}. Next, we discuss the inaproximability of the online~\ref{model:LTS} for an alternate objective function in \S\ref{sec:releaseBased}.

\subsection{Augmenting Reading Intervals}
\label{sec:ari}

We begin by stating the result that \fifo{}, which is standard in industry and sequences files based on their arrivals, achieves an arbitrarily poor competitive ratio. 
\begin{proposition}\label{prop:fifo_not_competitive}
The \fifo{} policy is not $c$-competitive for any constant $c$ in the online \ref{model:LTS}.
\end{proposition}

We now introduce a policy for the online~\ref{model:LTS} that combines ideas from \cite{AusielloFLST01} for the online dial-a-ride problem and \cite{baezayates1993searching} for point search in a plane. The algorithm partitions the tape into equally-sized contiguous intervals  and reads such intervals through incrementally larger right-to-left traversals. We name this policy augmenting reading interval (\zigzag{}), which is parameterized by a step size~$\alpha \geq 1$ and described as follows:

\noindent \underline{\zigzag{} Policy:}
	\begin{enumerate}
	    \item Let $\delta \equiv \gcd{ \{ \filesize{\file}:\file \in \fileset \}}$ and $\tapelength' \equiv \tapelength/\delta$ be the length and number of intervals, respectively;
		\item For $t = 1,\dots, \lceil\log( \tapelength')\rceil$:
		\begin{enumerate}
		    \item Move the tape head~$\min(\alpha^{t},\tapelength')$ intervals to the left of $\tapelength$.
		    \item Move the tape head back to~$\tapelength$.
		    %, reading all files fully contained in the last~$\min(2^{t},\tapelength')$ intervals.
% 			\item Position the tape head at $\tapelength_t \equiv \tapelength - \delta t$
% 			\item Move the tape head back to $0$, reading all files fully contained in $0$ and $\tapelength_t$ in order.
		\end{enumerate}
	\end{enumerate}
%Observe that all reading operations take place on Step 3-b. We show that~\zigzag{} is a constant-factor competitive algorithm for the online \ref{model:LTS}.
\begin{theorem}\label{thm:onlinealg}
    \zigzag{} is 7-competitive for the online \ref{model:LTS} for~$\alpha = 2$, and the competitive factor is asymptotically tight. Moreover, no other value of $\alpha$ achieves a better competitive ratio.
\end{theorem}
\begin{proof}{Proof of Theorem~\ref{thm:onlinealg}:}
We assume w.l.o.g. that $\delta = 
\gcd(\{\filesize{\file}:\file \in \fileset \}) = 1$ and, thus, $\tapelength' = \tapelength$. We use~$\{0, 1,2,\ldots,\tapelength\}$ in reversed order to simplify the notation used in this proof, i.e., positions 0 and~$\tapelength$ are the last and first positions of the tape, respectively. We divide the analysis in two cases.
%, based on when each file is read.

\paragraph{First reading:} Let~$p: \{1,2,\ldots,\tapelength\} \rightarrow \mathbb{N}$ denote the moment when each position of the tape is traversed rightwards for the first time. First, for every position~$\alpha^k$, $k \in \mathbb{N}$, we show that~$p(\alpha^k) = 2(\alpha + \alpha^2 + \ldots + \alpha^{k-1}) + \alpha^k$ by induction in~$k$. For~$k = 1$, we have $p(\alpha^1) = \alpha$, which is the time the tape head needs to move from 0 to $\alpha$, so the base holds. Let us assume that the formula holds for position~$\alpha^{k-1}$. Position~$\alpha^k$ is visited for the first time at time
\begin{eqnarray*}
p(\alpha^k) &=& 
\underbrace{p(\alpha^{k-1})}_{\text{Time to reach~$\alpha^{k-1}$}} 
+ 
\underbrace{\alpha^{k-1}}_{\text{Moving from~$\alpha^{k-1}$ to~$0$}} 
+ 
\underbrace{\alpha^{k}}_{\text{Moving from~$0$ to~$\alpha^{k}$ }} \\
&=&
    \underbrace{     2(\alpha + \alpha^2 + \ldots + \alpha^{k-2}) + \alpha^{k-1}}_{\text{Induction hypothesis}}
    +
    \alpha^{k-1}
    +
    \alpha^k
\\
&=&
    2(\alpha + \alpha^2 + \ldots + \alpha^{k-2} + \alpha^{k-1}) +
    \alpha^k,
\end{eqnarray*}
so the result holds. Next, observe that the first visit to files~$\alpha^{k}-1,\alpha^{k}-2,\ldots,\alpha^{k-1}+1$ take place after the first visit to~$\alpha^k$. In particular, observe that
position~$\alpha^{k-1}+1$ is the last among~$1,2,\ldots,\alpha^k$ to be traversed leftwards, an operation that takes place at time
\begin{eqnarray*}
p(\alpha^{k}+1) &=& \underbrace{p(\alpha^{k+1})}_{\text{Time to reach~$\alpha^{k+1}$}} + \underbrace{[\alpha^{k+1} - (\alpha^{k} + 1)]}_{\text{Moving from~$\alpha^{k+1}$ to~$\alpha^{k} + 1$ }} \\
&=& 
[2(\alpha + \alpha^2 + \ldots + \alpha^{k-1} + \alpha^{k}) +
    \alpha^{k+1}] + [\alpha^{k+1} - (\alpha^{k} + 1)] \\
&=& 
2(\alpha + \alpha^2 + \ldots + \alpha^{k} + \alpha^{k+1}) -(  \alpha^{k}+1)     
\end{eqnarray*}
The position of the file defines a lower bound on its response time, i.e., if a file starts at position~$\alpha^k+i$, the minimum response time is also~$\alpha^k+i$. Thus, the ratio~$c$ between the response time given by our policy for any file starting from~$1,2,\ldots,\alpha^k$ and the minimum response time for any request for~$\file$ released within the first~$\alpha^k$ time steps is maximum at position~$\alpha^{k-1} + 1$ and is given by
\begin{eqnarray*}
c 
	&\leq& 
	\frac{p(\alpha^{k}+1)}{\alpha^{k}+1} 
    = 
    \frac{ 2(\alpha + \alpha^2 + \ldots +  \alpha^{k+1} ) }{ \alpha^{k}+1 } -1    
    =
	\frac{ 2\alpha^{k+1}(1 + 1/\alpha + \ldots + 1/\alpha^{k} ) }{ \alpha^{k}+1 } -1
	\\
    &<&
    \frac{ 2\alpha^{k+2} }{ (\alpha^{k}+1)(\alpha-1) } -1 \\
    &<&
    \frac{ 2\alpha^{2} }{ \alpha-1 } -1,
\end{eqnarray*}
which is minimum at~$\alpha = 2$ and results in $c < 7$. Thus, \zigzag{} takes at most 7 times longer to reach and traverse any file leftwards for the first time than any other policy if~$\alpha = 2$. 

\paragraph{Other readings:}
Let us consider now the~$v$-th read operation of each file. First, observe that after reading~$2^k$ for the first time, the policy moves rightwards to~$\tapelength$ first and then left to~$2^{k+1}$, from where it moves back to read~$2^k$ for the second time. The same procedure is repeated for the other visits; we show that the~$v$-th visit of position~$2^k$ takes place at time $p(2^k,v) = 2\sum\limits_{q = 0}^{k+v-1}2^{q} - 2^{k}$ by induction in~$v$. For~$v = 1$, we have~$p(2^k,1) = p(2^k)$, so the base case holds. If we assume that the result holds for the~$(v-1)$-th visit, we have
\begin{eqnarray*}
p(2^k,v) 
    &=& 
        \underbrace{p(2^{k},v-1)}_{(v-1)\text{-th visit}} 
        + 
        \underbrace{2^{k}}_{\text{Moving from~$2^k$ to~$0$ }} 
        + 
        \underbrace{2^{k+v-1}}_{\text{Moving from~$0$ to~$2^{k+v-1}$}}
        +
        \underbrace{2^{k+v-1} - 2^k}_{\text{Moving from~$2^{k+v-1}$ to~$2^{k}$}}  
\\
    &=& 
        2\sum\limits_{q = 0}^{k+v-2}2^{q} 
        +
        2 \cdot 2^{k+v-1}
        - 
        2^k
        =
        2\sum\limits_{q = 0}^{k+v-1}2^{q} - 2^k.
\end{eqnarray*}
For arbitrary files~$2^{k} - q$, $0 \leq q < 2^{k-1}$, the~$v$-th visit takes place at time
$
p(2^{k} - q,v) = p(2^k,v) + q = 2\sum\limits_{q = 0}^{k+v-1}2^{q} - 2^k + q.
$
Thus, the difference between the~$(v+1)$-th and the~$v$-th visit for~$v \geq 1$ is
\begin{eqnarray*}
p(2^k - q,v+1) - p(2^k - q,v) 
    &=& 
        \underbrace{2^k - q}_{\text{Moving from~$2^k - q$ to 0}}
        +
        \underbrace{2^{k+v-1}}_{\text{Moving from~$0$ to~$2^{k+v-1}$}}
        +
        \underbrace{2^{k+v-1} - (2^k-q)}_{\text{Moving from~$2^{k+v-1}$ to~$2^{k}-q$}} \\
    &=&
        2^{k+v}.
\end{eqnarray*}
This implies that the competitive ratio~$c'$ for the~$v$-th visit of an arbitrary position~$2^k - q$, $v \geq 2$, which covers requests released after the~$(v-1)$-th visit, is \begin{eqnarray*}
c' &=& 
        \frac{p(2^k - q,v+1)}{p(2^k - q,v)}
    =
        1 + \frac{2^{k+v}}{p(2^k - q,v)}
%        {\sum\limits_{i = 0}^{v}2^{k+i} - 2 + q}    
    \leq
        1 + \frac{2^{k+v}}{2\sum\limits_{q = 0}^{k+v-1}2^{q} - 2^k + q}
%        2^{k+v+1} -1
            \leq
        1 + \frac{2^{k+v}}{2\cdot 2^{k+v} - 2^k + q}
        \leq
        2.
\end{eqnarray*}
It follows that~\zigzag{}  is 7-competitive for the online~\ref{model:LTS}. Finally, even though the inequality bounding the competitive factor for the first visit is strict, the factor 7 is asymptotically tight. More precisely, \zigzag{} consumes time~$2(2^1 + 2^2 + \ldots + 2^{k-1}) + 2^K \approx 3\cdot 2^k$ to reach position~$2^k$ for the first time. Moreover, once position~$2^k$ is reached, the tape head moves back to the end of the tape and then to~$2^{k+1}$ before going back to~$2^k+1$; in total, these movements require time~$4 \cdot 2^k -1$, so the visit time of~$2^k+1$ is approximately~$7 \cdot 2^k -1$. 
\hfill $\blacksquare$
\end{proof}	

% \smallskip

The proof of Theorem~\ref{thm:onlinealg} shows that the first visit is the most critical for the competitive factor of~\zigzag{}; this is not surprising, as late releases naturally have higher minimum response times. 

\subsection{Adjusted Response Times}
\label{sec:releaseBased}

Next, we discuss an online version of the~\ref{model:LTS} with alternative objective criteria. Specifically, we consider the setting where the response time is adjusted to incorporate the time at which a request arrived. For example, if a request is released at time 7 and serviced at time 10, its adjusted response time is $3$, while the actual response time (used in Proposition~\ref{prop:fifo_not_competitive} and Theorem~\ref{thm:onlinealg}) is~$10$. This relates to the concept of flow time (or time-in-system) objective criteria in the machine scheduling literature \citep{kanet1981minimizing}. We show that there is no policy with bounded-factor guarantees for the online~\ref{model:LTS} if the goal is to minimize the sum of the adjusted response times.
\begin{proposition}\label{prop:onlinebad}
	There is no $f$-competitive policy that minimizes the sum of adjusted response times  in the online~\ref{model:LTS} for any bounded function $f$.
\end{proposition}

%%%%%%%%%%%%%%%%%%%%%%%%%%%%%%%%%%%%%%%%%%%%%%%%%%%%%%%%%%%%
% Numerical Results
%%%%%%%%%%%%%%%%%%%%%%%%%%%%%%%%%%%%%%%%%%%%%%%%%%%%%%%%%%%%
\section{Numerical Study}
\label{sec:numerical}

% random policies from real distributions to numerically evaluate approximations
% policy names
% completion times / average completion times?
% numerical relative performance
% discuss why rewing stage is important
% remove number of files per batch

% \smallskip 
% \url{https://deepnote.com/project/Tape-Numerical-Analysis-Xas0aTP0T5Gbe3qZM4e-hw/%2FTapes-20210722-203935.ipynb}
% \smallskip 

In this section, we perform a numerical study of the sequencing policies investigated in this work. We begin in \S\ref{sec:policyperformance} with a sensitivity analysis on artificial instances based on empirical surveys of file distributions. In \S\ref{sec:landsat}, we provide a study on real tapes storing large-scale satellite imagery. 
Tables with detailed results are included in \S\ref{sec:appce} of the electronic companion.

The response times are calculated in seconds based on Linear Tape-Open 8 (LTO-8) magnetic tape storage technology with an average read speed of $\nu =$360 MB/s for uncompressed data~\citep{quantum2021}. The experiments are run on an Intel(R) Skylake CPU at 2.40GHz for the purposes of comparative analysis. For reference, commercial tape drives use low-power embedded processors (e.g., PowerPC and ARM), which are approximately two to three orders of magnitude slower. Sequencing algorithms in tapes are limited to less than a second in total runtime.

\subsection{Synthetic Instances}
\label{sec:policyperformance}

We compare the performance of the deterministic policies in terms of solution quality and time, denoting the exact approach \eqref{eq:rewindrec}\textemdash\eqref{eq:valuerec} by \exact{}. We omit~\ssf{} because the policy has no constant-factor approximation guarantees and is not adopted in practice
to the best of our knowledge. Parameters $\fileweight{\file}$ are drawn at random from  a Bernoulli distribution with a fixed probability $p$ per file, where $p \in \{0.25, 0.50, 0.75, 1.0\}$. File sizes are based on the study by \cite{douceur1999large}, which at the time reported a log-normal distribution with parameters $\mu = 8.46$ and $\sigma = 2.38$. We adjust these parameters based on  trends in storage requirements (e.g., \citealt{agrawal2007five}) and draw file sizes from a log-normal distribution with parameters $\mu = 13.04$ and $\sigma \in \{1.50, 2.00, 2.38, 2.5, 3.0\}$. For instance, $\mu = 13.04$ and $\sigma = 2.38$ correspond to an average of $\sim$7.8 megabytes (MB) and a standard deviation of $\sim$13.2 MB. To control for numerical issues, we truncate file sizes to the 90\% quantile of the corresponding distribution and divide them by 1,000, i.e., our unit of reference is a kilobyte. For \fifo{}, each run corresponds to the average of 1,000 sequences generated uniformly at random since the file ordering is arbitrary. %Furthermore, each run of~\exact{} is limited to 3,600 seconds and 8 GB of RAM.  

We consider two ranges of tape sizes. Small tapes are composed of $\numfiles \in \{100, 200, 300, 400\}$ files; the limit of 400 is due to the memory requirements of~\exact{}, where $n=400$ matches our computational limits. For larger cases, tracks in a modern tape store up to 400GB and may contain tens of thousands of files. Thus, we also evaluate large-scale instances of size $\numfiles \in \{800, 1600, 3200, 6400, 12800, 25600, 51200, 76800, 102400\}$.  We generate 10 instances per $(\numfiles, p, \mu, \sigma)$.

\paragraph{Policy Performance and Quality of Solutions.} Figure \ref{fig:filesystem_approx}-(a) depicts the average response time per file, i.e., the total response time divided by the number of requests. The results are aggregated by the number of files~$\numfiles$, and shaded areas correspond to the 95\% confidence interval. As expected, the average response times increase with~$\numfiles$ for all policies, but there are clear differences in terms of the quality of their solutions. The average response time for~\fifo{} reaches up to $4$ seconds for $\numfiles=400$, which is approximately twice the (optimal) response time of~\exact{}. \fififi{} improves on \fifo{} but underperforms with respect to \fifila{}, which is consistent with their relative theoretical performance. \lafila{} matched the optimal solution in this data set, so its curve coincides with~\exact{}. 
\begin{figure}[ht!]	
	\centering
	\captionsetup{justification=centering}
	\scalebox{1.0}{
		\begin{subfigure}[b]{0.5\textwidth}
			\includegraphics[width=0.95\textwidth]{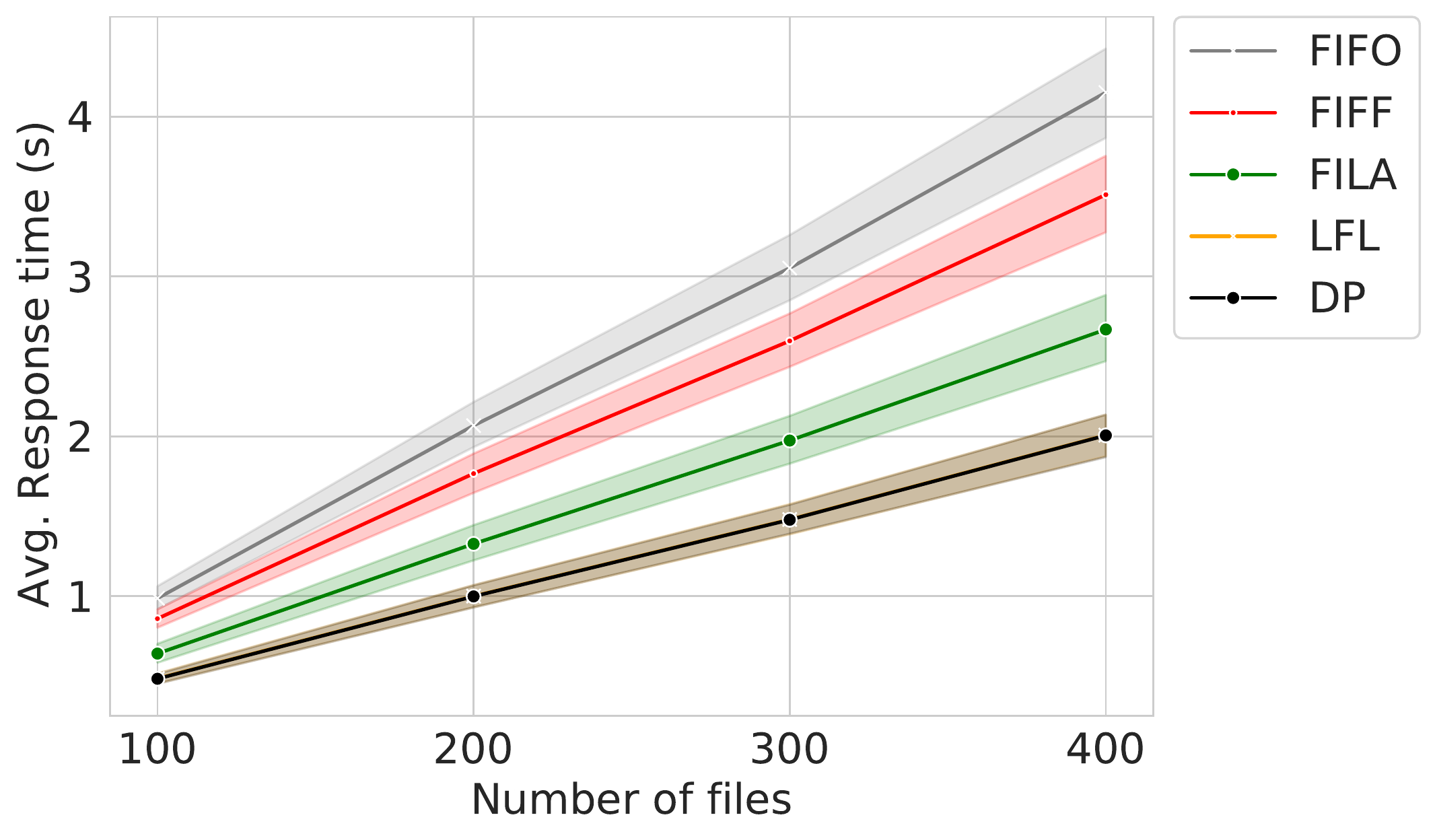}
			\caption{Average response time per file.}
		\end{subfigure}
		\begin{subfigure}[b]{0.5\textwidth}
			\includegraphics[width=1.0\textwidth]{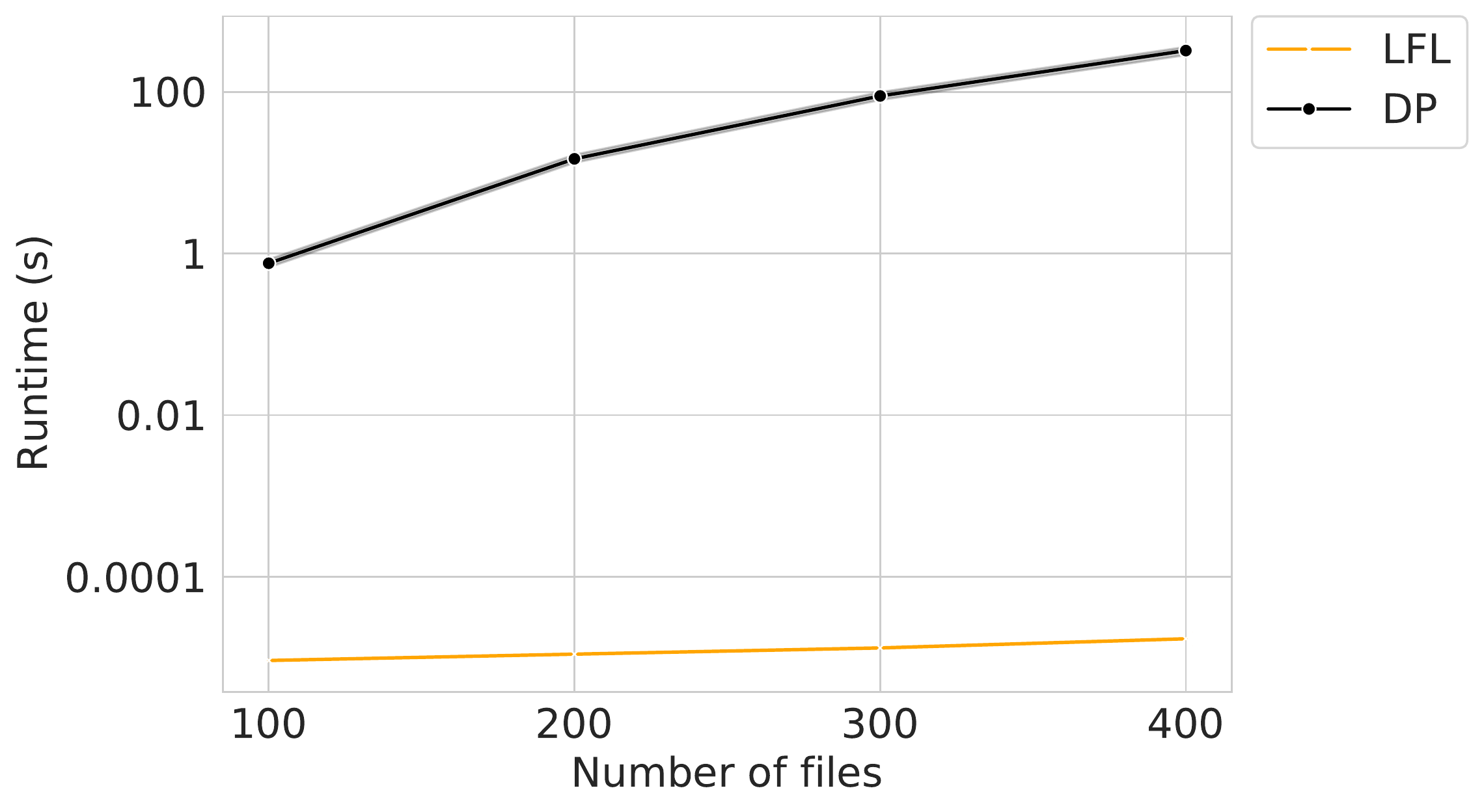}
			\caption{Runtime (in seconds, logarithmic scale).}
		\end{subfigure}
	}
	\caption{General performance on small synthetic instances (in color).}
	\label{fig:filesystem_approx}
\end{figure}

Figure~\ref{fig:filesystem_approx}-(b) illustrates the runtime of \lafila{} and \exact{}. Note that the other algorithms run in constant time and thereby are omitted from the plot. Times for \exact{} are several orders of magnitude higher than that of \lafila{}, requiring on average 107 seconds over all small instances. Thus, \exact{} is not adequate in scenarios where the number of files grow considerably larger than 100 (see also~\S\ref{sec:landsat}).

For scalability purposes, we next investigate the performance on large-scale tape sizes in Figures~\ref{fig:filesystem_approx_Large}-(a) and~\ref{fig:filesystem_approx_Large}-(b). 
\begin{figure}[ht!]	
	\centering
	\captionsetup{justification=centering}
	\scalebox{1.0}{
		\begin{subfigure}[b]{0.5\textwidth}
			\includegraphics[width=1.0\textwidth]{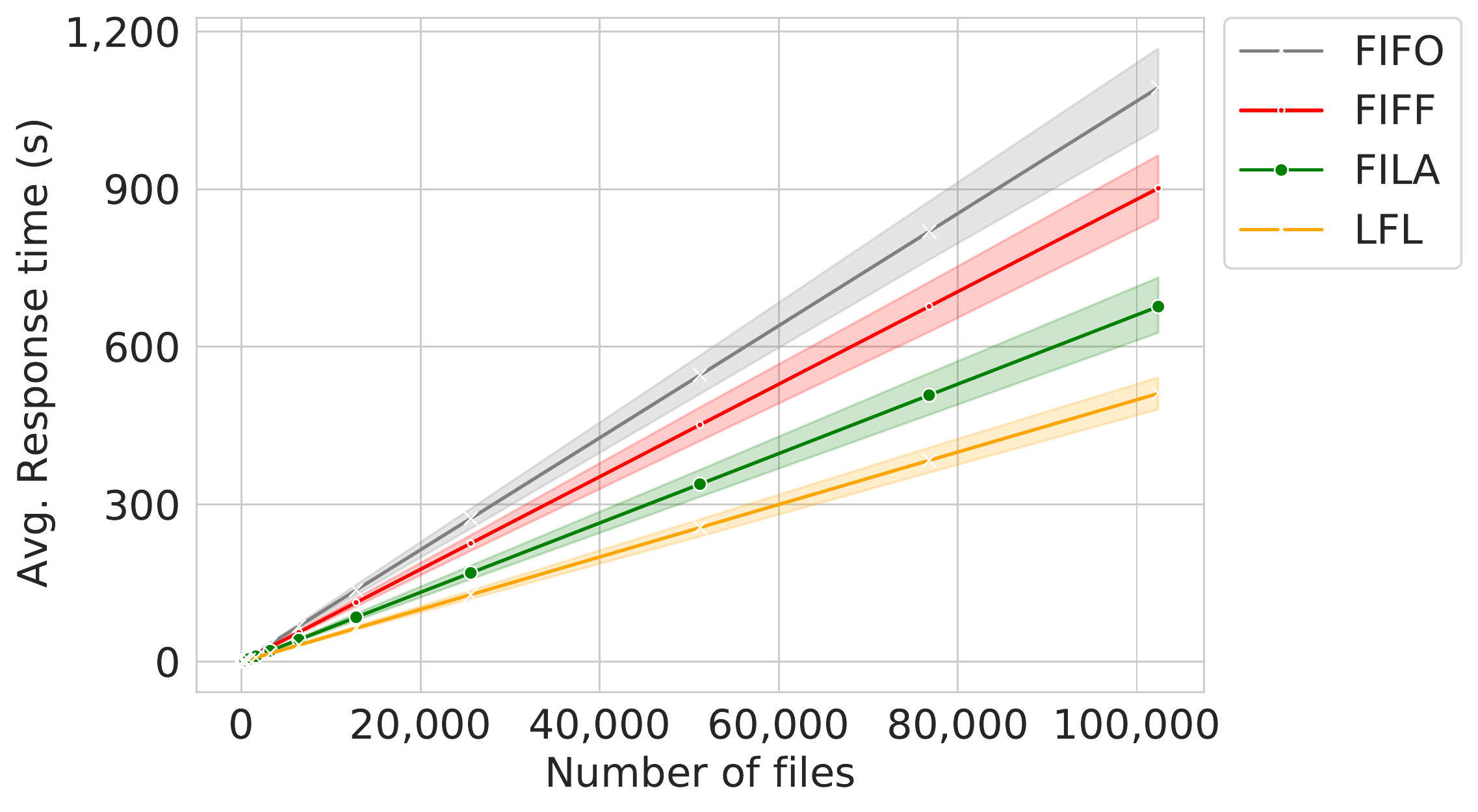}
			\caption{Average response time per file.}
		\end{subfigure}
		\begin{subfigure}[b]{0.5\textwidth}
			\includegraphics[width=0.93\textwidth]{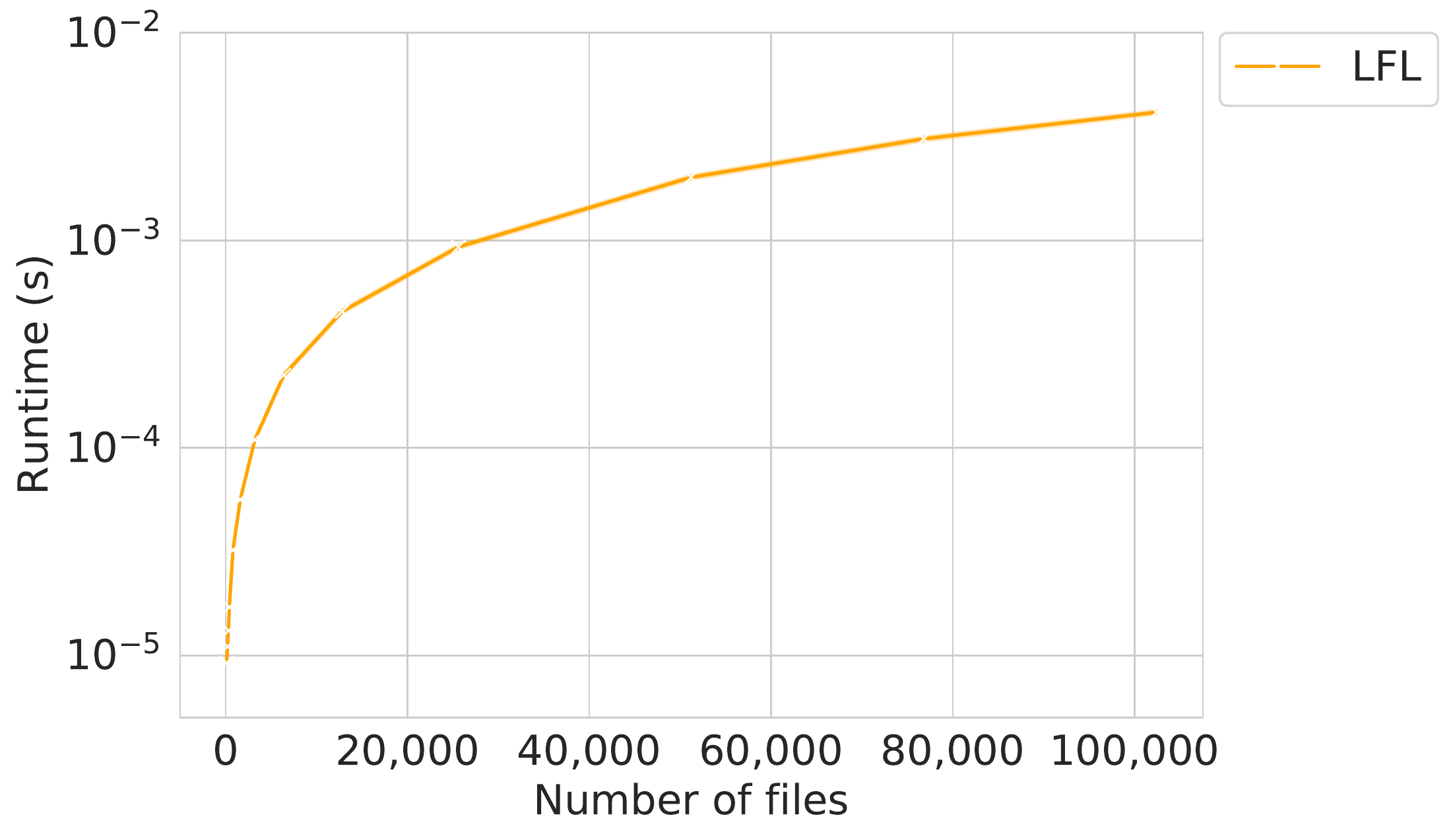}
			\caption{Runtime (in seconds, logarithmic scale).}
		\end{subfigure}
	}
	\caption{General performance on large synthetic instances (in color).}
	\label{fig:filesystem_approx_Large}
\end{figure}
%We omit~\exact{} from these experiments because it cannot be used to solve these instances due to large consumption of time and memory. 
The results in Figure~\ref{fig:filesystem_approx_Large}-(a) show that the performance trend of the approximate policies remains the same as the number of files grow. \lafila{} is consistently superior to the other policies, and its runtime (below 0.001 seconds) suggests that it is also practical. 

\paragraph{Empirical Approximation.} Figures \ref{fig:filesystem_ratio}-(a) and \ref{fig:filesystem_ratio}-(b) depict the average approximation ratio as a function of the request probability $p$ and the standard deviation parameter $\sigma$ of the log-normal distribution, respectively, on the small synthetic instances. The shaded areas represent the 95\% confidence intervals.
\begin{figure}[ht!]	
	\centering
	\captionsetup{justification=centering}
	\scalebox{1.0}{
		\begin{subfigure}[b]{0.5\textwidth}
			\includegraphics[width=1.0\textwidth]{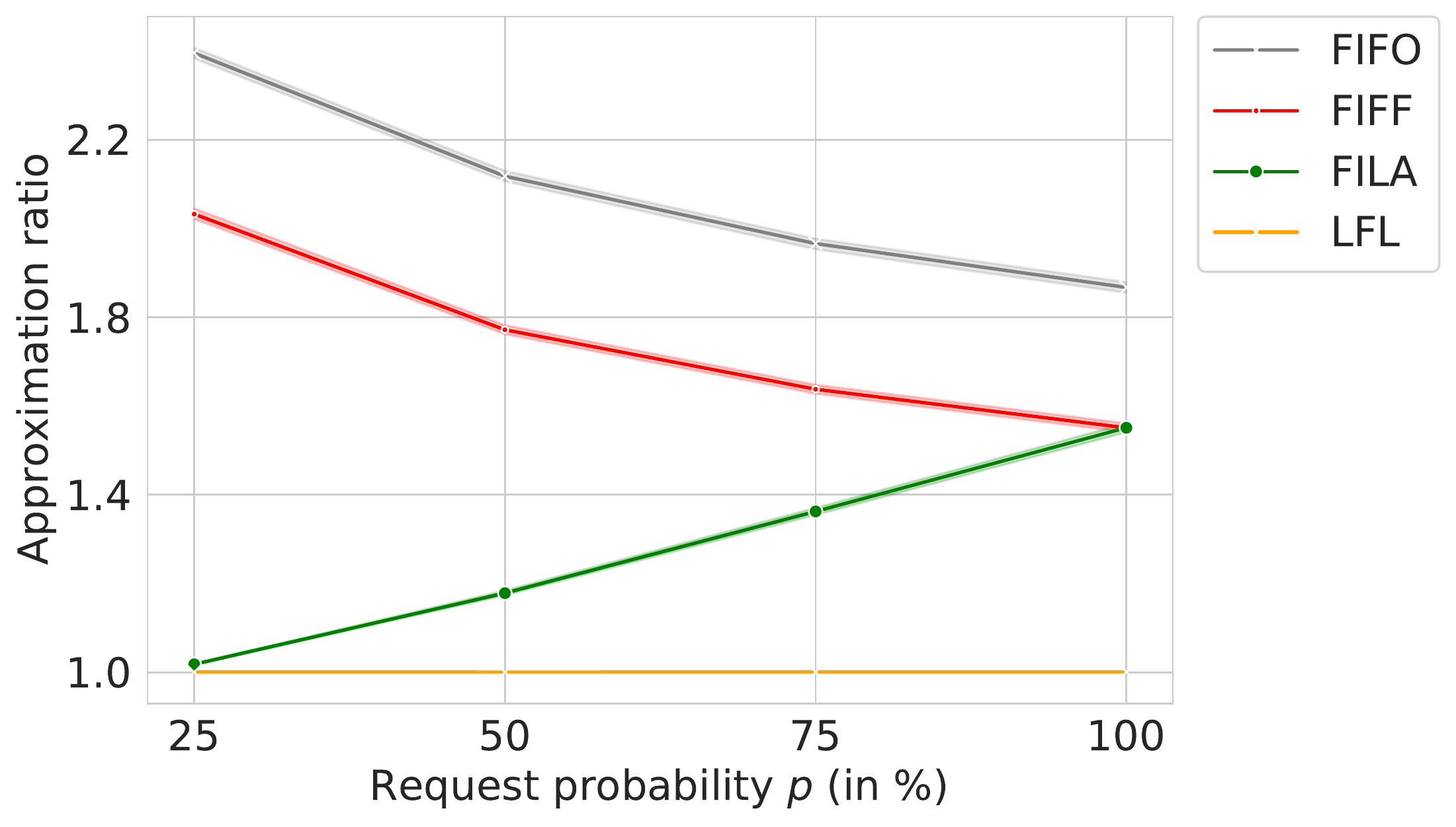}
			\caption{Performance per  request probability.}
		\end{subfigure}
		\begin{subfigure}[b]{0.5\textwidth}
			\includegraphics[width=1.0\textwidth]{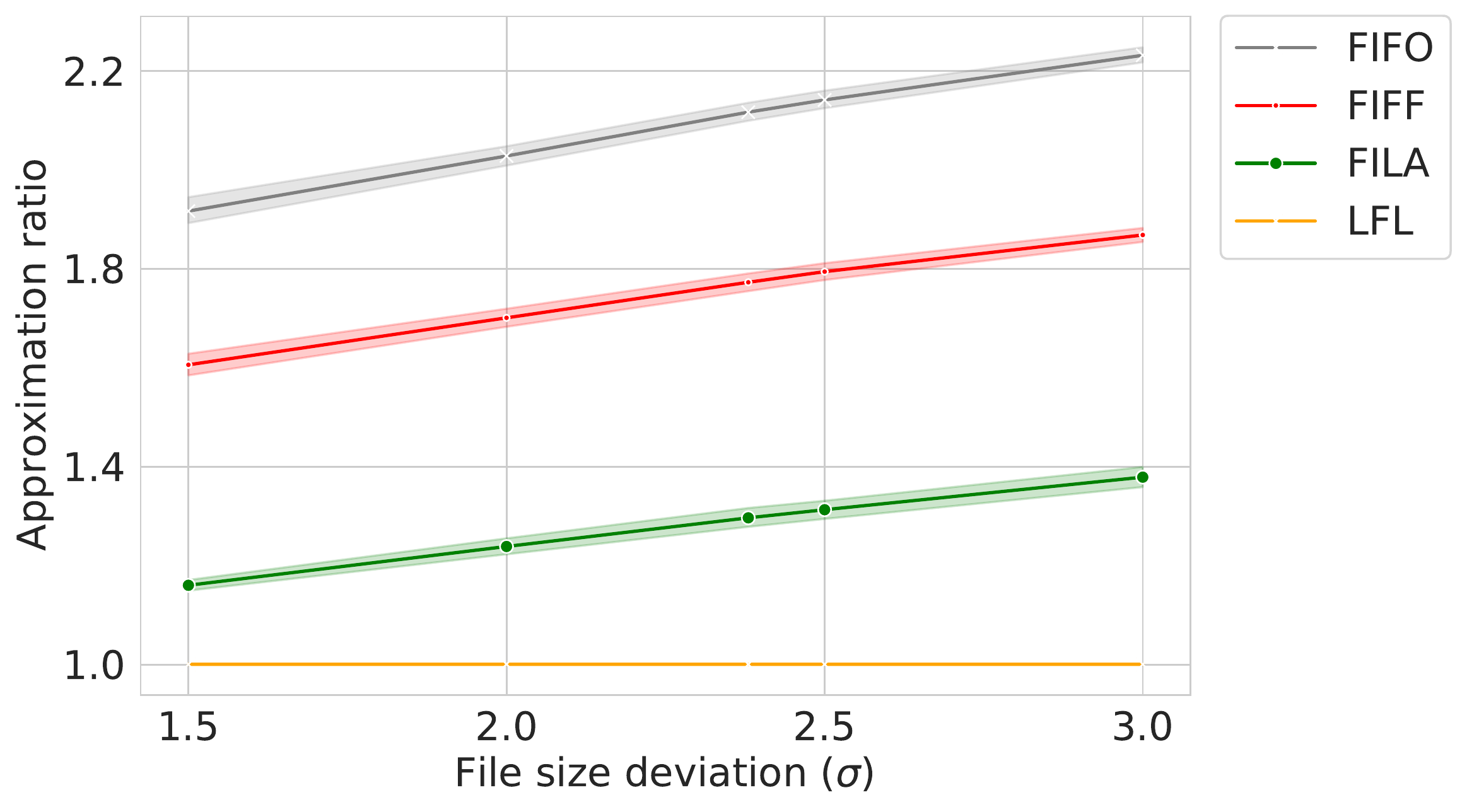}
			\caption{Performance per  file size deviation.}
		\end{subfigure}
	}
	\caption{Approximation ratio under different request probabilities $p$ and file size deviations $\sigma$ regimes (in color).}
	\label{fig:filesystem_ratio}
\end{figure}
 The values are relative to the optimal solutions obtained by \exact{}.  \fififi{} is  worse on instances with fewer requests, producing sequences having twice the optimal total response time for $p \le 0.4$. This occurs because \fififi{} postpones files to the forward stage, and thus it may take a significant amount of time to read them after reaching position $0$. \fifila{}, due to its symmetry to \fififi{}, has the opposite behavior and is preferable when there are few requests. The approximation ratio of all policies increases with the file deviation~$\sigma$. \fifila{} and \lafila{} exhibit significantly better empirical performance than their theoretical worst case. In particular, \lafila{}  has a maximum approximation ratio of approximately $1.029$ over all instances. \fifo{} is outperformed by all approximate policies.

\paragraph{Value of File Access Frequency Data.} We next investigate the value of pre-fixing sequences in the presence of distributional information captured in the \ref{model:SLTS}. We consider a single tape configuration with~$\numfiles = 250$ and file sizes drawn from a log-normal distribution with parameters $\mu = 13.04$ and $\sigma = 2.38$. For each $p \in \{0.2, 0.4, 0.6, 0.8\}$, we solve the \ref{model:SLTS} optimally assuming $\prob{\file} = p$ for all $\file \in \fileset$ and evaluate the performance of the resulting sequence $\ordertuple_p$ for 100 runs, where each file is requested at random with probability $p$. In particular, we assume that files are skipped in $\ordertuple_p$ if they are not requested in the run. Figure \ref{fig:filesystem_stochastic} compares the performance of the resulting method, dubbed \texttt{SLTS}, with~\fifo{}, \fififi{}, \fifila{}, and \lafila{} for each run. The results are aggregated by~$p$ and normalized by the optimal solution obtained using~\exact{}.
\begin{figure}[ht!]	
	\centering
	\captionsetup{justification=centering}
	\scalebox{1.02}{
		\begin{subfigure}[b]{0.475\textwidth}
			\includegraphics[width=1.0\textwidth]{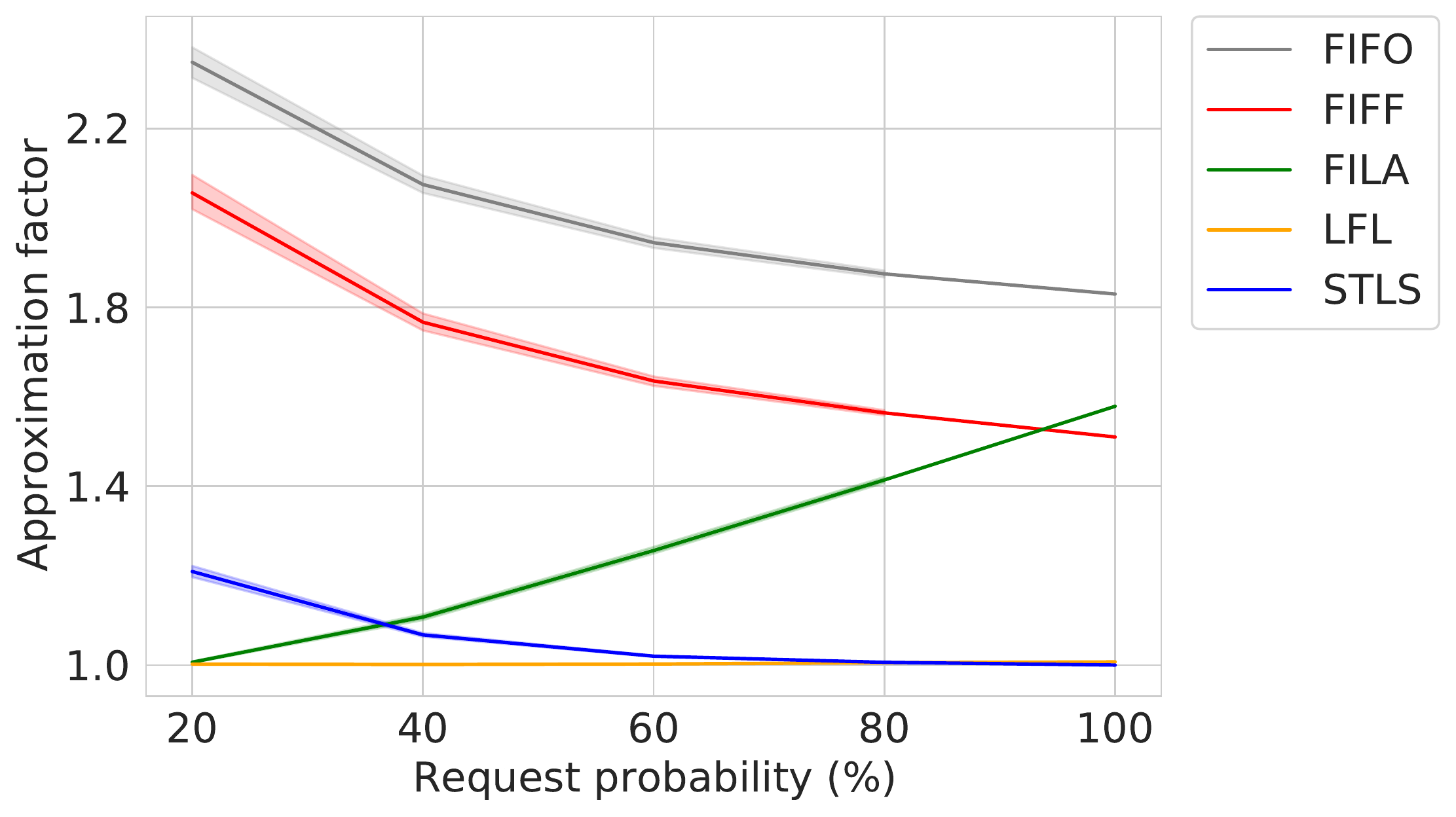}
			\caption{Performance per request probability.}
		\end{subfigure}
		\begin{subfigure}[b]{0.475\textwidth}
			\includegraphics[width=1.0\textwidth]{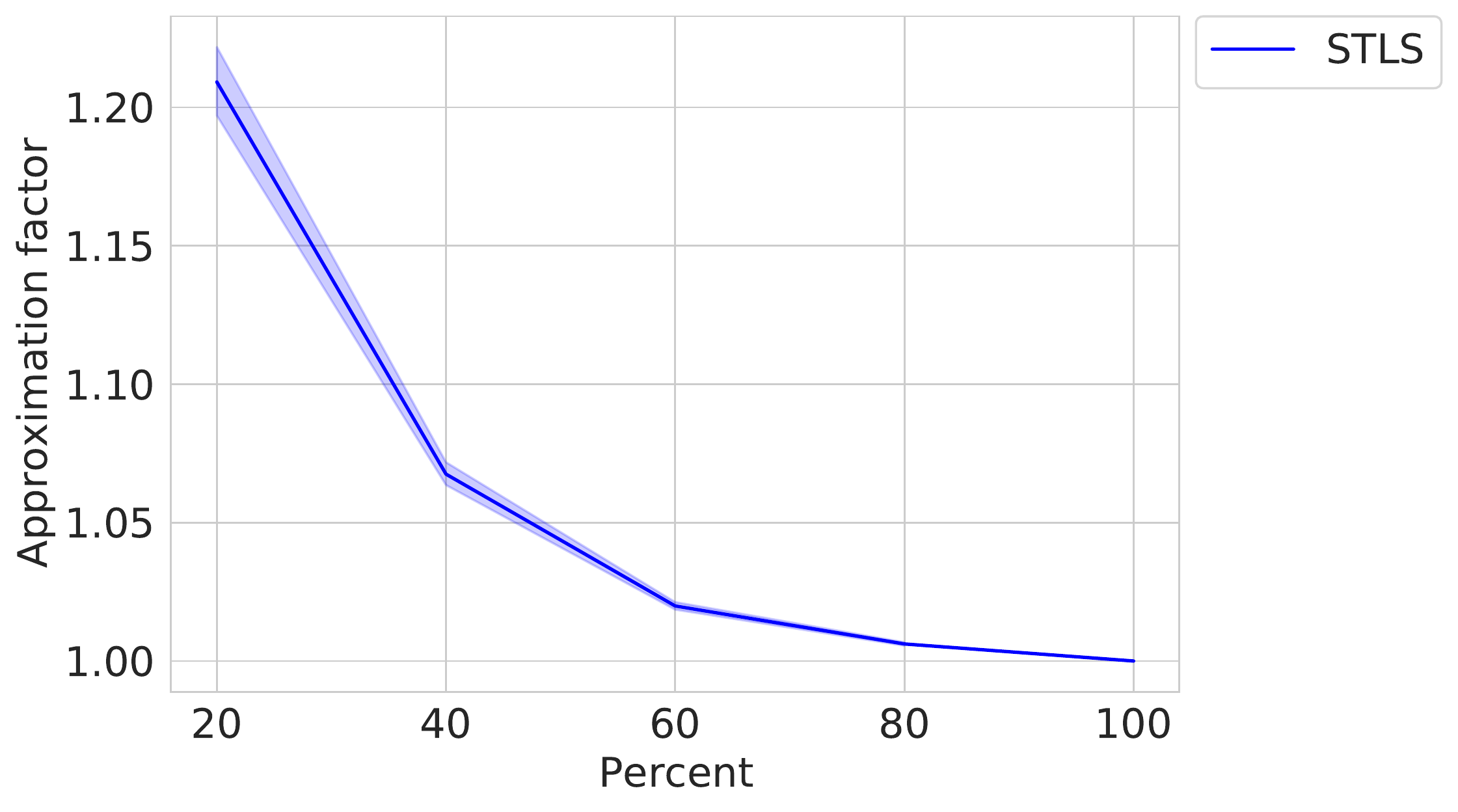}
			\caption{Performance of the \texttt{SLTS}.}
		\end{subfigure}
	}
	\caption{Approximation ratio under different request probabilities $p$ (in color).}
	\label{fig:filesystem_stochastic}
\end{figure}
The average performance of \texttt{SLTS} is superior to all existing policies at $p \ge 0.4$, where the number of requests is medium to high. For smaller values of $p$, \lafila{} outperforms \texttt{SLTS} by at most 20\%. However, Figure \ref{fig:filesystem_stochastic}-(b) suggests that larger gaps only occur for small values of $p$; otherwise, \texttt{SLTS} has strong overall performance. Notably, once $\ordertuple_p$ has been computed for the \ref{model:SLTS}, adjusting the sequence for the requested files takes negligible time, thus making it more interesting in scenarios where the computational resources are severely restricted. %The remaining policies performed similarly as before. 

\subsection{Landsat Instances}
\label{sec:landsat}

The Landsat program is a space mission that images the Earth's surface every 16 days. Landsat 8, developed as a collaboration between NASA and the U.S. Geological Survey (USGS), is the eighth of the series of satellites launched by the mission. The images measure different ranges of frequencies (called ``bands'') along the electromagnetic spectrum. Because the sensing process generates massive amounts of data, Landsat 8 data sets are split into a collection of tiles (or satellite scenes). Each tile is stored in a single file with approximately 3.5 GB of data and features 12 bands encoded by numerical pixel matrices. 
% Of those, 11 bands represent different ranges of the spectrum, and one additional band (particular to Landsat 8), known as Quality Assessment, classifies pixels with an accuracy of up to 80\% as water, snow, ice, cloud, or as invalid due to sensor errors during the image acquisition process. % \citep{Roy14}.
The files are used by machine learning algorithms to classify the health conditions of vineyards for precision viticulture purposes and include vineyards in the Atacama desert in Chile, the Serra Gaucha region in Brazil, and the Manduria region in Italy.
%  (courtesy of the U.S. Geological Survey). 

Our benchmark contains 100 instances with different configurations of the Landsat dataset. Every instance consists of 15 Landsat tiles, each composed of 12 files (and hence one file per band). The average file size is approximately 280 MB, with a small standard deviation (less than 1 MB).

\begin{figure}[ht!]
	\centering
	\begin{subfigure}[b]{0.495\textwidth}
    	\includegraphics[width=1.0\textwidth]{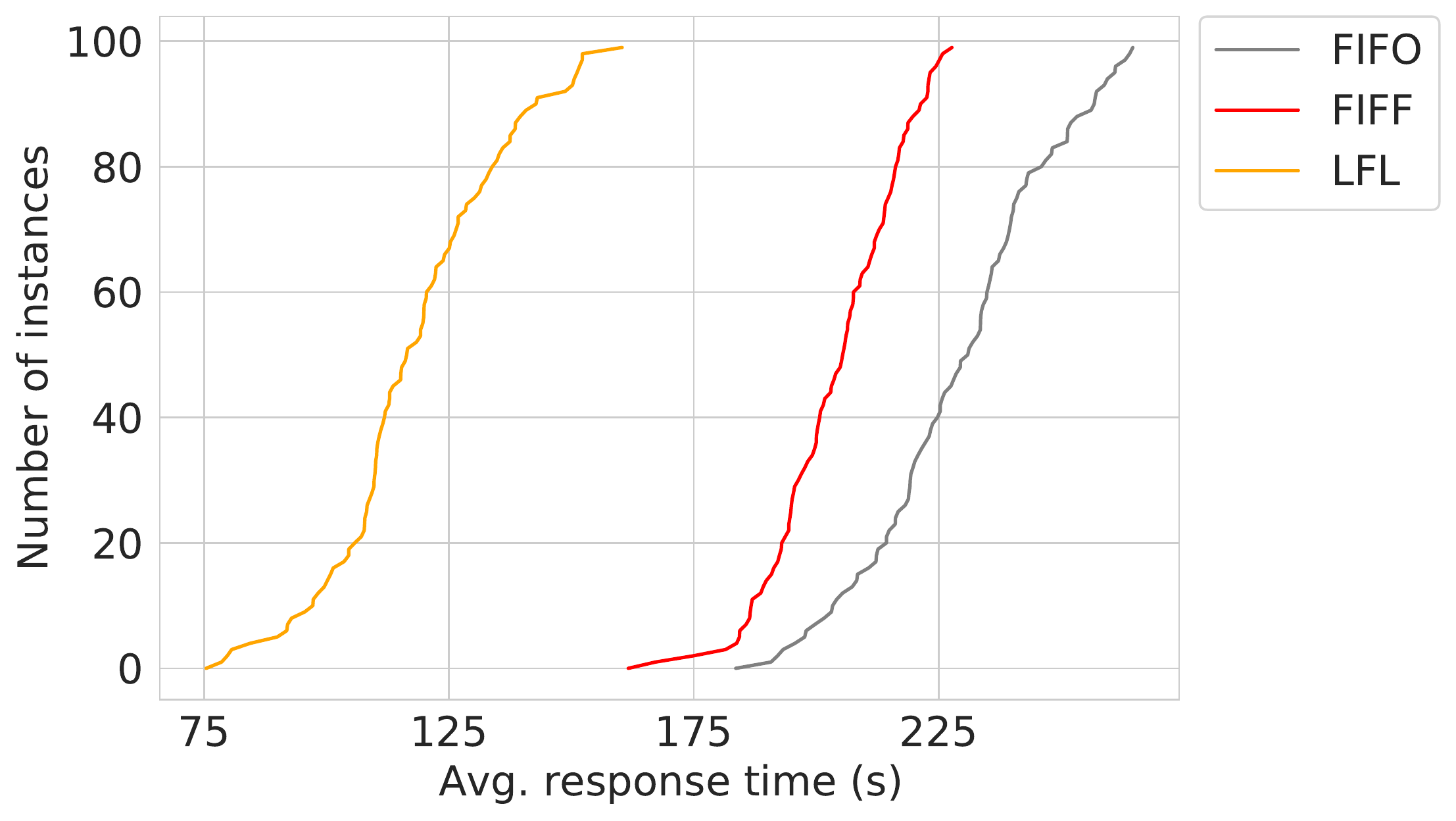}
	\end{subfigure}
	\hfill
	\begin{subfigure}[b]{0.495\textwidth}
	  	\includegraphics[width=1.0\textwidth]{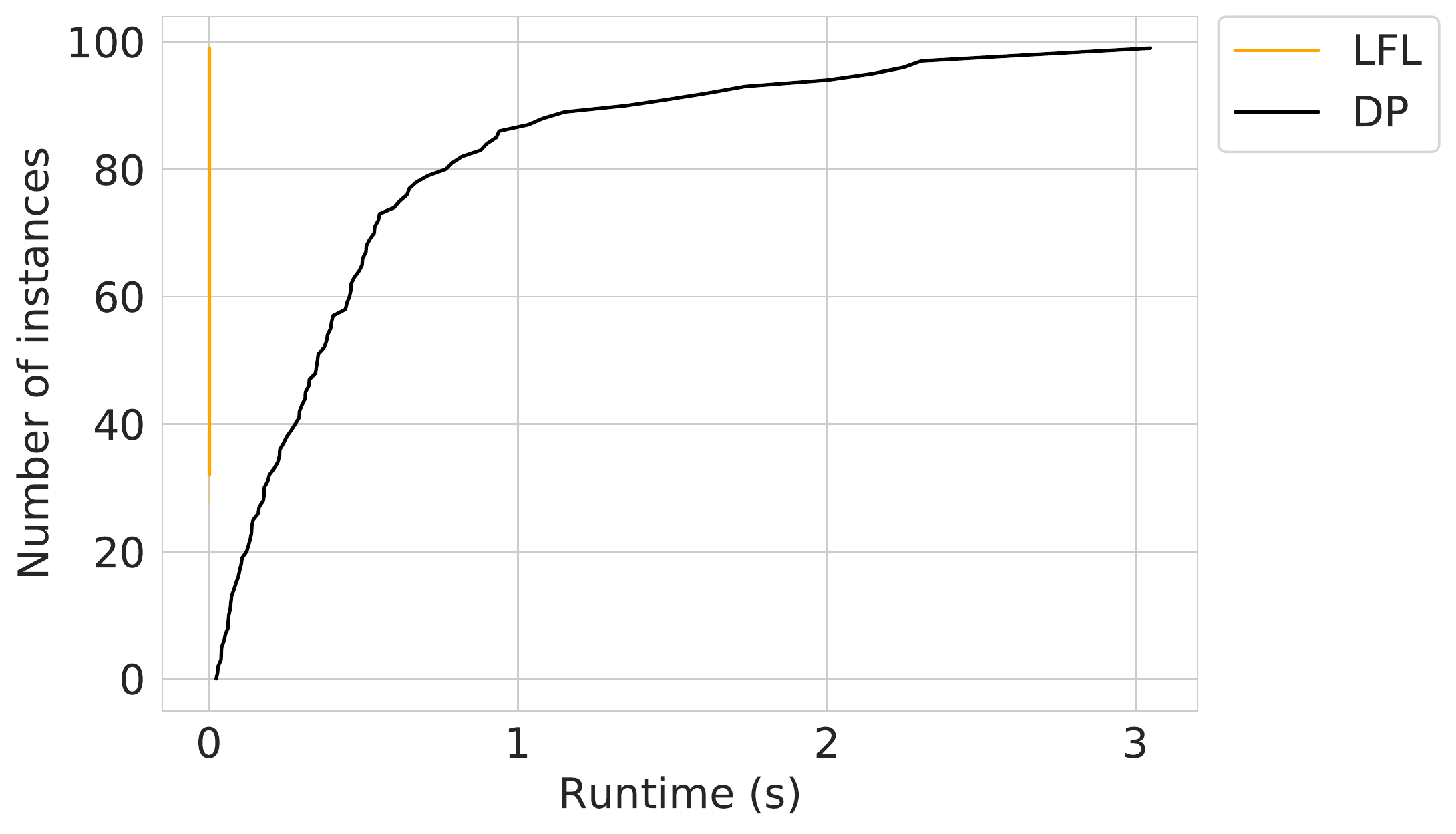}
	\end{subfigure}
	\caption{Performance on different instances of the Landsat dataset (in color).}
    \label{fig:landsat_performance}
\end{figure}

Figures \ref{fig:landsat_performance}-(a) and \ref{fig:landsat_performance}-(b) depict a cumulative performance profile of the approximation algorithms and \exact{}, respectively, on the Landsat dataset. We omit the performance plots of~\fifila{} and \exact{} in Figure \ref{fig:landsat_performance}-(a) because they are similar to~\lafila{}. In terms of the average response time, the policies can be divided into three different categories.  \fifo{} delivers the worst results, with an average response time of 230 seconds (standard deviation of 19). \fififi{} is approximately 13\% better than~\fifo{}, with an average response time of 204 seconds (standard deviation of 13). The remaining policies reduce the objective values by almost 50\%, with average response times of 118 seconds (standard deviation of 18). \exact{} improves upon the other algorithms by hundredths of seconds. The solution time curves in Figure~\ref{fig:landsat_performance}-(b) show that~\exact{} requires up to three seconds to solve these instances. The approximate  policies are fast enough to be used in practice and provide sequences under $0.02$ seconds.

% %%%%%%%%%%%%%%%%%%%%%%%%%%%%%%%%%%%%%%%%%%%%%%%%%%%%%%%%%%%%
% % Conclusions
% %%%%%%%%%%%%%%%%%%%%%%%%%%%%%%%%%%%%%%%%%%%%%%%%%%%%%%%%%%%%

\section{Conclusions}
\label{sec:conclusions}

This paper investigates policies to retrieve files stored in tapes to minimize the total response time. We consider three versions of the problem.
%incorporating deterministic and service access frequency of incoming requests. 
For the deterministic case, we present low-complexity policies with constant-factor approximation guarantees
and introduce a polynomial-time algorithm based on dynamic programming. We also develop a fully polynomial approximation scheme for a stochastic version of the problem where request probabilities are available. Finally, we study an online variant where requests arrive in real time, presenting the first constant-factor competitive algorithm. Our numerical analysis of both synthetic and real-world tape settings suggests that~\fifo{}, the standard policy in industry, is outperformed by other low-complexity methods and that policy decisions benefit when considering file access probabilities.

% Bibliography
% \newpage 
\bibliographystyle{plainnat}
\bibliography{references}

\newpage

\appendix%%%%%%%%%%%%%%%%%%%
%%%%%%%%%%%%%
%%%%%%%%%%%%

\section{Example}
\label{sec:example}

% Consider the following tape:
% Reading operations start at the last (rightmost) file. Solid and dashed arcs depict repositioning and reading operations, respectively. Labels represent the distance traversed when the movement is finished.

Consider a tape with five files, sizes $(\filesize{1}, \dots,  \filesize{5}) = (2,2,8,2,1)$, and tape velocity of $\nu = 1$ bit per time unit. 
The tape in Figure \ref{fig:exampleECScheduleA} depicts the sequence $\pi = (5,4,1,2,3)$, where solid (colored) arcs denote tape repositioning  and dashed (black) arcs denote file reading operations as in Figure \ref{fig:exampleSchedule1}. For instance, to read file 4 (the second in the ordering), the tape header finishes reading file 1 at time 2, repositions itself to the left of file 4 at time 9 (since it must traverse boths files 4 and 5), and finishes reading the file at time 15. The response time of file $4$ is the  time at which the tape head first started reading it, i.e., $9$. Note also that there is no need to reposition the tape header when reading file $2$ after file $1$. The total response time of $\pi$ is $\response_1(\pi) + \dots + \response_5(\pi) = 1+5+21+23+25 = 75$.

\begin{figure}[h!]
	\centering
	\scalebox{0.9}{
		\usetikzlibrary{arrows,backgrounds,snakes}
		\begin{tikzpicture}
			[
				rewindfile/.style={rectangle, draw=black, fill=black!2, thick, minimum size=5mm},
				forwardfile/.style={rectangle, draw=purple, fill=purple!2, thick, minimum size=5mm, text=purple},
				file/.style={rectangle, draw=black, fill=gray!5, thick, minimum size=5mm, text=black},
			]

			% Files
			\node[file, minimum width=22mm] (f1) {1};
			\node[file, minimum width=22mm, right=0.5mm of f1] (f2) {2};
			\node[file, minimum width=50mm, right=0.5mm of f2, node distance=50pt] (f3) {3};
			\node[file, minimum width=26mm, right=0.5mm of f3, node distance=50pt] (f4) {4};
			\node[file, minimum width=14mm, right=0.5mm of f4, node distance=50pt] (f5) {5};
			
			% File sizes
			\newcommand\drawsize[2] {
				\coordinate[above = 6mm of #1.west] (c1a#1);
				\coordinate[above = 6mm of #1.east] (c1b#1);
				\draw[<->] (c1a#1) -- (c1b#1) node[midway, fill=white, font=\scriptsize] {#2};
			}
			\drawsize{f1}{$\filesize{1} = 2$};
			\drawsize{f2}{$\filesize{2} = 2$};
			\drawsize{f3}{$\filesize{3} = 8$};
			\drawsize{f4}{$\filesize{4} = 2$};
			\drawsize{f5}{$\filesize{5} = 1$};

			% Sequence - back and forth
			\newcommand\drawsequence[6] {
				\coordinate[below = #4mm of #1.east] (d1a#1#2);
				\coordinate[below = #4mm of #2.west] (d1b#1#2);
				\draw[->, line width=0.3mm,color=orange] (d1a#1#2) -- (d1b#1#2);
				\node[font=\normalsize, color=orange, right=1mm of d1a#1#2.west] {#6};
				\coordinate[below = #5mm of #2.west] (d2a#1#2);
				\coordinate[below = #5mm of #2.east] (d2b#1#2);
				\draw[->, dashed, line width=0.5mm] (d2a#1#2) -- (d2b#1#2);
				\node[font=\normalsize, right=1mm of d2b#1#2.west] {#3};
			}
			% Forward stage - just draw dashed arcs
			\newcommand\drawsequenceforwardfile[6] {
				\coordinate[below = #5mm of #2.west] (d2a#1#2);
				\coordinate[below = #5mm of #2.east] (d2b#1#2);
				\draw[->, dashed, line width=0.5mm] (d2a#1#2) -- (d2b#1#2);
				\node[font=\normalsize, right=1mm of d2b#1#2.west] {#3};
			}

			% \drawsequence{first file}{last file}{response when returning}{position first}{position last}{response time for the first file};
			\drawsequence{f5}{f5}{2}{6}{12}{1};
			\drawsequence{f5}{f4}{7}{18}{24}{5};
			\drawsequence{f4}{f1}{23}{30}{36}{21};
			\drawsequenceforwardfile{f1}{f2}{25}{42}{42}{-1};
			\drawsequenceforwardfile{f2}{f3}{33}{48}{48}{-1};

			% Reading sequence arrow
			\coordinate[below = 5mm of f5.east, xshift=+12mm] (p1);
			\coordinate[below = 50mm of f5.east, xshift=+12mm] (p2);
			\draw[->] (p1) -- (p2) node[midway,above,rotate=270] {Reading sequence};

		\end{tikzpicture}
	}
	\caption{Example of tape movement for the sequence $\ordertuple = (5,4,1,2,3)$. 
	}
	\label{fig:exampleECScheduleA}	
\end{figure}
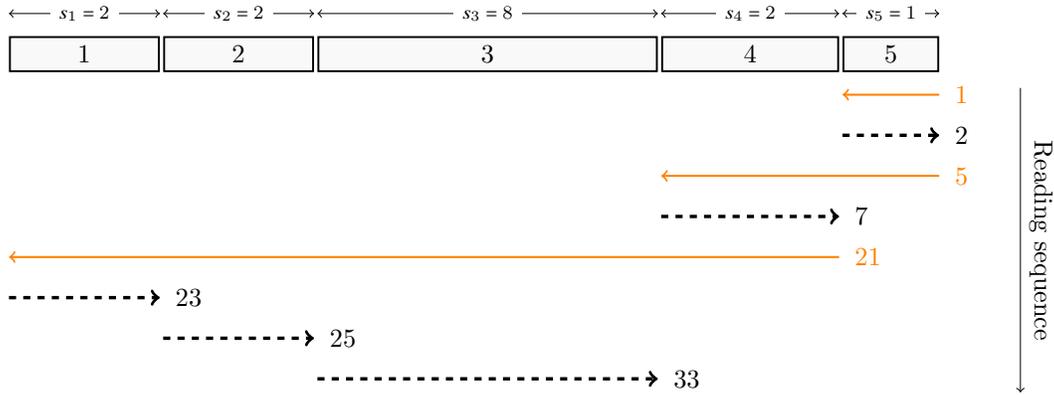
Figure \ref{fig:exampleECScheduleB} illustrates a simpler policy (First-File-First, or \fififi{}) that moves the tape head to the beginning of the tape and reads files in sequence, i.e., $\ordertuple' = (1,2,3,4,5)$. Its intuition is that the tape only requires to do a single left-to-right movement. The total response time of $\pi'$ is $15+17+19+27+29 = 107$, i.e., 42\% higher than from the sequence above. The difference is due to the presence of a large file with size 8, which must be traversed twice prior to reading files 4 and 5 in $\ordertuple'$, increasing the response time. This is the basis of worst-case scenarios investigated in \S\ref{sec:deterministic}.
\begin{figure}[h!]
	\centering
	\scalebox{0.9}{
		\usetikzlibrary{arrows,backgrounds,snakes}
		\begin{tikzpicture}
			[
				rewindfile/.style={rectangle, draw=black, fill=black!2, thick, minimum size=5mm},
				forwardfile/.style={rectangle, draw=purple, fill=purple!2, thick, minimum size=5mm, text=purple},
				file/.style={rectangle, draw=black, fill=gray!5, thick, minimum size=5mm, text=black},
			]
			% Files
			\node[file, minimum width=22mm] (f1) {1};
			\node[file, minimum width=22mm, right=0.5mm of f1] (f2) {2};
			\node[file, minimum width=50mm, right=0.5mm of f2, node distance=50pt] (f3) {3};
			\node[file, minimum width=26mm, right=0.5mm of f3, node distance=50pt] (f4) {4};
			\node[file, minimum width=14mm, right=0.5mm of f4, node distance=50pt] (f5) {5};
			
			% File sizes
			\newcommand\drawsize[2] {
				\coordinate[above = 6mm of #1.west] (c1a#1);
				\coordinate[above = 6mm of #1.east] (c1b#1);
				\draw[<->] (c1a#1) -- (c1b#1) node[midway, fill=white, font=\scriptsize] {#2};
			}
			\drawsize{f1}{$\filesize{1} = 2$};
			\drawsize{f2}{$\filesize{2} = 2$};
			\drawsize{f3}{$\filesize{3} = 8$};
			\drawsize{f4}{$\filesize{4} = 2$};
			\drawsize{f5}{$\filesize{5} = 1$};

			% Sequence - back and forth
			\newcommand\drawsequence[6] {
				\coordinate[below = #4mm of #1.east] (d1a#1#2);
				\coordinate[below = #4mm of #2.west] (d1b#1#2);
				\draw[->, line width=0.3mm,color=orange] (d1a#1#2) -- (d1b#1#2);
				\node[font=\normalsize, color=orange, right=1mm of d1a#1#2.west] {#6};
				\coordinate[below = #5mm of #2.west] (d2a#1#2);
				\coordinate[below = #5mm of #2.east] (d2b#1#2);
				\draw[->, dashed, line width=0.5mm] (d2a#1#2) -- (d2b#1#2);
				\node[font=\normalsize, right=1mm of d2b#1#2.west] {#3};
			}
			% Forward stage - just draw dashed arcs
			\newcommand\drawsequenceforwardfile[6] {
				\coordinate[below = #5mm of #2.west] (d2a#1#2);
				\coordinate[below = #5mm of #2.east] (d2b#1#2);
				\draw[->, dashed, line width=0.5mm] (d2a#1#2) -- (d2b#1#2);
				\node[font=\normalsize, right=1mm of d2b#1#2.west] {#3};
			}

			% \drawsequence{first file}{last file}{response when returning}{position first}{position last}{response time for the first file};
			\drawsequence{f5}{f1}{17}{6}{12}{15};
			\drawsequenceforwardfile{f1}{f2}{19}{-1}{18}{-1};
			\drawsequenceforwardfile{f2}{f3}{27}{-1}{24}{-1};
			\drawsequenceforwardfile{f3}{f4}{29}{-1}{32}{-1};

			% Reading sequence arrow
			\coordinate[below = 5mm of f5.east, xshift=+12mm] (p1);
			\coordinate[below = 40mm of f5.east, xshift=+12mm] (p2);
			\draw[->] (p1) -- (p2) node[midway,above,rotate=270] {Reading sequence};

		\end{tikzpicture}
	}
	\caption{Example of tape movement for the sequence $\ordertuple' = (1,2,3,4,5)$. 
	}
	\label{fig:exampleECScheduleB}
\end{figure}
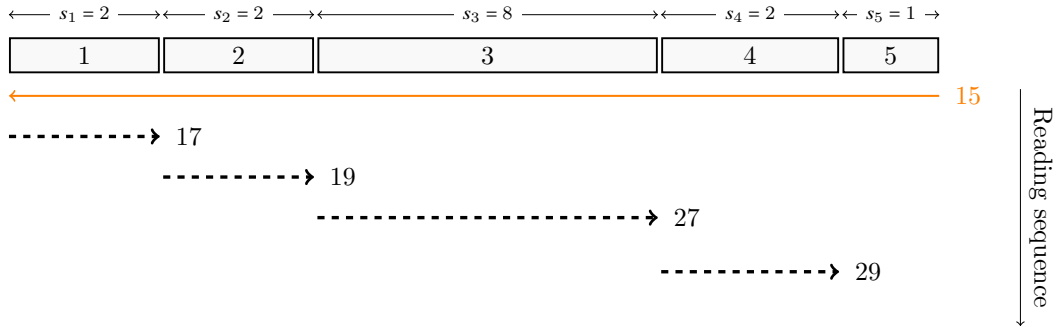

\section{Tables and Complementary Numerical Results}
\label{sec:appce}

\subsection{Synthethic Instances}\label{sec:tables_filesystem}

Tables~\ref{tab:fifo}, \ref{tab:fififi}, \ref{tab:fifila}, \ref{tab:lafila}, and~\ref{tab:exact} 
present the results of~\fifo{}, \fififi{}, \fifila{}, \lafila{}, and~\exact{}, respectively, for synthethic instances.  An entry corresponds to the average response time (in seconds, truncated to two decimal places) and the standard deviation (truncated to one decimal place) based on each configuration $(p,\sigma)$, where~$p$ is the probability that a file in the tape is requested and~$\sigma$ is the parameter of the log-normal distribution for file sizes (with $\mu = 13.04$ for all cases). Each value is the mean of the 10 instances for the associated configuration. Due to memory and time limitations, we conducted experiments with~\exact{} using only small instances (with up to 400 files), so  Table~\ref{tab:exact} has fewer columns than the others.

\begin{landscape}
\begin{table}[]
\centering
% \footnotesize
% \scriptsize
\caption{Average response times - \fifo{} (in seconds).}
\label{tab:fifo}
\tiny
\resizebox{\columnwidth}{!}{%
\begin{tabular}{ll|rrrr|rrrrrrrrr}
\toprule
    \multicolumn{ 2}{c|}{\textbf{Parameters}} &               \multicolumn{ 4}{c|}{{\bf Small Instances }}  &
    \multicolumn{ 9}{c}{{\bf Large Instances }}
    \\
    \hline
$p (\%)$  & $\sigma$ &  100    &  200    &  300    &  400    &  800    &  1,600   &  3,200   &  6,400   &  12,800  &  25,600  &  51,200  &  76,800 &  102,400 \\
\midrule
25 &  1.50 &  0.4 (0.0) & 0.84 (0.1) & 1.29 (0.1) & 1.74 (0.1) &  3.57 (0.2) &  7.13 (0.3) & 14.69 (0.3) &  28.76 (0.8) &  57.98 (0.9) &  116.48 (1.3) &   348.36 (2.9) &   234.08 (3.3) &   465.83 (6.1) \\
     25 &  2.00 & 0.64 (0.1) &  1.3 (0.2) &  1.9 (0.1) &  2.6 (0.3) &  5.68 (0.2) & 11.04 (0.2) & 22.69 (0.7) &  44.57 (1.6) &  90.25 (1.2) &  180.28 (3.2) &    541.9 (6.8) &    361.5 (4.0) &   725.54 (6.0) \\
     25 &  2.38 & 0.87 (0.2) & 1.85 (0.3) & 3.03 (0.3) & 3.94 (0.4) &  7.92 (0.4) & 16.25 (0.3) & 32.59 (0.8) &  64.99 (0.9) & 132.47 (2.6) &  266.76 (4.1) &   787.18 (7.5) &   525.47 (7.7) & 1050.55 (19.8) \\
     25 &  2.50 &  1.0 (0.2) & 2.26 (0.3) & 3.25 (0.3) & 4.38 (0.4) &  9.18 (0.8) &  18.5 (1.1) & 36.96 (1.7) &  73.64 (1.6) & 148.52 (2.2) &  296.44 (3.3) &  886.47 (10.5) &   594.51 (7.0) &  1191.59 (9.4) \\
     25 &  3.00 & 1.69 (0.2) &  3.7 (0.5) & 5.42 (0.5) & 7.38 (0.5) & 15.47 (0.7) & 30.81 (1.8) & 63.74 (2.5) & 125.09 (3.1) & 251.81 (6.1) & 505.91 (10.0) & 1522.17 (22.2) & 1016.58 (16.2) & 2040.57 (23.3) \\
     50 &  1.50 & 0.43 (0.0) & 0.86 (0.1) &  1.3 (0.1) & 1.83 (0.1) &  3.66 (0.2) &  7.35 (0.3) & 14.72 (0.4) &  29.01 (0.4) &  58.36 (1.0) &  116.16 (2.1) &   349.69 (3.3) &   232.06 (1.7) &   463.81 (6.1) \\
     50 &  2.00 & 0.66 (0.1) & 1.39 (0.1) & 2.11 (0.3) &  2.8 (0.2) &  5.56 (0.3) & 11.25 (0.4) & 22.59 (0.5) &  45.16 (0.8) &  90.28 (1.5) &   180.5 (2.7) &   541.76 (6.8) &   363.03 (4.5) &   728.73 (5.7) \\
     50 &  2.38 & 0.93 (0.2) & 2.02 (0.2) & 2.95 (0.3) & 3.97 (0.4) &  8.01 (0.5) & 15.58 (0.7) & 32.58 (0.5) &  66.04 (0.6) & 132.32 (2.2) &  261.59 (6.0) &   796.13 (6.5) &   529.41 (4.7) &  1060.76 (7.4) \\
     50 &  2.50 & 1.11 (0.1) & 2.11 (0.2) & 3.57 (0.1) & 4.55 (0.3) &  9.12 (0.4) & 18.54 (0.7) &  36.9 (1.4) &  74.36 (1.4) & 149.45 (2.2) &  297.37 (4.8) &  892.11 (10.1) &    592.5 (8.2) &  1179.89 (9.6) \\
     50 &  3.00 & 1.77 (0.3) & 3.77 (0.5) & 5.33 (0.5) & 7.54 (0.7) & 15.48 (1.0) &  30.5 (1.6) & 62.69 (2.8) & 125.52 (3.6) & 253.55 (6.6) &  508.52 (7.4) & 1526.99 (13.7) & 1007.36 (13.6) & 2031.28 (17.9) \\
     75 &  1.50 & 0.43 (0.1) & 0.91 (0.0) & 1.37 (0.1) & 1.79 (0.1) &  3.57 (0.1) &  7.32 (0.3) & 14.23 (0.4) &  29.27 (0.4) &   58.3 (0.7) &   117.1 (1.7) &   348.38 (3.7) &   232.02 (3.0) &   468.05 (4.5) \\
     75 &  2.00 & 0.67 (0.1) & 1.33 (0.1) & 2.08 (0.2) & 2.84 (0.2) &  5.65 (0.3) & 11.17 (0.4) & 22.47 (0.4) &  44.85 (0.6) &  90.49 (1.5) &  180.61 (1.8) &   542.67 (7.4) &   362.46 (5.1) &   725.88 (4.4) \\
     75 &  2.38 & 1.02 (0.2) & 2.06 (0.3) & 2.95 (0.2) & 4.07 (0.3) &  8.01 (0.4) & 16.11 (0.3) & 32.53 (1.2) &  65.49 (1.5) & 131.51 (1.8) &  262.44 (4.0) &  785.32 (13.3) &   525.57 (3.8) & 1051.15 (11.1) \\
     75 &  2.50 & 1.13 (0.2) & 2.28 (0.3) & 3.29 (0.3) & 4.66 (0.5) &  9.57 (0.6) & 18.78 (0.6) & 37.19 (1.2) &  74.57 (1.3) & 149.44 (1.5) &  300.69 (2.9) &  891.19 (10.6) &   593.17 (8.5) & 1190.18 (14.3) \\
     75 &  3.00 &  1.9 (0.3) & 4.11 (0.4) & 5.73 (0.6) & 8.14 (0.8) & 15.24 (1.0) & 31.12 (0.7) & 62.38 (2.3) & 126.82 (3.5) & 253.87 (4.2) &  504.67 (4.8) & 1526.21 (17.8) &  1014.08 (7.8) & 2026.82 (17.8) \\
    100 &  1.50 & 0.44 (0.1) &  0.9 (0.1) & 1.36 (0.1) & 1.76 (0.1) &  3.62 (0.2) &  7.41 (0.2) & 14.37 (0.4) &  29.32 (0.5) &  57.95 (0.8) &  116.42 (1.1) &   349.04 (3.3) &   235.01 (1.3) &   465.53 (4.5) \\
    100 &  2.00 & 0.68 (0.1) & 1.41 (0.1) & 2.09 (0.2) & 2.84 (0.2) &  5.63 (0.3) & 11.38 (0.4) & 22.17 (0.4) &  45.46 (0.7) &  90.17 (1.2) &  181.82 (2.2) &   546.38 (7.3) &    361.2 (4.4) &   723.95 (7.0) \\
    100 &  2.38 & 1.03 (0.2) & 1.98 (0.3) & 2.86 (0.3) & 3.97 (0.3) &  8.01 (0.4) & 16.49 (0.8) & 32.56 (1.7) &  65.45 (1.3) & 133.01 (2.7) &  263.42 (4.2) &    792.1 (7.0) &   527.21 (8.3) &  1059.02 (8.5) \\
    100 &  2.50 & 1.07 (0.2) & 2.18 (0.2) & 3.42 (0.3) & 4.61 (0.3) &  8.96 (0.5) & 18.89 (0.7) & 37.09 (1.4) &   74.4 (1.4) & 148.98 (3.7) &  296.87 (4.3) &  891.42 (11.7) &   596.85 (9.2) &  1197.52 (9.7) \\
    100 &  3.00 & 1.88 (0.4) & 4.06 (0.6) & 5.75 (0.3) & 7.63 (0.8) & 15.88 (0.7) & 32.56 (1.3) & 63.17 (2.7) & 126.99 (3.7) & 256.24 (6.0) &  505.43 (7.5) &  1522.9 (17.4) & 1015.09 (13.3) & 2022.33 (16.2) \\
\bottomrule
\end{tabular}
}
\end{table}

\begin{table}[ht]
\centering
\caption{Average response time - \fififi{} (in seconds).}
\label{tab:fififi}
\tiny
\resizebox{\columnwidth}{!}{%
\begin{tabular}{ll|rrrr|rrrrrrrrr}
\toprule
    \multicolumn{ 2}{c|}{\textbf{Parameters}} &               \multicolumn{ 4}{c|}{{\bf Small Instances }}  &
    \multicolumn{ 9}{c}{{\bf Large Instances }}
    \\
    \hline
$p (\%)$  & $\sigma$ &  100    &  200    &  300    &  400    &  800    &  1,600   &  3,200   &  6,400   &  12,800  &  25,600  &  51,200  &  76,800 &  102,400 \\
\midrule
25 &  1.50 & 0.36 (0.0) & 0.75 (0.1) & 1.12 (0.1) & 1.53 (0.1) &  3.04 (0.1) &  6.04 (0.2) & 12.13 (0.2) &  23.81 (0.5) &   48.2 (0.5) &  95.97 (0.5) &   288.13 (1.3) & 192.76 (1.1) &   386.12 (0.8) \\
     25 &  2.00 &  0.6 (0.1) & 1.15 (0.1) & 1.67 (0.1) & 2.27 (0.2) &  4.79 (0.2) &  9.26 (0.3) & 18.68 (0.4) &  37.42 (0.9) &  74.59 (0.9) & 149.37 (1.6) &   448.52 (3.1) & 298.63 (1.8) &   597.95 (1.9) \\
     25 &  2.38 & 0.81 (0.1) & 1.64 (0.2) & 2.62 (0.3) & 3.31 (0.3) &  6.85 (0.3) & 13.52 (0.3) & 27.16 (0.7) &  53.98 (0.6) & 108.92 (1.5) & 219.17 (2.2) &   651.24 (3.3) & 436.18 (2.9) &   868.26 (4.3) \\
     25 &  2.50 & 0.92 (0.2) & 1.97 (0.3) & 2.89 (0.3) & 3.77 (0.4) &  7.75 (0.7) & 15.52 (1.0) & 30.42 (0.9) &  61.14 (1.2) &  123.6 (1.7) & 245.32 (2.8) &   737.22 (4.9) & 491.74 (3.6) &    982.5 (3.9) \\
     25 &  3.00 &  1.5 (0.2) & 3.21 (0.3) & 4.79 (0.6) &  6.4 (0.5) & 13.11 (0.5) & 25.96 (1.1) & 53.38 (1.8) &  103.9 (2.2) & 209.67 (5.1) & 418.97 (5.1) & 1257.22 (12.2) & 839.29 (7.1) & 1676.04 (11.0) \\
     50 &  1.50 & 0.37 (0.0) & 0.73 (0.0) &  1.1 (0.1) & 1.53 (0.1) &  3.03 (0.1) &  6.12 (0.2) &  12.2 (0.3) &  23.93 (0.4) &  48.17 (0.4) &  96.28 (0.8) &   287.91 (1.1) & 192.27 (1.3) &   384.85 (1.6) \\
     50 &  2.00 & 0.59 (0.1) & 1.18 (0.1) & 1.83 (0.2) & 2.34 (0.2) &  4.61 (0.3) &  9.36 (0.2) & 18.81 (0.4) &  37.22 (0.6) &  74.64 (0.5) & 149.37 (1.3) &   447.71 (2.6) & 300.06 (1.5) &   598.31 (3.4) \\
     50 &  2.38 & 0.81 (0.1) & 1.75 (0.1) & 2.58 (0.3) & 3.39 (0.4) &  6.71 (0.4) & 13.06 (0.6) & 26.96 (0.6) &  54.16 (0.5) & 108.49 (1.3) & 217.13 (2.9) &   653.28 (3.7) & 434.89 (2.5) &   869.02 (4.8) \\
     50 &  2.50 & 0.96 (0.1) & 1.78 (0.2) &  3.0 (0.2) & 3.86 (0.3) &  7.69 (0.5) & 15.33 (0.6) & 30.59 (1.0) &  61.41 (0.9) & 123.11 (1.6) & 245.95 (2.7) &   737.44 (2.9) & 491.24 (2.9) &   979.62 (4.1) \\
     50 &  3.00 & 1.53 (0.3) & 3.24 (0.4) & 4.51 (0.5) & 6.41 (0.6) & 13.11 (0.7) &  25.1 (1.4) & 52.02 (1.8) & 104.25 (2.6) & 209.35 (3.4) & 418.48 (4.0) &  1257.83 (6.0) & 837.24 (6.8) & 1672.09 (11.4) \\
     75 &  1.50 & 0.36 (0.0) & 0.77 (0.0) & 1.15 (0.1) & 1.51 (0.1) &  2.99 (0.1) &  6.03 (0.2) & 11.85 (0.2) &  24.12 (0.3) &  48.19 (0.3) &  96.56 (0.6) &   288.39 (1.1) & 192.26 (0.9) &    384.9 (0.7) \\
     75 &  2.00 & 0.57 (0.1) & 1.16 (0.1) & 1.76 (0.1) & 2.36 (0.1) &  4.71 (0.2) &  9.24 (0.3) & 18.55 (0.3) &  37.19 (0.6) &  74.59 (1.1) & 149.09 (1.5) &   449.01 (2.1) &  298.6 (2.2) &   598.39 (3.0) \\
     75 &  2.38 & 0.84 (0.1) & 1.78 (0.2) & 2.45 (0.2) & 3.38 (0.2) &  6.62 (0.3) & 13.34 (0.5) & 26.91 (1.0) &  54.06 (0.8) &  108.4 (1.0) & 216.27 (1.9) &   650.15 (3.7) & 433.99 (1.9) &   870.44 (3.2) \\
     75 &  2.50 &  1.0 (0.2) & 1.93 (0.2) & 2.78 (0.3) & 3.84 (0.3) &  7.94 (0.5) & 15.72 (0.4) &  30.6 (1.0) &  61.85 (1.4) & 123.26 (1.4) & 246.05 (1.3) &   738.06 (5.5) & 490.27 (3.5) &   980.84 (4.4) \\
     75 &  3.00 & 1.63 (0.3) & 3.47 (0.3) & 4.72 (0.4) & 6.83 (0.7) & 12.75 (0.8) &  26.3 (0.7) & 51.91 (1.7) & 104.39 (2.3) & 208.17 (2.3) & 419.61 (4.3) &  1260.37 (7.4) & 837.47 (7.3) &  1671.08 (9.9) \\
    100 &  1.50 & 0.38 (0.0) & 0.77 (0.1) & 1.14 (0.1) & 1.46 (0.1) &  3.02 (0.1) &  6.11 (0.2) & 11.96 (0.3) &  24.15 (0.3) &  47.96 (0.3) &  96.04 (0.8) &   288.74 (1.3) & 192.82 (1.0) &   384.58 (1.5) \\
    100 &  2.00 & 0.58 (0.1) & 1.16 (0.1) & 1.76 (0.2) &  2.4 (0.2) &  4.65 (0.2) &  9.42 (0.4) & 18.35 (0.4) &  37.55 (0.5) &  75.01 (0.5) & 149.88 (1.1) &   449.87 (4.4) & 298.37 (2.9) &   598.48 (3.5) \\
    100 &  2.38 & 0.86 (0.2) & 1.65 (0.2) & 2.38 (0.2) & 3.37 (0.2) &  6.62 (0.3) & 13.54 (0.6) & 27.15 (1.3) &  53.85 (0.8) & 109.06 (1.7) &  216.9 (2.2) &   652.28 (4.7) & 435.55 (4.1) &   871.17 (3.4) \\
    100 &  2.50 & 0.92 (0.1) & 1.81 (0.2) & 2.83 (0.2) & 3.85 (0.3) &  7.43 (0.4) & 15.74 (0.6) & 30.55 (0.9) &  61.79 (1.3) &  123.4 (2.0) & 244.96 (2.3) &    737.2 (4.1) & 491.43 (4.4) &   983.18 (5.0) \\
    100 &  3.00 & 1.59 (0.3) & 3.44 (0.5) & 4.85 (0.3) & 6.42 (0.6) & 13.15 (0.7) & 26.59 (1.0) & 52.22 (2.3) & 105.02 (2.8) & 211.26 (3.6) & 418.43 (4.6) &  1254.62 (7.6) & 838.35 (7.9) &  1677.34 (9.6) \\
\bottomrule
\end{tabular}
}
\end{table}

\begin{table}[ht]
\centering
\caption{Average response time - \fifila{} (in seconds). }
\label{tab:fifila}
\resizebox{\columnwidth}{!}{%
% \scriptsize
\tiny
\begin{tabular}{ll|rrrr|rrrrrrrrr}
\toprule
    \multicolumn{ 2}{c|}{\textbf{Parameters}} &               \multicolumn{ 4}{c|}{{\bf Small Instances }}  &
    \multicolumn{ 9}{c}{{\bf Large Instances }}
    \\
    \hline
$p (\%)$  & $\sigma$ &  100    &  200    &  300    &  400    &  800    &  1,600   &  3,200   &  6,400   &  12,800  &  25,600  &  51,200 & 76,800 &  102,400 \\
\midrule
25 &  1.50 & 0.19 (0.04) & 0.36 (0.05) & 0.58 (0.04) & 0.76 (0.04) &  1.53 (0.11) &  2.97 (0.11) &  6.13 (0.13) &  11.91 (0.26) &  23.97 (0.26) &  47.96 (0.46) &  143.72 (0.74) &  96.51 (1.04) &   192.36 (0.47) \\
     25 &  2.00 & 0.32 (0.04) & 0.59 (0.13) & 0.85 (0.07) & 1.14 (0.14) &  2.36 (0.12) &  4.71 (0.22) &  9.26 (0.21) &  18.74 (0.74) &  37.38 (0.75) &  74.45 (1.41) &  225.33 (2.23) &  149.71 (1.6) &   297.96 (1.49) \\
     25 &  2.38 &  0.4 (0.11) & 0.77 (0.14) & 1.32 (0.17) & 1.73 (0.24) &  3.46 (0.25) &  6.77 (0.22) &  13.3 (0.49) &  27.03 (0.39) &  54.92 (1.17) & 109.54 (1.32) &   324.53 (1.6) & 218.31 (2.13) &   434.73 (3.52) \\
     25 &  2.50 & 0.48 (0.15) & 1.02 (0.23) & 1.42 (0.22) & 1.89 (0.23) &   4.0 (0.45) &  7.93 (0.51) &  15.58 (0.9) &  30.62 (0.68) &  61.69 (0.95) & 122.59 (1.46) &  367.07 (3.01) & 244.46 (2.45) &   492.17 (3.24) \\
     25 &  3.00 & 0.77 (0.17) & 1.71 (0.22) &  2.54 (0.4) & 3.38 (0.24) &  6.57 (0.47) & 12.98 (0.64) & 26.63 (1.04) &    51.3 (1.7) & 103.79 (1.34) & 209.25 (4.34) &  629.61 (5.27) & 419.58 (4.78) &   839.61 (6.21) \\
     50 &  1.50 & 0.25 (0.03) & 0.51 (0.04) & 0.75 (0.06) & 1.03 (0.05) &  2.06 (0.11) &   4.0 (0.17) &  8.08 (0.23) &   16.0 (0.25) &  32.08 (0.34) &  64.09 (0.87) &  192.04 (1.04) &  127.99 (0.8) &    256.89 (1.1) \\
     50 &  2.00 & 0.38 (0.08) &  0.8 (0.07) & 1.21 (0.14) &  1.61 (0.1) &  3.15 (0.26) &   6.28 (0.3) & 12.53 (0.15) &   24.93 (0.6) &  49.81 (0.59) &  99.44 (1.29) &  298.97 (1.77) & 199.93 (1.36) &   399.21 (1.13) \\
     50 &  2.38 &  0.51 (0.1) & 1.13 (0.12) & 1.69 (0.18) & 2.31 (0.27) &   4.5 (0.34) &  8.87 (0.29) & 17.94 (0.27) &  36.01 (0.95) &  72.79 (1.65) & 144.15 (2.13) &  436.17 (2.57) & 289.05 (2.86) &     579.4 (4.3) \\
     50 &  2.50 &  0.6 (0.12) & 1.13 (0.23) & 2.07 (0.16) & 2.59 (0.35) &  5.08 (0.28) & 10.18 (0.64) & 20.49 (0.85) &  40.66 (1.33) &  83.22 (0.79) & 164.16 (2.18) &  491.89 (3.01) & 326.64 (2.37) &   653.87 (4.97) \\
     50 &  3.00 & 0.96 (0.24) & 2.23 (0.46) &  3.12 (0.5) & 4.26 (0.57) &  8.63 (0.79) & 17.05 (1.17) & 34.85 (1.66) &  69.75 (1.56) & 138.78 (2.86) & 280.94 (2.73) &  839.18 (8.76) & 556.21 (6.99) &  1115.11 (9.63) \\
     75 &  1.50 & 0.31 (0.04) & 0.62 (0.04) & 0.94 (0.08) &  1.25 (0.1) &  2.49 (0.09) &  5.01 (0.22) &  9.84 (0.23) &  20.04 (0.34) &  40.17 (0.48) &  80.55 (0.64) &  240.57 (1.34) & 159.78 (1.28) &   321.13 (0.99) \\
     75 &  2.00 & 0.44 (0.09) & 0.93 (0.12) & 1.41 (0.14) & 1.98 (0.18) &  3.85 (0.23) &  7.67 (0.24) & 15.53 (0.35) &   30.76 (0.4) &  62.07 (0.85) & 124.46 (1.03) &  373.19 (2.23) &  248.7 (2.17) &   498.42 (2.86) \\
     75 &  2.38 & 0.72 (0.16) & 1.46 (0.15) & 2.07 (0.12) &  2.84 (0.3) &   5.61 (0.3) & 11.11 (0.38) & 22.42 (1.03) &  45.42 (1.34) &  90.99 (1.44) & 179.84 (1.95) &  541.38 (4.44) &  362.05 (2.9) &   724.47 (2.44) \\
     75 &  2.50 &  0.8 (0.14) & 1.61 (0.23) & 2.41 (0.26) & 3.32 (0.27) &  6.66 (0.41) & 12.91 (0.64) & 25.72 (1.02) &  51.26 (1.12) &  102.55 (1.9) & 206.43 (2.51) &  613.65 (4.91) & 409.13 (2.26) &   817.27 (4.48) \\
     75 &  3.00 & 1.35 (0.31) &  2.9 (0.45) & 4.13 (0.52) & 5.83 (0.44) & 10.54 (1.05) & 21.81 (0.72) & 43.53 (1.87) &  86.79 (2.92) & 174.76 (2.68) & 351.24 (4.49) & 1052.72 (7.89) & 698.67 (7.62) &    1394.4 (9.3) \\
    100 &  1.50 & 0.38 (0.05) & 0.73 (0.07) & 1.15 (0.06) & 1.47 (0.08) &  3.02 (0.15) &  6.13 (0.19) & 12.01 (0.42) &   24.11 (0.4) &  47.96 (0.42) &  95.88 (0.78) &  288.88 (1.04) & 192.52 (1.07) &   384.64 (1.61) \\
    100 &  2.00 & 0.57 (0.12) &   1.2 (0.1) & 1.75 (0.18) & 2.37 (0.16) &  4.68 (0.33) &   9.4 (0.51) & 18.42 (0.42) &   37.8 (0.55) &   75.0 (0.93) &  150.0 (1.39) &  450.28 (3.68) & 298.01 (3.35) &    597.9 (3.62) \\
    100 &  2.38 & 0.86 (0.22) & 1.73 (0.24) &  2.5 (0.37) &  3.3 (0.17) &  6.76 (0.36) & 13.45 (0.56) & 27.06 (1.52) &  54.15 (1.35) & 109.29 (1.71) & 216.55 (2.49) &  653.67 (4.54) & 434.72 (4.71) &    870.3 (4.72) \\
    100 &  2.50 & 0.86 (0.11) & 1.87 (0.19) & 2.87 (0.24) & 3.85 (0.35) &  7.52 (0.54) & 15.48 (0.63) & 30.45 (1.28) &  61.68 (0.99) & 123.17 (2.65) & 246.07 (2.51) &   736.54 (4.5) & 492.22 (5.41) &   983.42 (6.66) \\
    100 &  3.00 & 1.64 (0.35) & 3.27 (0.46) & 4.69 (0.33) & 6.47 (0.73) &   13.2 (0.8) & 27.17 (1.24) & 52.22 (2.32) & 104.46 (2.81) &   212.1 (3.3) & 418.58 (4.87) & 1254.78 (7.94) & 841.76 (8.21) & 1676.29 (12.15) \\
\bottomrule
\end{tabular}
}
\end{table}

\begin{table}[ht]
\centering
\caption{Average response time - \lafila{} (in seconds).}
\label{tab:lafila}
\tiny
% \scriptsize
\resizebox{\columnwidth}{!}{%
\begin{tabular}{ll|rrrr|rrrrrrrrr}
\toprule
    \multicolumn{ 2}{c|}{\textbf{Parameters}} &               \multicolumn{ 4}{c|}{{\bf Small Instances }}  &
    \multicolumn{ 9}{c}{{\bf Large Instances }}
    \\
    \hline
$p (\%)$  & $\sigma$ &  100    &  200    &  300    &  400    &  800    &  1,600   &  3,200   &  6,400   &  12,800  &  25,600  &  51,200  & 76.800 &  102,400 \\
\midrule
25 &  1.50 & 0.19 (0.0) & 0.36 (0.1) & 0.58 (0.0) & 0.76 (0.0) & 1.53 (0.1) &  2.97 (0.1) &  6.13 (0.1) & 11.91 (0.3) &  23.97 (0.3) &  47.96 (0.5) & 143.72 (0.7) &  96.51 (1.0) & 192.36 (0.5) \\
     25 &  2.00 & 0.31 (0.0) & 0.59 (0.1) & 0.84 (0.1) & 1.14 (0.1) & 2.36 (0.1) &   4.7 (0.2) &  9.26 (0.2) & 18.74 (0.7) &  37.38 (0.7) &  74.44 (1.4) & 225.33 (2.2) & 149.71 (1.6) & 297.96 (1.5) \\
     25 &  2.38 & 0.39 (0.1) & 0.75 (0.1) & 1.29 (0.2) & 1.69 (0.2) & 3.42 (0.2) &   6.7 (0.2) & 13.19 (0.5) & 26.69 (0.3) &  54.47 (1.2) & 108.71 (1.4) & 321.42 (2.9) & 216.88 (2.5) & 431.17 (4.4) \\
     25 &  2.50 & 0.47 (0.1) & 0.98 (0.2) & 1.38 (0.2) & 1.85 (0.2) & 3.92 (0.5) &  7.75 (0.5) & 15.24 (0.9) & 30.02 (0.7) &  60.84 (1.0) & 120.55 (1.7) & 361.68 (4.7) & 240.18 (2.4) & 483.73 (3.8) \\
     25 &  3.00 & 0.73 (0.2) & 1.64 (0.2) & 2.39 (0.3) & 3.22 (0.2) & 6.27 (0.4) & 12.43 (0.7) & 25.47 (1.0) &  49.2 (1.6) &  99.73 (1.5) & 200.34 (4.1) & 603.77 (5.7) & 402.52 (4.2) & 804.97 (5.8) \\
     50 &  1.50 & 0.23 (0.0) & 0.46 (0.0) & 0.69 (0.1) & 0.96 (0.0) & 1.92 (0.1) &  3.75 (0.2) &  7.53 (0.2) & 14.87 (0.2) &  29.89 (0.3) &  59.68 (0.8) & 178.82 (0.8) & 119.24 (0.8) & 239.18 (1.1) \\
     50 &  2.00 & 0.33 (0.1) &  0.7 (0.1) & 1.06 (0.1) & 1.39 (0.1) & 2.71 (0.2) &  5.49 (0.2) & 10.94 (0.2) & 21.72 (0.5) &  43.46 (0.5) &  86.88 (1.0) &  260.8 (1.5) & 174.77 (1.1) & 348.92 (1.2) \\
     50 &  2.38 & 0.44 (0.1) & 0.96 (0.1) & 1.43 (0.2) & 1.93 (0.2) & 3.75 (0.3) &  7.33 (0.3) & 14.98 (0.3) & 30.22 (0.6) &  60.88 (1.3) &  120.8 (1.9) & 365.11 (2.3) & 242.34 (2.1) & 485.11 (3.3) \\
     50 &  2.50 & 0.51 (0.1) & 0.94 (0.2) &  1.7 (0.1) & 2.13 (0.2) & 4.21 (0.3) &  8.44 (0.5) & 16.94 (0.8) & 33.74 (0.8) &  68.49 (0.9) & 135.73 (1.9) & 407.12 (2.0) & 270.52 (1.8) & 540.64 (3.3) \\
     50 &  3.00 & 0.77 (0.2) & 1.76 (0.3) & 2.41 (0.4) & 3.35 (0.4) & 6.85 (0.5) & 13.38 (0.9) & 27.49 (1.3) & 54.91 (1.5) & 110.17 (2.1) & 222.03 (2.2) & 664.82 (6.4) & 441.32 (4.7) & 883.26 (8.0) \\
     75 &  1.50 & 0.25 (0.0) & 0.52 (0.0) & 0.78 (0.1) & 1.04 (0.1) & 2.06 (0.1) &  4.15 (0.2) &  8.16 (0.2) &  16.6 (0.2) &  33.28 (0.4) &  66.72 (0.5) & 199.25 (0.9) & 132.54 (0.8) & 265.77 (0.7) \\
     75 &  2.00 & 0.35 (0.1) & 0.71 (0.1) & 1.09 (0.1) & 1.49 (0.1) & 2.94 (0.1) &  5.85 (0.2) & 11.76 (0.2) & 23.47 (0.4) &  47.24 (0.6) &  94.63 (0.7) & 284.16 (1.7) & 189.19 (1.6) & 379.16 (1.8) \\
     75 &  2.38 & 0.52 (0.1) & 1.04 (0.1) & 1.49 (0.1) & 2.02 (0.2) & 4.03 (0.2) &   8.0 (0.3) & 16.14 (0.7) & 32.65 (0.8) &  65.43 (0.9) &  129.8 (1.3) & 390.17 (2.6) & 260.88 (1.5) & 522.32 (1.8) \\
     75 &  2.50 & 0.58 (0.1) & 1.13 (0.2) & 1.67 (0.2) & 2.31 (0.2) & 4.72 (0.3) &  9.21 (0.4) & 18.22 (0.7) &  36.5 (0.8) &  72.93 (1.3) & 146.15 (1.3) & 436.39 (3.3) & 290.16 (1.8) & 579.93 (2.8) \\
     75 &  3.00 & 0.93 (0.2) & 1.95 (0.2) & 2.73 (0.3) & 3.91 (0.4) & 7.06 (0.6) & 14.64 (0.4) & 29.27 (1.2) & 58.43 (1.7) & 117.26 (1.6) & 235.85 (2.8) & 706.83 (4.7) &  469.5 (5.1) & 936.65 (6.2) \\
    100 &  1.50 & 0.28 (0.0) & 0.55 (0.0) & 0.84 (0.0) & 1.08 (0.1) & 2.25 (0.1) &  4.53 (0.1) &  8.85 (0.3) &  17.9 (0.3) &  35.41 (0.3) &  70.95 (0.5) & 213.53 (0.9) & 142.51 (0.9) & 284.36 (1.1) \\
    100 &  2.00 & 0.39 (0.1) & 0.78 (0.1) & 1.18 (0.1) & 1.59 (0.1) & 3.13 (0.2) &  6.32 (0.3) & 12.32 (0.3) & 25.34 (0.4) &   50.4 (0.5) &  100.8 (0.9) & 302.39 (2.8) & 200.25 (2.3) &  402.0 (2.3) \\
    100 &  2.38 & 0.54 (0.1) & 1.07 (0.2) & 1.55 (0.2) & 2.12 (0.1) & 4.21 (0.2) &  8.58 (0.4) & 17.12 (0.9) &  34.1 (0.7) &  69.08 (1.2) & 137.24 (1.5) &  413.1 (2.9) & 275.28 (3.0) & 550.55 (2.8) \\
    100 &  2.50 & 0.55 (0.1) & 1.15 (0.1) & 1.74 (0.2) & 2.38 (0.2) & 4.65 (0.3) &   9.7 (0.4) & 18.95 (0.7) & 38.35 (0.8) &  76.59 (1.6) & 152.62 (1.6) &  458.0 (2.7) &  305.8 (3.2) & 611.96 (3.6) \\
    100 &  3.00 & 0.95 (0.2) & 1.95 (0.3) & 2.74 (0.2) &  3.8 (0.4) &  7.7 (0.5) & 15.91 (0.7) &  30.4 (1.4) & 61.26 (1.8) & 123.86 (1.9) & 244.44 (3.2) & 733.35 (4.4) & 491.18 (5.1) & 979.49 (7.5) \\
\bottomrule
\end{tabular}
}
\end{table}
\end{landscape}

\begin{table}[ht]
\centering
\caption{Average response times - \exact{} (in seconds).}
\label{tab:exact}
\tiny
% \scriptsize
\begin{tabular}{ll|rrrr}
\toprule
    \multicolumn{ 2}{c|}{\textbf{Parameters}} &               \multicolumn{ 4}{c}{{\bf Small Instances }}  
    \\
    \hline
$p (\%)$  & $\sigma$ &  100    &  200    &  300    &  400    \\
\midrule
25 &  1.50 & 0.19 (0.0) & 0.36 (0.1) & 0.58 (0.0) & 0.75 (0.0) \\
     25 &  2.00 & 0.31 (0.0) & 0.59 (0.1) & 0.84 (0.1) & 1.13 (0.1) \\
     25 &  2.38 & 0.39 (0.1) & 0.75 (0.1) & 1.29 (0.2) & 1.69 (0.2) \\
     25 &  2.50 & 0.47 (0.1) & 0.98 (0.2) & 1.38 (0.2) & 1.85 (0.2) \\
     25 &  3.00 & 0.72 (0.2) & 1.63 (0.2) & 2.39 (0.3) & 3.21 (0.2) \\
     50 &  1.50 & 0.23 (0.0) & 0.46 (0.0) & 0.69 (0.1) & 0.96 (0.0) \\
     50 &  2.00 & 0.33 (0.1) &  0.7 (0.1) & 1.06 (0.1) & 1.39 (0.1) \\
     50 &  2.38 & 0.44 (0.1) & 0.96 (0.1) & 1.42 (0.2) & 1.93 (0.2) \\
     50 &  2.50 & 0.51 (0.1) & 0.94 (0.2) &  1.7 (0.1) & 2.13 (0.2) \\
     50 &  3.00 & 0.77 (0.2) & 1.76 (0.3) & 2.41 (0.4) & 3.35 (0.4) \\
     75 &  1.50 & 0.25 (0.0) & 0.52 (0.0) & 0.78 (0.1) & 1.04 (0.1) \\
     75 &  2.00 & 0.35 (0.1) & 0.71 (0.1) & 1.09 (0.1) & 1.49 (0.1) \\
     75 &  2.38 & 0.52 (0.1) & 1.04 (0.1) & 1.49 (0.1) & 2.02 (0.2) \\
     75 &  2.50 & 0.57 (0.1) & 1.13 (0.2) & 1.67 (0.2) & 2.31 (0.2) \\
     75 &  3.00 & 0.92 (0.2) & 1.95 (0.2) & 2.73 (0.3) & 3.91 (0.4) \\
    100 &  1.50 & 0.28 (0.0) & 0.55 (0.0) & 0.84 (0.0) & 1.08 (0.1) \\
    100 &  2.00 & 0.39 (0.1) & 0.78 (0.1) & 1.18 (0.1) & 1.59 (0.1) \\
    100 &  2.38 & 0.54 (0.1) & 1.07 (0.2) & 1.55 (0.2) & 2.12 (0.1) \\
    100 &  2.50 & 0.55 (0.1) & 1.15 (0.1) & 1.74 (0.2) & 2.38 (0.2) \\
    100 &  3.00 & 0.94 (0.2) & 1.94 (0.3) & 2.74 (0.2) & 3.79 (0.4) \\
\bottomrule
\end{tabular}
\end{table}

Table~\ref{tab:runtimes} reports the average runtime and respective standard deviation of~\exact{} and~\lafila{} (also in seconds, with two decimal places for the average and one for the standard deviation) based on each  configuration $(p,\sigma)$. We omit the runtime of~\lafila{} for instances with less than 12,800 files as well as the runtime of all the other algorithms because these values are negligibly small. Moreover, the standard deviation for~\lafila{} is always equal to zero when truncated to two decimal places, so we omit this information from Table~\ref{tab:runtimes} as well.
\begin{table}[ht]
\centering
\caption{Average runtime - \lafila{} and~\exact{} (in seconds).}
\label{tab:runtimes}
\tiny
% \scriptsize
\begin{tabular}{ll|rrrr|rrrrr}
\toprule
\multicolumn{ 2}{c|}{\textbf{Parameters}} & \multicolumn{ 4}{c|}{$\exact{}$}
&
\multicolumn{ 5}{c}{$\lafila{}$}
    \\
    \hline
$p (\%)$  & $\sigma$ & 100    & 200    &  300    &  400    &  12,800  & 25,600  & 51,200  & 76,800  & 102,400 \\
%     & {} & \multicolumn{4}{l}{DP} & \multicolumn{13}{l}{LFL} \\
%     & num\_files & 100    & 200    &  300    &  400    & 100    & 200    & 300    & 400    & 800    & 1600   & 3200   & 6400   & 12800  & 25600  & 51200  & 76800  & 102400 \\
% percent & sigma &        &        &         &         &        &        &        &        &        &        &        &        &        &        &        &        &        \\
\midrule
25 &  1.50 & 0.028 (0.02) &  0.597 (0.21) &     3.22 (2.37) &   10.438 (3.39) &   0.0 & 0.001 & 0.002 & 0.001 &  0.002 \\
     25 &  2.00 & 0.047 (0.03) &  0.599 (0.31) &    3.343 (1.05) &   11.437 (3.31) &   0.0 & 0.001 & 0.002 & 0.001 &  0.003 \\
     25 &  2.38 & 0.034 (0.02) &  0.508 (0.26) &    3.766 (1.89) &   10.125 (4.21) &   0.0 & 0.001 & 0.002 & 0.001 &  0.003 \\
     25 &  2.50 & 0.032 (0.03) &  0.604 (0.15) &    3.303 (1.91) &      9.57 (4.0) &   0.0 & 0.001 & 0.002 & 0.002 &  0.003 \\
     25 &  3.00 & 0.051 (0.03) &  0.488 (0.15) &    3.212 (1.12) &    10.841 (2.8) &   0.0 & 0.001 & 0.002 & 0.002 &  0.003 \\
     50 &  1.50 & 0.335 (0.11) &  6.968 (1.13) &  36.842 (10.59) &  134.96 (27.24) & 0.001 & 0.001 & 0.004 & 0.002 &  0.005 \\
     50 &  2.00 & 0.342 (0.13) &  5.816 (1.18) &  37.434 (10.55) & 141.646 (28.48) & 0.001 & 0.001 & 0.004 & 0.002 &  0.005 \\
     50 &  2.38 &  0.344 (0.1) &  5.701 (0.96) &   39.338 (5.67) & 124.112 (33.33) & 0.001 & 0.001 & 0.003 & 0.002 &  0.005 \\
     50 &  2.50 &  0.33 (0.12) &  5.228 (0.91) &  40.472 (11.89) & 128.625 (23.31) & 0.001 & 0.001 & 0.003 & 0.002 &  0.005 \\
     50 &  3.00 & 0.325 (0.14) &  5.837 (1.84) &   42.639 (9.65) & 125.374 (17.22) &   0.0 & 0.001 & 0.003 & 0.002 &  0.004 \\
     75 &  1.50 & 1.019 (0.06) & 19.309 (2.58) &  117.204 (9.28) &  419.384 (45.4) & 0.001 & 0.001 & 0.004 & 0.002 &  0.005 \\
     75 &  2.00 & 0.939 (0.17) & 19.322 (2.65) & 119.615 (12.32) &  434.223 (33.8) & 0.001 & 0.001 & 0.004 & 0.002 &  0.005 \\
     75 &  2.38 & 1.054 (0.08) & 20.058 (2.28) & 114.735 (13.43) & 437.175 (35.62) & 0.001 & 0.001 & 0.004 & 0.002 &  0.005 \\
     75 &  2.50 & 1.019 (0.16) &  19.505 (1.3) &   117.44 (7.08) & 423.632 (28.27) & 0.001 & 0.001 & 0.003 & 0.002 &  0.005 \\
     75 &  3.00 & 1.014 (0.16) &  20.896 (2.4) &  117.295 (8.46) &  442.009 (44.0) &   0.0 & 0.001 & 0.003 & 0.002 &  0.004 \\
    100 &  1.50 & 1.646 (0.01) & 32.932 (0.18) &  196.048 (0.63) &  725.787 (5.38) & 0.001 & 0.001 & 0.004 & 0.002 &  0.005 \\
    100 &  2.00 & 1.644 (0.01) & 32.939 (0.16) &  196.853 (1.47) &   724.231 (4.9) &   0.0 & 0.001 & 0.003 & 0.002 &  0.004 \\
    100 &  2.38 & 1.648 (0.01) & 33.173 (0.85) &  195.943 (0.97) &  723.494 (4.57) &   0.0 & 0.001 & 0.003 & 0.002 &  0.004 \\
    100 &  2.50 & 1.646 (0.01) & 33.401 (1.61) &   195.61 (0.58) &   724.318 (7.0) &   0.0 & 0.001 & 0.003 & 0.002 &  0.004 \\
    100 &  3.00 & 1.653 (0.02) & 32.804 (0.06) &  196.289 (0.93) & 726.582 (11.72) &   0.0 & 0.001 & 0.003 & 0.002 &  0.004 \\
\bottomrule
\end{tabular}
\end{table}

Table~\ref{tab:stoch} compares the average quality of the solution produced by  each algorithm, given by the ratio between the solution produced by the algorithm and the optimal solution identified by~$\exact$. \lafila{} delivers solid performance, as observed in the other experiments, but we observe that its advantage of~\texttt{SLTS} gets reduced as the percentage with which files are requested decreases.

\begin{table}[ht]
\centering
\caption{Average approximation ratio of each algorithm on stochastic instances in comparison with~\exact{}.}
\label{tab:stoch}
\tiny
% \scriptsize
\begin{tabular}{l|lllll}
\toprule
\multicolumn{ 1}{c|}{\textbf{$p (\%)$}} & \multicolumn{ 5}{c}{Algorithm}
    \\
    \hline
 & \fifo{} &    \fififi{} & \fifila{}  &    \lafila{}  &      \texttt{SLTS} \\
\midrule
      20 & 2.35 (0.2) & 2.06 (0.2) & 1.01 (0.0) &  1.0 (0.0) & 1.21 (0.1) \\
      40 & 2.07 (0.1) & 1.77 (0.1) & 1.11 (0.0) &  1.0 (0.0) & 1.07 (0.0) \\
      60 & 1.94 (0.1) & 1.64 (0.1) & 1.26 (0.0) &  1.0 (0.0) & 1.02 (0.0) \\
      80 & 1.87 (0.0) & 1.56 (0.0) & 1.41 (0.0) & 1.01 (0.0) & 1.01 (0.0) \\
     100 & 1.83 (0.0) & 1.51 (0.0) & 1.58 (0.0) & 1.01 (0.0) &  1.0 (0.0) \\
\bottomrule
\end{tabular}
\end{table}

\subsection{Landsat Instances}\label{sec:tables_landsat}

For the real-world tape settings, we consider a case study on the use of remote sensing for precision viticulture in 3 zones of interest: the Atacama desert in Chile, the Serra Gaucha region in Brazil, and the Manduria region in Italy.  
The Atacama region is represented by a cluster of 4 Landsat-8 tiles, and it is known for its dry weather and a mostly cloud-free sky all year round, which translates into crisp satellite imagery. 
%The Atacama region  is represented by a cluster of 4 Landsat-8 tiles. 
The software accessing these tiles computes the Normalized Difference Vegetation Index (NDVI) of the region by reading Landsat bands number 4 and 5, which represent the reflective radiation in visible Red and Near-Infrared wavelengths, respectively. Differently from the Atacama region, Serra Gaucha and Manduria are frequently covered with clouds that, at times, invalidate whole satellite scenes \citep{CloudyEarth}. The Serra Gaucha region is covered by  6 Landsat-8 tiles that, once processed, are converted into false-color images that emphasize vegetation; the image composition procedure used in this application requires the use of bands number 4, 5, and 6. Last, Manduria is composed of five tiles. The associated application uses a  neural network that processes the assembled data to classify the vineyards' health conditions.

Each aforementioned region is analysed by a different software, but in all cases read requests are contingent on the analysis of
% %All software processing the aforementioned zones look into 
the QA band. 
% %launching parallel readers for 
% %the actual bands of interest are released. 
If an image has a cloud cover above a certain threshold, then no bands other than the QA are read; the missing data is then interpolated or extrapolated at a later time using data from past and future scenes. This aspect is important, as cloud cover thresholds have a direct impact on read requests. Our experiments involve only the second stage of this pipeline, i.e., we assume that the QA bands have been inspected already and that all read requests involve spectrum bands of images that satisfy the cloud cover conditions.

% %. That is, we do not consider the problem of scheduling read requests used for the inspection of QA bands, we just consider but the second stage problem, which filters out files that do not satisfy the cloud cover conditions.

% %An implementation detail worth noting is that parallel readers are orchestrated by a high-performance I/O interface implemented on top of Message Passing Interface (MPI)~(\cite{Thakur99}). As a consequence of the collective I/O techniques employed by MPI, release times of jobs related to bands within a same satellite scene are always identical.

%Each instance of the Landsat benchmark consist of 15 tiles, each composed of 12 files (that is, one file per band). The average size of a file is approximately 280Mb.
The results of our experiments involving the Landsat instances are presented in Table~\ref{tlandsat}. Each instance  is associated with a configuration~$(\alpha,\beta,\gamma)$;
% % in $\{
% % (2,15,22),
% % (17,15,32),
% % (7,21,22),
% % (9,18,38),
% % (17,21,32),
% % (9,15,22),
% % (2,26,38),
% % (17,21,22),
% % (2,21,22),
% % (17,26,38)
% % \}$; 
numbers $\alpha$, $\beta$, and~$\gamma$ belong to the interval~$[0,100]$ and indicate the threshold for cloud covers adopted by the applications using data associated with the Atacama, Serra Gaucha, and Manduria region, respectively. The data set has 10 instances for each of the following configurations of~$(\alpha,\beta,\gamma)$: 
(2,15,22),
(17,15,32),
(7,21,22),
(9,18,38),
(17,21,32),
(9,15,22),
(2,26,38),
(17,21,22),
(2,21,22), and
(17,26,38). The first column in the table indicates the configuration~$\alpha-\beta-\gamma$, and each entry reports the average response times (in seconds) and respective standard deviations obtained by each algorithm over the 10 instances using  two decimal places for the average and one for the standard deviation).

%Finally, in order to avoid numerical issues, the unit for file sizes used in this data set if megabytes. 

\begin{table}[ht!]
\centering
% \footnotesize
\scriptsize
\caption{Average response times for Landsat instances (in seconds).}
\label{tlandsat}
\begin{tabular}{c|cccccc}
\toprule
\multicolumn{1}{c|}{Instance} & \multicolumn{5}{c}{Algorithm} \\
$\alpha-\beta-\gamma$       &   \fifo{} & \fififi{} & \fifila{} & \lafila{} & \exact{} \\
\midrule
02-15-22 & 237.46 (16.2) & 208.07 (13.8) & 127.19 (21.9) & 127.19 (21.9) & 127.19 (21.9) \\
      16-15-32 & 221.34 (12.9) & 211.54 (11.0) & 103.07 (16.6) & 103.07 (16.6) & 103.07 (16.6) \\
      07-21-22 & 221.82 (18.6) & 204.93 (15.1) &   115.7 (8.8) &   115.7 (8.8) &   115.7 (8.8) \\
      09-18-38 & 235.87 (20.7) & 204.64 (14.7) & 115.85 (14.9) & 115.85 (14.9) & 115.84 (14.9) \\
      17-21-32 &  228.6 (23.8) & 201.98 (16.7) &  119.2 (16.0) &  119.2 (16.0) &  119.2 (16.0) \\
      09-15-22 & 226.28 (16.7) &  201.13 (9.1) &  131.4 (21.4) &  131.4 (21.4) &  131.4 (21.4) \\
      02-26-38 & 224.83 (23.4) & 200.44 (11.2) & 108.49 (18.4) & 108.49 (18.4) & 108.49 (18.4) \\
      17-21-22 & 235.41 (13.3) & 209.43 (11.7) & 116.59 (18.1) & 116.59 (18.1) & 116.59 (18.1) \\
      02-21-22 & 231.84 (15.2) & 198.67 (10.2) & 122.83 (13.9) & 122.83 (13.9) & 122.83 (13.9) \\
      17-26-38 & 227.12 (23.0) & 198.34 (15.6) & 118.09 (14.0) & 118.09 (14.0) & 118.08 (14.0) \\
\bottomrule
\end{tabular}
\end{table}

\section{Proofs}

For ease of reference, we organize the proofs following the structure of the manuscript.

\subsection{The \ref{model:LTS} with Deterministic Service Requests}

We use the following auxiliary lemma for the results below.
\begin{lemma}
	\label{lemma:onlyatomic}
	Suppose all files have the same size, i.e., $\filesize{\file} = \filesize{}$ for all $\file \in \fileset$. Given an arbitrary set of requested files, there exists an optimal solution $\optordertuple$ such that $\order^*_{t} > \order^*_{t+1}$ for any rewind-stage file $\order^*_{t} \in \rewindphase_{\optordertuple}$, $t < \numfiles$. That is, files in the rewind stage are read in descending order of their indices.
\end{lemma}
\begin{proof}{Proof of Lemma \ref{lemma:onlyatomic}.}
	Let $\ordertuple$ be a solution to the \ref{model:LTS} with a (non-necessarily contiguous) subsequence $(\order_{t},\ldots,\order_{t+k})$, $k \ge 1$, of rewind-stage files such that $\order_{t} < \ldots < \order_{t+k}$, i.e., they are read in ascending order in $\ordertuple$. We can assume $\fileweight{\order_{t}} = \fileweight{\order_{t+k}} = 1$; otherwise, both files could be delayed to the forward stage without impacting the objective of the~\ref{model:LTS}. Since the tape head starts at position $\tapelength$, the file $\order_{t+k}$ is eventually traversed from right to left at some moment prior to reading $\order_{t}$. Create a new sequence $\ordertuple'$ where $\order_{t+k}$ is read immediately after the last-serviced file $\order_{t'}$ prior to reaching $\leftfile{\order_{t+k}}$ for the first time, where $t' < t$; observe that such a file must exist, as we would have $\order_1 = \order_t$ otherwise, thus contradicting the assumption that~$\order_t$ is a rewind-stage file. 
	
	The response time of each of the $k-1$ files $(\order_{t},\ldots,\order_{t+k-1})$ in~$\ordertuple'$ increases by $2\filesize{}$ with respect to~$\ordertuple$. Thus, 
	the increase in the response time for reading these $k-1$ files is at most $2\filesize{}(k-1)$ in $\ordertuple'$, or potentially less if any of such files is not requested. Conversely, the response time of $\order_{t+k}$ in $\ordertuple'$ decreases by at least $2\filesize{}(k-1)$, or more if the subsequence $(\order_{t},\ldots,\order_{t+k})$ is not contiguous. Since the response times of the files read after before~$\order_{t}$ or after~$\order_{t+k}$ do not change, the objective value of $\ordertuple'$ is smaller than or equal to $\ordertuple$. The result follows from  the iterative application of this argument on  any optimal solution~$\optordertuple$ to the \ref{model:LTS}. \hfill $\blacksquare$
\end{proof}

\bigskip

\begin{proof}{Proof of Proposition \ref{prop:allforward}.} 
	Let $\optordertuple$ be an optimal solution to the \ref{model:LTS} that violates the condition of the statement, i.e., there exist indices $t < t'$ such that $\order^*_{t}, \order^*_{t'} \in \forwardphase_{\optordertuple}$ and $\order^*_{t} > \order^*_{t'}$. By definition, $1 \in \forwardphase_{\optordertuple}$, which implies that the tape head must traverse $\order^*_{t'}$ from left to right  prior to reading $\order^*_{t}$. Suppose that such a traversal happens first between two forward-stage files $\order^*_{\hat{t}}$ and $\order^*_{\hat{t}+1}$ that are read consecutively in $\optordertuple$. Notice that we must have $\hat{t}, \hat{t}+1 \le t$ and $\order^*_{\hat{t}} < \order^*_{\hat{t}+1}$, as~$\order^*_{t'}$ is traversed from left to right. If we replace the contiguous subsequence $\dots,\order^*_{\hat{t}},\order^*_{\hat{t}+1},\dots$ in $\optordertuple$ by 
	$$
	\dots,\order^*_{\hat{t}},\order^*_{t'},\order^*_{\hat{t}+1},\dots
	$$
	that is, we read $\order^*_{t'}$ the first time it is traversed from left to right, the response time of $\order^*_{t'}$ decreases if $\fileweight{\order^*_{t'}} = 1$. Moreover, the response times of the files that precede $\order^*_{t'}$ in $\optordertuple$ do not change, while the response times for the files that succeed $\order^*_{t'}$ in $\optordertuple$ can only decrease since there is no need to rewind the tape to $\leftfile{\order^*_{t'}}$. Thus, the total response time either remains the same or decreases. \hfill $\blacksquare$ 
\end{proof}

\bigskip 

\begin{proof}{Proof of Proposition~\ref{prop:basepolicies}.}  We divide the proof in two parts.

\begin{enumerate}
    \item [(a)] Consider the following family of instances:
	\begin{itemize}
		\item File~$1$:  $\fileweight{1} = 1$ and $\filesize{1} = 1-\epsilon$ for some arbitrarily small~$\epsilon > 0$;
		\item File~$2$:  $\fileweight{2} = 0$ and $\filesize{2} = \numfiles^2$; and
		\item Files~$3,\ldots,\numfiles$:   $\fileweight{\file} = 1$ and $\filesize{\file} = 1$ for $\file= 3,\dots,\numfiles$.
	\end{itemize}
	
	The sum of the response times of the sequence~$\ordertuple$ delivered by~\fififi{} and~\ssf{} is in
%	$\sum\limits_{t = 1}^{\numfiles} \policyval{t}{\ordertuple} \in 
$\mathcal{O}\left(\numfiles^3\right)$, 
	since files $3, 4, \dots, \numfiles$ are read after the tape head traverses the second file twice, an operation that consumes 
	$\mathcal{O}\left(\numfiles^2\right)$. Conversely, the sum of the response times for the sequence~$\ordertuple' = (3,4,\ldots,\numfiles,1,2)$ is in %$\sum\limits_{t = 1}^{\numfiles} \policyval{t}{\ordertuple'} \in
	$\mathcal{O}\left(\numfiles^2\right)$. Finally, \fifo{} delivers the same result if the system submits the request using the same order as~$\ordertuple$.
	
	\item [(b)] Let us consider an instance in which all files are of the same size and are requested, i.e.,	$\numfiles = \numrequests$ and $\filesize{\file} = \filesize{}$ for all $\file \in \fileset$ for some~$\filesize{} \in \mathbb{N}^*$. Let~$\optordertuple$ be an optimal sequence such that $\rewindphase_{\optordertuple} \neq \emptyset$. By Lemma \ref{lemma:onlyatomic} and Proposition \ref{prop:allforward}, we can assume that the forward- and rewind-stage files are ordered, i.e., 
	$\order^*_{t} < \order^*_{t+1}$ for all $\order^*_{t}, \order^*_{t+1} \in \forwardphase_{\optordertuple}$, and $\order^*_{t} > \order^*_{t+1}$ for all $\order^*_{t},\order^*_{t+1} \in \rewindphase_{\optordertuple}$, $t < \numfiles$. Moreover, let $\hat{t} > 1$ be the file with smallest $\leftfile{\order^*_{\hat{t}}}$
	such that $\order^*_{i} \equiv  \order^*_{\hat{t}}-1 \in \forwardphase_{\optordertuple}$ and $\order^*_{\hat{t}} \in \rewindphase_{\optordertuple}$ (i.e., file~$\order^*_{\hat{t}}$ is in the rewind stage, and the file~$\order^*_{i}$ on its left is in the forward stage); the existence of such a file~$\order^*_{\hat{t}}$
	%and $\order^*_{\hat{t}-1} < \order^*_{\hat{t}}$, 
	follows from $1 \in \forwardphase_{\optordertuple}$ and the ascending ordering of files in the tape. %Let $\order^*_{\hat{t}}$ be such a file with smallest $\leftfile{\order^*_{\hat{t}}}$, i.e., $\{1,\ldots,\optordertuple_{i}\} \in \forwardphase_{\optordertuple}$.

	Create a new ordering $\ordertuple'$ such that $\order^*_{\hat{t}}$ is read immediately after $\order^*_{i}$ in the forward stage in $\ordertuple'$. Since all files have the same size, the response time of $\order^*_{\hat{t}}$ in $\ordertuple'$ increases by 
	%$2\leftfile{\order^*_{\hat{t}}}$
	$2\filesize{}(\order^*_{i}-1)$ 
	with respect to its response time in $\optordertuple$, as the~$\order^*_{i}-1$   files positioned to the left of $\order^*_{\hat{t}}$ in the tape, all belonging to the forward stage, are now serviced beforehand. However, the response time of the files $ 1, \ldots, \order^*_{i}-1$ reduce each by $2\filesize{}$ in $\ordertuple'$, as the reading of $\order^*_{\hat{t}}$ is delayed. Thus, the objective value of \ref{model:LTS} is not impacted, since all files are requested and the response time of all remaining files do not change in $\ordertuple'$. Thus, $\ordertuple'$ is optimal and also preserves the forward-stage ascending ordering. The optimality of~\fififi{} follows from the iterative application of this argument to any optimal solution. \hfill $\blacksquare$
	\end{enumerate}
\end{proof}

\bigskip

\begin{proof}{Proof of Proposition~\ref{prop:fifilaPerformance}.}
    The response time of every file $\file \in \fileset$, $\fileweight{\file} = 1$, in the sequence $\ordertuple$ generated by \fifila{} is upper-bounded by $3(\tapelength - \leftfile{\file})$, which is achieved if all files positioned on the right of $\file$ are read before the tape head reaches~$\leftfile{\file})$. As the response time of~$\file$ is never smaller than $\tapelength - \leftfile{\file}$, i.e., the distance between the beginning of the tape and $\file$, \fifila{} is a 3-approximation for~\ref{model:LTS}. Finally, observe that the approximation ration of~\fifila{} is asymptotically tight, as one can observe by inspecting the performance of the algorithm on instances with a single large file positioned at the last position and small files positioned on its left. 
    %Since inequality~\eqref{eq:optcond} is a necessary condition for optimality, the sequence produced by~\lafila{} has either the same or lower objective value than the one produced by \fifila{}, and hence it is a 3-approximation as well. Finally, we show that \fifila{}  terminates in $\mathcal{O}(\numfiles^2)$ steps. First, based on Proposition~\ref{prop:necessary_optimality_condition}, a file $\order_t$ that violates inequality \eqref{eq:optcond} is in the rewind stage and is postponed to the forward stage in Step 2-a. Thus, $\order_t$ is not selected again and the algorithm must necessarily terminate in $\bigo(\numfiles)$ iterations. The computational complexity follows because a violation of \eqref{eq:optcond} can be identified in~$\bigo(\numfiles)$ by simulating the reading sequence and evaluating the corresponding times of the first and the second visit for each rewind-stage file.    % Finally, as the solution of~\lafila{} is also considered by~\subexact, it follows that~\subexact preserves the 3-approximation guarantee of~\lafila{}.    
    
    Next, we show the optimality of~\fifila{} for instances where all files are of the same size. Let $\optordertuple$ be an optimal sequence such that there exists at least one file $1 \neq \file \in \forwardphase_{\optordertuple}$ with $\fileweight{\file} = 1$.  By Lemma \ref{lemma:onlyatomic} and Proposition \ref{prop:allforward}, we can assume that $\order^*_{t} < \order^*_{t+1}$ for all $\order^*_{t}, \order^*_{t+1} \in \forwardphase_{\optordertuple}$, and $\order^*_{t} > \order^*_{t+1}$ for all $\order^*_{t},\order^*_{t+1} \in \rewindphase_{\optordertuple}$, $t < \numfiles$. Furthermore, $\fileweight{\file} = 1$  for all $\file \in \rewindphase_{\optordertuple}$, as postponing to the forward stage every file~$\file$ where $\fileweight{\file} = 0$ does not increase the total response time. Let~$\file$ be the forward-stage file $1 \neq \file \in \forwardphase_{\optordertuple}$, $\fileweight{\file} = 1$, with the largest $\leftfile{\file}$. 
Let $\filep \in \rewindphase_{\optordertuple}$ be the rewind-stage file that is closest and on the right of $\file$ in the tape (i.e., $\filep > \file$), and create a new sequence $\ordertuple'$ where $\file$ is read in the rewind stage immediately after $\filep$ (if such a file~$\filep$ does not exist, then set $\ordertuple'_1 = \file$).
%is first in $\ordertuple'$). 
The response time of $\file$ in $\ordertuple'$ reduces by at least $2 \filesize{} (\file - 1)$ with respect to its response time in~$\optordertuple$, which corresponds to the time to reach $0$ from $\leftfile{\file}$ (while eventually reading other files in the rewind stage) and return to $\leftfile{\file}$. The response times of files $1, 2, \dots, \file-1$, however, increase by exactly $2 \filesize{}$ each. Since not all files $1, 2, \dots, \file-1$ have requests, the objective value of $\ordertuple'$ in \ref{model:LTS} is not larger than the objective value of~$\optordertuple$. The result follows from the iterative application of the procedure above.
%does not increase. Moreover, the rewind-stage file ordering is maintained since $\file$ has the largest $\leftfile{\file}$ across all such files, and the argument can be repeated iteratively. 
\hfill $\blacksquare$
\end{proof}

\bigskip

\begin{proof}{Proof of Proposition~\ref{prop:necessary_optimality_condition}.} 

First, observe that~$\neigh{t}$ always exist since~$\order_{t}$ is a rewind-stage file and  all files are traversed in the forward stage. It also follows that
%By using similar arguments to those employed in the proof of Proposition~\ref{prop:rewindblock}, it follows that 
a requested file is serviced the first time the tape head traverses it from left to right in any optimal solution; otherwise, the total response time can only increase. 

The left-hand side expression of inequality~\eqref{eq:optcond} is equal to the amount of time~$\Delta_{t}$ elapsed between the first and the second visit of~$\order_{t}$, i.e.,
\[
\Delta_{t} 
\equiv
% \underbrace{\dist{\order_{t-1}}{\order_{t+1}}}_{\text{Moving from~$\leftfile{t}$ to~$\leftfile{t+1}$}} 
% + 
\underbrace{
\sum\limits_{t'=t}^{ \neigh{t}} 
		\left( \filesize{\order_{t'}} + \dist{\order_{t'}}{\order_{t'+1}} \right)}_{\text{Servicing requests~$\order_{t},\ldots,\order_{\neigh{t} -1}$}}
% 		+
% \underbrace{\filesize{\order_{\neigh{t}}}}_{\text{Servicing~$\order_{\neigh{t}}$}}
% 		+
% \underbrace{\dist{\order_{\neigh{t}}}{\order_{t}}.}_{\text{Moving from~$\rightfile{ \order_{\neigh{t}}}$ to~$\leftfile{t}$}}
\]

The right-hand side of inequality~\eqref{eq:optcond} corresponds to the increase~$\delta_{t}$ in the objective value from the response times of $\order_{t+1},\ldots,\order_{\numfiles}$ that incur when reading~$\order_{t}$, i.e., 
\[
\delta_{t} = \underbrace{2\filesize{\order_{t}}}_{\substack{\text{Reading~$\order_{t}$ and} \\ \text{returning to~$\leftfile{\order_{t}}$}}} 
\cdot 
\underbrace{
\left(
 \numrequests - \sum_{t' = 1}^{t} \fileweight{\order_{t'}}\right) }_{ \substack{\text{Number of requests} \\ \text{in~$\order_{t+1},\ldots,\order_{\numfiles}$} } }.
\]
By postponing~$\order_{t}$ to position~$\neighs{\ordertuple}{t}+1$ in~$\ordertuple$, we obtain a solution~$\ordertuple'$ that increases the response time of~$\order_t$  by~$\Delta_t$ and decreases the response times of~$\order_{t+1},\ldots,\order_{\numfiles}$ by~$\delta_t$. Thus, if inequality~\eqref{eq:optcond}
is violated and~$\Delta_t < \delta_t$, then~$\ordertuple$ is not optimal, as its objective value is greater than that of~$\ordertuple'$. \hfill $\blacksquare$
\end{proof}

\subsection{An Exact Polynomial-Time Algorithm}

% \subsubsection{Proof of Theorem \ref{thm:lts_exact}}

% We introduce additional notation to simplify the presentation of the proofs.  
Let~$\pordertuple = (\porder_1, \porder_2, \dots, \porder_{\blocksize})$ be a subsequence of a solution~$\ordertuple$ to the~\ref{model:LTS}; when applicable, we use~$\porder_0$ and~$\porder_{\blocksize+1}$ to denote the elements of~$\ordertuple$ preceding and succeeding~$\pordertuple$, respectively. 
%  We write $\pendingrequestsk{t}  \equiv k - \sum\limits_{t'=1}^{t} \fileweight{\porder_{t'}}$ to denote the number of pending requests after servicing the first $t$ files of~$\pordertuple$, $t \leq \blocksize$, while assuming that there are~$k$ pending requests by the time~$\pordertuple$ starts to be executed; $\pordertuple$ is always clear from the context, so we omit it from the notation in~$\pendingrequestsk{t}$. Moreover,  we use
% 	\[
% 		\subsequencevalk{l}{u}{\ordertuple} 
% 		\equiv 
% 		\sum_{t=l+1}^{u} 
% 		\pendingrequestsk{t-1}
% 		\left( \filesize{\porder_{t-1}} + \dist{\porder_{t-1}}{\porder_{t}} \right)
% 	\]
% to denote the increase to the response times of all requests incurring due to the execution of~$\porder_{l}, \porder_{l+1}, \ldots, \porder_{u}$, starting from the moment when the tape head reaches~$\leftfile{\porder_l}$. 

\medskip

\begin{proof}{Proof of Proposition \ref{prop:rewindblock}.}
	Suppose the assumptions of the proposition hold but $t > t'$, i.e.,
	$\order^*_{t}$ and~$\order^*_{t'}$ are rewind-stage files, there is at least one forward-stage file~$\file$ such that~$\order^*_{t'} < \file < \order^*_{t}$, and $\order^*_{t}$ succeeds $\order^*_{t'}$ in $\optordertuple$ (i.e., $\order^*_{t'}$ is read first). This implies that the tape head eventually moves from $\order^*_{t'}$ to $\order^*_{t}$ during the rewind stage. 	
	This movement necessarily traverses the forward-stage file $\file$ since $\rightfile{\order^*_{t'}} < \leftfile{\file}< \leftfile{\order^*_{t}}$ by assumption. %That is, there exist two rewind-stage files $\dots,\optordertuple_{\eta},\optordertuple_{\eta+1},\dots$ that are consecutive in~$\optordertuple$  	such that $\rightfile{\optordertuple_{\eta}} < \leftfile{\file}< \leftfile{\optordertuple_{\eta+1}}$, i.e., the tape head traverses $\file$ from its left to its right after reading~$\optordertuple_{\eta}$ and before reading~$\optordertuple_{\eta+1}$ (notice that possibly $\eta=t'$ and $\eta+1=t$). 
	We can therefore modify~$\optordertuple$  to read $\file$ in this movement, i.e., add  $(\order^*_{\eta},\file,\order^*_{\eta+1})$ as a subsequence, to reduce the response time of $\file$ without  impacting the response time of any other file; in the special case where~$\fileweight{\file} = 1$, this adjustment leads to a sequence with lower response time, thus contradicting the optimality of~$\optordertuple$.
	\hfill $\blacksquare$
\end{proof}

\medskip

To demonstrate Theorem \ref{thm:lts_exact}, we will use the concept of  \textit{blocks} defined below:
\begin{definition}[Block]
	\label{def:block}  
    Given a sequence $\ordertuple$ and two files $\file, \filep \in \fileset$, $\file \le \filep$, a subset $\block{\file}{\filep}^{\ordertuple} \equiv \{\file, \file+1,\file+2,\ldots,\filep\} \subseteq \fileset$ of $\filep-\file+1$ adjacent tape files is a block of $\ordertuple$ if 
	\begin{itemize}
		\item[(a)] All files in~$\block{\file}{\filep}^{\ordertuple}$ are rewind-stage files in~$\ordertuple$, i.e., $\block{\file}{\filep}^{\ordertuple} \subseteq \rewindphase_{\ordertuple}$; and  
  		\item[(b)] $\ordertuple$ has a contiguous sequence~$\pordertuple$ which is a permutation of the files  in~$\block{\file}{\filep}^{\ordertuple}$, i.e., $\ordertuple = (\order_1,\order_2,\ldots,\porder_1,\porder_2,\ldots,\porder_{\blocksize},\ldots,\order_{n})$
		for $\{\porder_1, \porder_2, \dots, \porder_b\} = \block{\file}{\filep}^{\ordertuple}$. 
	\end{itemize}
\end{definition}

We refer to the \textit{execution of block $\block{\file}{\filep}^{\ordertuple}$} as the execution of the read operations of~$\pordertuple$. For technical convenience, we assume that the execution of~$\block{\file}{\filep}^{\ordertuple}$ starts when~$\rightfile{\filep}$ is reached for the first time  and ends when~$\leftfile{\file}$ is reached after all the read operations. We write~$\block{\file}{\filep}$ to denote~$\block{\file}{\filep}^{\ordertuple}$ when~$\ordertuple$ is clear from the context. Moreover, we also say that  $\block{\file}{\filep}$ is on the left of~$\block{\file'}{\filep'}$ if $\file < \file'$.

A block~$\block{\file}{\filep}$ is \textit{maximal} if it is neighbored by forward-stage files~$\file-1$ and~$\filep+1$ and it is not properly contained in any other block. A block is \textit{simple} if all its files are read in a single left-to-right movement. Figure \ref{fig:exampleOptSolution} illustrates the block structure with respect to the sequence~$\optordertuple = (7,8,3,2,4,1,5,6,9)$. The solution contains the maximal blocks  $\block{7}{8} = \{7,8\}$ (block 1) and $\block{2}{4} = \{2,3,4\}$ (block 2) induced by the forward-stage files $\forwardphase_{\optordertuple} = \{1,5,6,9\}$. Block 1 is simple because files $7, 8$ are read in a single left-to-right movement. Block 2, in turn, is not simple because file $3$ is read prior to file $2$. However, the non-maximal block $\block{3}{3} = \optordertuple_3 = \{3\}$ is simple. 

It follows from Proposition~\ref{prop:rewindblock} that the forward-stage files in~$\ordertuple$ define maximal blocks. Note also that each rewind-stage file belongs to some maximal block of $\optordertuple$, i.e., the union of all maximal blocks of~$\ordertuple$ spans $\rewindphase_{\ordertuple}$. 

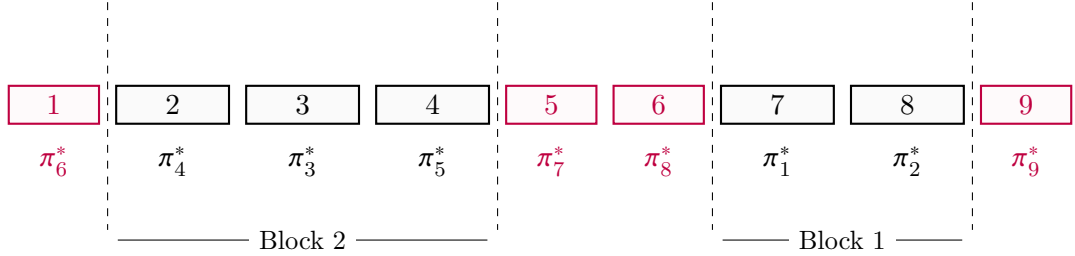
\begin{figure}[ht]
    \centering
	\usetikzlibrary{arrows,backgrounds,snakes}
	\begin{tikzpicture}
		[
			rewindfile/.style={rectangle, draw=black, fill=black!2, thick, minimum size=5mm},
			forwardfile/.style={rectangle, draw=purple, fill=purple!2, thick, minimum size=5mm, text=purple}			
		]
		
		% Files
		\node[forwardfile, minimum width=12mm] (f1) {1};
		\node[rewindfile, minimum width=15mm, right=2mm of f1] (f2) {2};
		\node[rewindfile, minimum width=15mm, right=2mm of f2, node distance=50pt] (f3) {3};
		\node[rewindfile, minimum width=15mm, right=2mm of f3, node distance=50pt] (f4) {4};
		\node[forwardfile, minimum width=12mm, right=2mm of f4, node distance=50pt] (f5) {5};
		\node[forwardfile, minimum width=12mm, right=2mm of f5, node distance=50pt] (f6) {6};
		\node[rewindfile, minimum width=15mm, right=2mm of f6, node distance=50pt] (f7) {7};
		\node[rewindfile, minimum width=15mm, right=2mm of f7, node distance=50pt] (f8) {8};
		\node[forwardfile, minimum width=12mm, right=2mm of f8, node distance=50pt] (f9) {9};
		
		% Labels 
		\node[below = 1.5mm of f7] {$\order^*_1$};
		\node[below = 1.5mm of f8] {$\order^*_2$};
		\node[below = 1.5mm of f3] {$\order^*_3$};
		\node[below = 1.5mm of f2] {$\order^*_4$};
		\node[below = 1.5mm of f4] {$\order^*_5$};
		\node[text=purple, below = 1.5mm of f1] {$\order^*_6$};
		\node[text=purple, below = 1.5mm of f5] {$\order^*_7$};
		\node[text=purple, below = 1.5mm of f6] {$\order^*_8$};
		\node[text=purple, below = 1.5mm of f9] {$\order^*_9$};

		% Blocks dashes
		\node[right=0.5pt of f1, yshift=15mm, xshift=-1.5pt] (l1a) {};
		\node[right=0.5pt of f1, yshift=-18mm, xshift=-1.5pt] (l1b) {};
		\draw[dashed] (l1a) -- (l1b);

		\node[right=0.5pt of f4, yshift=15mm, xshift=-1.5pt] (l2a) {};
		\node[right=0.5pt of f4, yshift=-18mm, xshift=-1.5pt] (l2b) {};
		\draw[dashed] (l2a) -- (l2b);

		\node[right=0.5pt of f6, yshift=15mm, xshift=-1.5pt] (l3a) {};
		\node[right=0.5pt of f6, yshift=-18mm, xshift=-1.5pt] (l3b) {};
		\draw[dashed] (l3a) -- (l3b);

		\node[right=0.5pt of f8, yshift=15mm, xshift=-1.5pt] (l4a) {};
		\node[right=0.5pt of f8, yshift=-18mm, xshift=-1.5pt] (l4b) {};
		\draw[dashed] (l4a) -- (l4b);

		% Block labels
		\path[-,every node/.style={font=\sffamily}]
		(l1b)
			edge node[fill=white!3,font=\small] {Block 2} (l2b)
		(l3b)
			edge node[fill=white!3,font=\small] {Block 1} (l4b)
		;
	\end{tikzpicture}
	\caption{Example depicting the block structure of an optimal sequence $\optordertuple = (7,8,3,2,4,1,5,6,9)$. Forward-stage files are $\forwardphase_{\optordertuple} = \{1, 5, 6, 9\}$ (colored) and rewind-stage files are $\rewindphase_{\optordertuple} = \{2, 3, 4, 7, 8\}$.}
	\label{fig:exampleOptSolution}
\end{figure}

The reading stages within a maximum block are equivalent to the stages of the full problem. For example, files $2$ and $4$ can be perceived as ``forward-stage'' files when processing block 2 in Figure \ref{fig:exampleOptSolution}. Next, we show that one can identify the optimal subsequence for each block recursively, and use the resulting subsequences to compose a full solution to \ref{model:LTS}. Specifically, we solve the subproblem associated with each block as an instance of the~\ref{model:LTS}, using the number of pending requests prior to the block executation as a parameter of the optimization problem. To formalize this, Proposition \ref{prop:allrewind} shows that each block can be sequenced separately if the forward-stage files are known. 

\begin{proposition}
	\label{prop:allrewind}
	Let $\optordertuple$ be an optimal solution to the~\ref{model:LTS} and $\block{\file}{\filep}$ be a block of $\optordertuple$, with %size 
	$\blocksize \equiv |\block{\file}{\filep}|$. Let~$f$ be the number of files read in~$\optordertuple$ when~$\rightfile{\filep}$ is reached for the first time. 
	%Suppose that the first file read in $\block{\file}{\filep}$ is in the $f$-th position of  $\optordertuple$. 
	The subsequence~$\pordertuple^*$ of $\optordertuple$ describing the execution of $\block{\file}{\filep}$
	is the minimizer of the function $\optblock{\block{\file}{\filep}}{k}$ with $k = \numrequests - \sum\limits_{t=1}^{f-1} \fileweight{\order^*_{t}}$, where 
	\begin{align}
		\label{eq:optV}
		\optblock{\block{\file}{\filep}}{k}
		\equiv
		\min_{\pordertuple}
		\left\{
			k \dist{\filep}{\porder_1}
			+
			\sum_{t=2}^{\blocksize} 
				\left(
					k - \sum_{t'=1}^t \fileweight{\porder_{t'}}
				\right)
			\left( 
				\filesize{\porder_{t-1}} + \dist{\porder_{t-1}}{\porder_{t}} 
			\right)
			+
			\left(
				k - \sum_{t=1}^{\blocksize} \fileweight{\porder_{t}}
			\right)
			(\filesize{\porder_{\blocksize}} + \dist{\porder_{\blocksize}}{\file})
		\right\}
	\end{align}
	is the total increase in response time when executing $\block{\file}{\filep}$ optimally with $k$ file requests remaining.
\end{proposition}
\begin{proof}{Proof of Proposition \ref{prop:allrewind}.} Each term in the objective function of Problem~\ref{eq:optV} represents the increase in the response time of the pending tasks in different stages of~$\pordertuple$. We partition the terms in $\optblock{\block{\file}{\filep}}{k}$ into three parts:
	\begin{align*}
		\underbrace{k \dist{\filep}{\porder_1}}_{(a)}
		+
		\underbrace{
			\sum_{t=2}^{\blocksize} 
				\left(
					k - \sum_{t'=1}^t \fileweight{\porder_{t'}}
				\right)
			\left( 
				\filesize{\porder_{t-1}} + \dist{\porder_{t-1}}{\porder_{t}} 
			\right)
		}_{(b)}
		+
		\underbrace{
			\left(
				k - \sum_{t=1}^{\blocksize} \fileweight{\porder_{t}}
			\right)
			(\filesize{\porder_{\blocksize}} + \dist{\porder_{\blocksize}}{\file}).
		}_{(c)}
	\end{align*}
	
	%Recall from the formulation of \eqref{eq:objRewriting} that every bit traversed by the tape head increases the response time by the number of pending requests. 
	The increases in the response time associated with the terms above are: a) moving the tape head from the right of $\filep$ to the left of~$\porder_1$; b) servicing the files in $\block{\file}{\filep}$ except~$\porder_b$; and c) servicing~$\porder_b$ and  moving the tape head from the right of $\porder_{\blocksize}$ to the left of $\file$. Therefore, $\optblock{\block{\file}{\filep}}{k}$ optimizes the trade-off between the response times of requests in~$\block{\file}{\filep}$ and requests read in future movements. Consequently, if an optimal solution uses block~$\block{\file}{\filep}$ and has~$k$ pending tasks when~$\block{\file}{\filep}$ is executed, then sequence~$\pordertuple^*$ given by~$\optblock{\block{\file}{\filep}}{k}$ composes an optimal solution. 
	
	Let~$(\pordertuple_1, \pordertuple_2,\ldots,\pordertuple_q)$ be a decomposition of~$\optordertuple$ into simple blocks $\block{\file_1}{\filep_1}, \block{\file_2}{\filep_2}, \ldots, \block{\file_q}{\filep_q}$, ordered as they appear in~$\optordertuple$,
	%according to the moment they are executed, 
	and let~$k_1, k_2,\ldots,k_q$ be number of pending requests when these blocks start to be executed, respectively. We show that an optimal solution~$\optordertuple$ consists of sequences~$\pordertuple_y$ generated by~$\optblock{\block{\file_y}{\filep_y}}{k_y}$ for each~$y$ in~$\{1,2,\ldots,q\}$ by induction in~$y$. First, if~$y = 1$, we have~$k_1 = m$ (as no file has been read yet), so~$\optblock{\block{\file_1}{\filep_1}}{m}$ identifies a sequence that optimizes the trade-off of response times for all files. Moreover, the execution time of the other blocks is irrelevant to files in~$\optblock{\block{\file_1}{\filep_1}}{m}$. Therefore, 
$\pordertuple_1$ derived from Problem~\ref{eq:optV} composes an optimal solution~$\optordertuple$.

Let us assume that Problem~\ref{eq:optV} delivers the optimal sequences for the first~$y-1$ blocks.  From the induction hypothesis, the problem admits an optimal solution~$\optordertuple$ in which the time when~$\rightfile{\filep_y}$ is reached for the first time coincides with the time~$\pordertuple_y$ starts to be executed. Moreover, the reading times of~$\pordertuple_{y+1}, \pordertuple_{y+2}, \ldots, \pordertuple_{q}$ is
%and the forward stage are
not revelant to the response time of files in~$\block{\file_y}{\filep_y}$, Finally, $\optblock{\block{\file_y}{\filep_y}}{k_y}$ optimizes the trade-off between the response times of pending requests in and out of~$\block{\file_y}{\filep_y}$, so~$\pordertuple_y$ is also part of an optimal solution~$\optordertuple$. \hfill $\blacksquare$
\end{proof}

\bigskip

\begin{proof}{Proof of Theorem~\ref{thm:lts_exact}.}   The state space size for both recursions is $\mathcal{O}(\numfiles^3)$. Since each value function evaluation requires time $\mathcal{O}(\numfiles)$ (with memoization), all the entries of~$\rewindrec(\file,\filep, k)$ and~$\forwardrec(\file,\filep, k)$ can be computed in time~$\bigo(\numfiles^4)$. We show that~$\rewindphase(1, \numfiles, \numrequests)$ gives an exact solution to the~\ref{model:LTS} by showing the correctness of both recursions by induction in the size~$\filep-\file$ of their input sequences. 

The problem is trivial if~$\numfiles = \numrequests = 0$. Thus, let us consider the case~$\numfiles = \numrequests = 1$, which admits only one solution: the tape head moves from~$\rightfile{\file}$ to~$\leftfile{\file}$ and then back. $\rewindrec(\file,\file, k)$ considers exactly this strategy in the second case in~$\eqref{eq:rewindrec}$. As for~$\forwardrec(\file,\file, k)$, there is no~$\filep'$ such that~$\file \leq \filep' < \file$, so the second term in~\eqref{eq:valuerec} vanishes and its unique solution is given by the first and the third terms, which evaluate to the expected total response time.

Assume now that both recursions are correct for sequences of size at most~$q-1$, and pick an arbitrary instance with~$\numfiles = q$ files. For~$\rewindrec(1,\numfiles, k)$, there are two main families of sequences to consider. In the first family, file $\numfiles$ is the last to be read. Sequences with this characteristic are considered by~$\forwardrec(1,\numfiles, k)$ in the second ``min'' case in~\eqref{eq:rewindrec}. Specifically, the sequences in~$\forwardrec(1,\numfiles, k)$ first move the tape head to~$\leftfile{\numfiles}$ and then proceeds by considering all the sequences that read~$1,2,\ldots,\numfiles-1$ before moving back to read~$\numfiles$. Thus, since the two subproblems derived from each~$\filep'$ in~$\{1,2,\ldots,\numfiles-1\}$ in~\eqref{eq:rewindrec} are solved to optimality (by induction), $\forwardrec(1,\numfiles, k)$ yields the optimal solution value whereby~$\numfiles$ is the last file to be read. $\rewindrec(1,\numfiles, k)$ then augment any resulting sequence by moving to~$\leftfile{1}$ without affecting the final results (as all requests have been read).

Next, we consider the second family of sequences, i.e., solutions where~$\numfiles$ is not the last to be read, which are addressed only by the rewind-stage recursion. All solutions of this type are considered in the first ``min'' case in~\eqref{eq:rewindrec}, which selects a file~$\filep'$ to partition the original problem into two. The solutions of the subproblems for each choice of~$\filep'$ are independent
%, and the solutions delivered by the rewind-stage function in both subproblems are 
and optimal (by the induction hypothesis). Thus, the sequence that solves $\forwardrec(1,\numfiles, k)$ is also optimal, thus completing the proof. \hfill $\blacksquare$
\end{proof}

\subsection{The Stochastic~\ref{model:LTS}}

% \begin{proposition}[Scale Invariance]\label{prop:scale}
% An optimal solution to an  instance~$I$ of the~\ref{model:LTS} remains optimal if we multiply all file sizes and file weights by~$\alpha > 0$ and~$\beta > 0$, respectively.
% \end{proposition}
% \begin{lemma}[Scale Invariance]\label{prop:scale}
% An optimal solution to an  instance~$I$ of the~\ref{model:LTS} remains optimal if we multiply all file sizes and file weights by~$\alpha > 0$ and~$\beta > 0$, respectively.
% \end{lemma}
\begin{proof}{Proof of Lemma~\ref{lemma:scale}.} Given an optimal solution~$\optordertuple$ to~$I$ with optimal value~$R^*$, the objective value~$R'$ of~$\optordertuple$ when applied to the scaled instance~$I'$ is 
\begin{eqnarray*}
    R' &=& 
	\sum_{t=1}^{\numfiles} \left( \alpha\numrequests - \sum_{t'=1}^{t-1} \alpha \fileweight{\order_{t'}} \right) \left( \beta \filesize{\order_{t-1}} + \beta\dist{\order_{t-1}}{\order_{t}} \right) \\ &=&
	\sum_{t=1}^{\numfiles} \alpha\left( \numrequests - \sum_{t'=1}^{t-1}  \fileweight{\order_{t'}} \right) \beta \left(  \filesize{\order_{t-1}} + \dist{\order_{t-1}}{\order_{t}} \right) \\
	&=&
	\alpha \beta
	\sum_{t=1}^{\numfiles} \left( \numrequests - \sum_{t'=1}^{t-1}  \fileweight{\order_{t'}} \right) \left(  \filesize{\order_{t-1}} + \dist{\order_{t-1}}{\order_{t}} \right) 
	=
	\alpha \beta R^*.
\end{eqnarray*}
Similar arguments show that any solution to~$I'$ with objective~$R''$ gives a solution~$R''' = \frac{R''}{\alpha\beta}$ to~$I$. Thus, any solution to~$I'$ with objective value~$R'$ strictly smaller than $\alpha \beta R^*$ is feasible to~$I$ and attains a smaller objective value than~$R^*$ in~$I$, a contradiction. \hfill $\blacksquare$
\end{proof}

\subsection{The Online~\ref{model:LTS}}

\begin{proof}{Proof of Proposition~\ref{prop:fifo_not_competitive}.}  Let us consider the family of instances where all files have size 1 and  all requests are released within the first time step and induce the sequence~$\ordertuple = (\numfiles/2, \numfiles/2 + 1, \numfiles/2 - 1, \numfiles/2 + 2, \numfiles/2 - 2, \ldots, \numfiles, 1$) according to the~\fifo{} policy. Observe that the movements of the tape reader defined by~$\ordertuple$ follow  a zig-zag pattern, starting from the middle of the tape. Let us consider the response times of the files~$1, 2, \ldots, \numfiles/2$. For convenience, we define~$v(k)$ to denote the response time of file~$\numfiles/2-k$ for~$k \in \{0,1,\ldots,\numfiles-1\}$; observe that file~$\numfiles - k$ is the~$(k+1)$-th to be read among files in~$0, 1, 2, \ldots, \numfiles/2$. By construction, we have~$v(0) = \numfiles/2$, and for~$k \in \{1,\ldots,\numfiles-1\}$ we have
\begin{eqnarray*}
v(k) 
&=& 
\underbrace{v(k-1)}_{ \substack{ \text{Response time of} \\ \text{file~$\numfiles/2-(k-1)$} }  } 
+
\underbrace{2k}_{\substack{\text{ Time to move from} \\ \text{$\leftfile{\numfiles/2 - (k-1)}$ to~$\rightfile{\numfiles/2 + k}$ }  }}
+
\underbrace{2k+1}_{ \substack{ \text{Time to move from} \\ \text{$\rightfile{\numfiles/2 + k}$ to~$\leftfile{\numfiles/2 - k}$}  }} 
%\\
%&=&
=
v(k-1) + 4k + 1.
\end{eqnarray*}
We show that~$v(k) = 2k^2 + 3k$ for~$k \in \{1,\ldots,\numfiles-1\}$. Direct verification shows that the result holds for~$k = 1$; assuming that the result holds for~$k-1$, we have
\begin{eqnarray*}
v(k) 
    &=& v(k-1) + 4k + 1 
    \\    &=& 
    \underbrace{2(k-1)^2 + 3(k-1)}_{\text{Induction hypothesis}} + 4k + 1 \\
    &=& 2k^2 - 4k + 2 +3k -3 +4k + 1 
    \\
    &=& 
    2k^2 + 3k, 
\end{eqnarray*}
so the identity follows. Therefore, the sum of the response times of files~$1, 2, \ldots, \numfiles/2$ in~$\ordertuple$ is 
\begin{eqnarray*}
\sum_{k = 0}^{\numfiles/2}v(k) 
    &=& 
        \sum_{k = 0}^{\numfiles/2}\left(   \frac{\numfiles}{2} +2k^2 + 3k\right)
    = 
        \bigo(\numfiles^2) + 2\sum_{k = 0}^{\numfiles/2}k^2 
        %\\
    %&=& 
    = 
        \bigo(\numfiles^2) +  2\frac{\frac{\numfiles}{2}(\frac{\numfiles}{2} + 1)(\numfiles+1)}{6} = \bigo(\numfiles^3)
\end{eqnarray*}
Finally, the sum of the response times for  solution~$\ordertuple' \equiv (1,2,\ldots,\numfiles)$ is~$\bigo(n^2)$, so it follows that the~\fifo{} policy is not~$c$-competitive for any constant~$c$.
\hfill $\blacksquare$
\end{proof}

\bigskip

\begin{proof}{Proof of Proposition~\ref{prop:onlinebad}.}
We present an adversarial strategy to derive the result. First, let us assume w.l.o.g. that $\numfiles = 3$, all files have size 1, and that the tape head is in~$\leftfile{2}$ (left position of file 2) at time~$t$. Let~$\mathcal{A}_0$ be the family of policies that move the tape head to the right and reach~$\leftfile{3}$ at time~$t+1$. In this case, it suffices for the adversarial environment to requests file~$2$ at time~$t+1$. In this case, an optimal strategy would only start moving the tape head to the right  at time~$t+1$, thus achieving an adjusted response time of zero, whereas any policy in~$\mathcal{A}_0$ would deliver an adjusted response time of at least 2. Thus, the competitive ratio of any algorithm in~$\mathcal{A}_0$  is unbounded. The same adversarial strategy shows that the competitive ratio of any policy that moves the reader to the left and reaches~$\leftfile{1}$ at time~$t+1$ also has an unbounded competitive ratio. Finally, for policies that maintains the tape reader positioned at~$\leftfile{2}$ until time~$t+1$, it suffices for the adversarial environment to release a request for file~$1$ at time~$t+1$, thus yielding an unbounded competitive ratio. Thus, there is no policy that delivers a bounded competitive ratio for the online~\ref{model:LTS} when considering the minimization of the sum of adjusted response  times.
\hfill $\blacksquare$
\end{proof}

%%%%%%%%%%%%%%%%%
\end{document}